\documentclass[letter,11pt]{article}

\usepackage{mathtools}
\usepackage{dsfont}
\usepackage{changepage}
\usepackage{algo}
\usepackage{algorithm, algorithmic, cite}
\usepackage{float,graphicx,verbatim,fullpage,hyperref,amssymb,amsmath,amsthm,amsfonts,enumerate,multicol, xspace,comment,xcolor,thm-restate,url, cite}
\usepackage[T1]{fontenc}
\usepackage[margin=1in]{geometry}

\usepackage{amssymb,amsmath,amsthm,amsfonts}
\usepackage{subfig}
\usepackage{graphicx,comment,tikz}
\usepackage[margin=1in]{geometry}
\newcommand*{\email}[1]{\texttt{#1}}
\usepackage{hyperref}
\allowdisplaybreaks

\newcommand{\poly}{\text{poly}}
\newcommand{\OPT}{\text{OPT}}
\newcommand{\opt}{\text{OPT}}
\newcommand{\proj}{\text{proj}}
\newcommand{\mmkc}{\textsc{Minmax $k$-cut}\xspace}
\newcommand{\mskc}{\textsc{Minsum $k$-cut}\xspace}
\newcommand{\mlpnormkc}{\textsc{Min $\ell_p$-norm $k$-cut}\xspace}
\newcommand{\mlonenormkc}{\textsc{Min $\ell_1$-norm $k$-cut}\xspace}
\newcommand{\mlinfinitynormkc}{\textsc{Min $\ell_{\infty}$-norm $k$-cut}\xspace}
\newcommand{\multiwaycut}{\textsc{Multiway cut}\xspace}
\newcommand{\minmaxmultiwaycut}{\textsc{Minmax Multiway cut}\xspace}
\newcommand{\lrange}{\{0,1,\ldots,\lambda\}}
\newcommand{\mcd}{\mathcal{D}}
\newcommand{\mcc}{\mathcal{C}}
\newcommand{\mcp}{\mathcal{P}}
\newcommand{\mcf}{\mathcal{F}}
\newcommand{\mcq}{\mathcal{Q}}
\newcommand{\mcr}{\mathcal{R}}
\newcommand{\krange}{\{0,1,\ldots,2k^2\}}
\newcommand{\prange}{\{0,1,\ldots,p\}}
\newcommand{\ct}{\chi(t)}
\newcommand{\tpi}{\Tilde{\Pi}}

\newcommand{\fs}{{$f$-sound}\xspace}
\newcommand{\fsness}{$f$-soundness\xspace}
\newcommand{\fc}{{$f$-correct}\xspace}
\newcommand{\fcness}{$f$-correctness\xspace}
\newcommand{\rs}{{$r$-sound}\xspace}
\newcommand{\rsness}{$r$-soundness\xspace}
\newcommand{\rc}{{$r$-correct}\xspace}
\newcommand{\rcness}{$r$-correctness\xspace}
\newcommand{\gs}{{$g$-sound}\xspace}
\newcommand{\gsness}{$g$-soundness\xspace}
\newcommand{\gc}{{$g$-correct}\xspace}
\newcommand{\gcness}{$g$-correctness\xspace}
\newcommand{\hs}{{$h$-sound}\xspace}
\newcommand{\hsness}{$h$-soundness\xspace}
\newcommand{\hc}{{$h$-correct}\xspace}
\newcommand{\hcness}{$h$-correctness\xspace}
\newcommand{\hps}{{$\mu$-sound}\xspace}
\newcommand{\hpsness}{$\mu$-soundness\xspace}
\newcommand{\hpc}{{$\mu$-correct}\xspace}
\newcommand{\hpcness}{$\mu$-correctness\xspace}
\newcommand{\nus}{{$\nu$-sound}\xspace}
\newcommand{\nusness}{$\nu$-soundness\xspace}
\newcommand{\nuc}{{$\nu$-correct}\xspace}
\newcommand{\nucness}{$\nu$-correctness\xspace}

\newcommand{\R}{\mathbb{R}}
\newcommand{\E}{\mathbb{E}}

\newcommand{\cP}{\mathcal{P}}

\newcommand{\minmax}{minmax\xspace}

\newcommand{\cost}{\text{cost}}

\newtheorem{theorem}{Theorem}[section]
\newtheorem{lemma}{Lemma}[section]

\newtheorem{proposition}{Proposition}[section]
\newtheorem{claim}{Claim}[section]
\newtheorem{observation}{Observation}[section]

\newtheorem{definition}{Definition}[section]

\def\final{0}  
\def\iflong{\iffalse}
\ifnum\final=0  
\newcommand{\knote}[1]{{\color{red}[{\tiny Karthik: \bf #1}]\marginpar{\color{red}*}}}
\newcommand{\wnote}[1]{{\color{blue}[{\tiny Weihang: \bf #1}]\marginpar{\color{blue}*}}}
\newcommand{\todonote}[1]{{\color{red}[{\tiny TODO: \bf #1}]\marginpar{\color{red}*}}}
\else 
\newcommand{\knote}[1]{}
\newcommand{\wnote}[1]{}
\newcommand{\todonote}[1]{}
\fi  

\def\set#1{\left\{ #1 \right\}}

\title{Fixed Parameter Approximation Scheme for  Min-max $k$-cut\thanks{University of Illinois, Urbana-Champaign, Email: \email{\{karthe, weihang3\}@illinois.edu}. Supported in part by NSF grants CCF-1814613 and CCF-1907937.}
}
\author{
Karthekeyan Chandrasekaran
\and 
Weihang Wang
}
\date{}
\begin{document}

\maketitle

\begin{abstract}
    We consider the graph $k$-partitioning problem under the min-max objective, termed as \mmkc. The input here is a graph $G=(V,E)$ with non-negative edge weights $w:E\rightarrow \R_+$ and an integer $k\ge 2$ and the goal is to partition the vertices into $k$ non-empty parts $V_1, \ldots, V_k$ so as to minimize $\max_{i=1}^k w(\delta(V_i))$. 
    Although minimizing the sum objective $\sum_{i=1}^k w(\delta(V_i))$, termed as \mskc, has been studied extensively in the literature, very little is known about minimizing the max objective. We initiate the study of \mmkc by showing that it is NP-hard and W[1]-hard when parameterized by $k$, and design a parameterized approximation scheme when parameterized by $k$. 
    The main ingredient of our parameterized approximation scheme is an exact algorithm for \mmkc that runs in time $(\lambda k)^{O(k^2)}n^{O(1)}$, where $\lambda$ is value of the optimum and $n$ is the number of vertices. Our algorithmic technique builds on the technique of Lokshtanov, Saurabh, and Surianarayanan \cite{LSS} who showed a similar result for \mskc. 
    Our algorithmic techniques are more general and can be used to obtain parameterized approximation schemes for minimizing $\ell_p$-norm measures of $k$-partitioning for every $p\ge 1$. 
\end{abstract}


\newpage
\setcounter{page}{1}

\section{Introduction}\label{sec:intro}
Graph partitioning problems are fundamental for their intrinsic theoretical value as well as applications in clustering. 
In this work, we consider graph partitioning under the \emph{minmax} objective. The input here is a graph $G=(V,E)$ with non-negative edge weights $w:E\rightarrow \R_+$ along with an integer $k\ge 2$ and the goal is to partition the vertices of $G$ into $k$ non-empty parts $V_1, \ldots, V_k$ so as to minimize $\max_{i=1}^k w(\delta(V_i))$; here, $\delta(V_i)$ is the set of edges which have exactly one end-vertex in $V_i$ and $w(\delta(V_i)):=\sum_{e\in \delta(V_i)}w(e)$ is the total weight of the edges in $\delta(V_i)$. We refer to this problem as \mmkc. 


\medskip
\noindent \emph{Motivations.}
Minmax objective for optimization problems has an extensive literature in approximation algorithms. It is relevant in scenarios where the goal is to achieve fairness/balance---e.g., load balancing in multiprocessor scheduling, discrepancy minimization, min-degree spanning tree, etc. 
In the context of graph cuts and partitioning, recent works (e.g., see \cite{CGS17, AKS19, KMZ19}) have proposed and studied alternative minmax objectives that are different from \mmkc. 

The complexity of \mmkc was also raised as an open problem by Lawler \cite{La73}. Given a partition $V_1, \ldots, V_k$ of the vertex set of an input graph, one can measure the quality of the partition in various natural ways. Two natural measures are (i) the max objective given by $\max_{i=1}^k w(\delta(V_i))$ and (ii) the sum objective given by $\sum_{i=1}^k w(\delta(V_i))$. We will discuss other \emph{$\ell_p$-norm measures} later. Once a measure is defined, a corresponding optimization problem involves finding a partition that minimizes the measure. 
We will denote the optimization problem where the goal is to minimize the sum objective as \mskc. 

\medskip
\noindent \emph{\mskc and prior works.}
For $k=2$, the objectives in \mmkc and \mskc coincide owing to the symmetric nature of the graph cut function (i.e., $w(\delta(S))=w(\delta(V\setminus S))$ for all $S\subseteq V$) but the objectives differ for $k\ge 3$. \mskc has been studied extensively in the algorithms community leading to fundamental graph structural results. 
We briefly recall the literature on \mskc. 

Goldschmidt and Hochbaum \cite{GH88, GH94} showed that \mskc is NP-hard when $k$ is part of input by a reduction from \textsc{CLIQUE} and designed the first polynomial time algorithm for fixed $k$. Their algorithm runs in time $n^{O(k^2)}$, where $n$ is the number of vertices in the input graph. Subsequently, Karger and Stein \cite{KS96} gave a random contraction based algorithm that runs in time $\tilde{O}(n^{2k-2})$. Thorup \cite{Th08} gave a tree-packing based deterministic algorithm that runs in time $\tilde{O}(n^{2k})$. 
The last couple of years has seen renewed interests in \mskc with  
exciting progress \cite{Ma18, GLL18-SODA, GLL18-FOCS, CQX19, Li19, GLL20-STOC, GHLL20, CKLPPSW18}. 
Very recently, Gupta, Harris, Lee, and Li \cite{GLL20-STOC, GHLL20} have shown that the Karger-Stein algorithm in fact runs in $\tilde{O}(n^k)$ time; $n^{(1-o(1))k}$ seems to be a lower bound on the run-time of any algorithm \cite{Li19}. 
The hardness result of Goldschmidt and Hochbaum as well as their algorithm inspired Saran and Vazirani \cite{SV95} to consider \mskc when $k$ is part of input from the perspective of approximation. They showed the first polynomial-time $2$-approximation for \mskc. Alternative $2$-approximations have also been designed subsequently \cite{RS08, ZNI05}. For $k$ being a part of the input, Manurangsi\cite{Ma18} showed that there does not exist a polynomial-time  $(2-\epsilon)$-approximation for any constant $\epsilon>0$ under the Small Set Expansion Hypothesis. 

\mskc has also been investigated from the perspective of fixed-parameter algorithms. 
It is known that \mskc when parameterized by $k$ is W[1]-hard and does not admit a $f(k)n^{o(1)}$-time algorithm for any function $f(k)$ \cite{DEFPR03, FPT-book}. 
Motivated by this hardness result and Manurangsi's $(2-\epsilon)$-inapproximability result, 
Gupta, Lee, and Li \cite{GLL18-SODA} raised the question of whether there exists a \emph{parameterized approximation algorithm} for \mskc when parameterized by $k$, i.e., can one obtain a $(2-\epsilon)$-approximation in time $f(k)n^{O(1)}$ for some constant $\epsilon>0$? As a proof of concept, they designed a $1.9997$-approximation algorithm that runs in time $2^{O(k^6)}n^{O(1)}$ \cite{GLL18-SODA} and a $(1+\epsilon)$-approximation algorithm that runs in time $(k/\epsilon)^{O(k)}n^{k+O(1)}$ \cite{GLL18-FOCS}. Subsequently, Kawarabayashi and Lin \cite{KL20} designed a $(5/3+\epsilon)$-approximation algorithm that runs in time $2^{O(k^2 \log k)}n^{O(1)}$. 
This line of work culminated in a \emph{parameterized approximation scheme} when parameterized by $k$---Lokshtanov, Saurabh, and Surianarayanan \cite{LSS} designed a $(1+\epsilon)$-approximation algorithm that runs in time $(k/\epsilon)^{O(k)}n^{O(1)}$. We emphasize that, from the perspective of algorithm design, a parameterized approximation scheme is more powerful than a parameterized approximation algorithm. 

\medskip
\noindent \emph{Fixed-terminal variants.}
A natural approach to solve both \mmkc and \mskc is to solve their fixed-terminal variants: 
The input here is a graph $G=(V,E)$ with non-negative edge costs $w:E\rightarrow \R_+$ along with $k$ terminals $v_1, \ldots, v_k\in V$ and the goal is to partition the vertices into $k$ parts $V_1, \ldots, V_k$ such that $v_i\in V_i$ for every $i\in [k]$ so as to minimize the measure of interest for the partition. 
The fixed-terminal variant of \mskc, popularly known as \multiwaycut, is NP-hard for $k\ge 3$ \cite{DJPSY94} and has a rich literature. It admits a $1.2965$ approximation \cite{SV14} and does not admit a $(1.20016-\epsilon)$-approximation for any constant $\epsilon>0$ under the unique games conjecture \cite{BCKM20}. 
The fixed-terminal variant of \mmkc, known as \minmaxmultiwaycut, is NP-hard for $k\ge 4$ \cite{svitkina-tardos} and admits an $O(\sqrt{\log n \log k})$-approximation \cite{BFKMNNS14}. 
Although fixed-terminal variants are natural approaches to solve global cut problems 
(similar to using min $\set{s,t}$-cut to solve global min-cut), they have two limitations: (1) they are not helpful when $k$ is part of input and (2) even for fixed $k$, they do not give the best algorithms (e.g., even for $k=3$, \multiwaycut is NP-hard while \mskc is solvable in polynomial time as discussed above). 

\medskip
\noindent \emph{\mmkc vs \mskc}. There is a fundamental structural difference between \mmkc and \mskc. The optimal solution to \mskc satisfies a nice property: 
assuming that the input graph is connected, every part in an optimal partition for \mskc induces a connected subgraph. 
Hence, \mskc is also phrased as the problem of deleting a subset of edges with minimum weight so that the resulting graph contains at least $k$ connected components. However, this nice property does not hold for \mmkc as illustrated by the example in Figure \ref{fig:disconnected-part}. 

\begin{figure}[ht]
    \centering
    \includegraphics[width=0.6\textwidth]{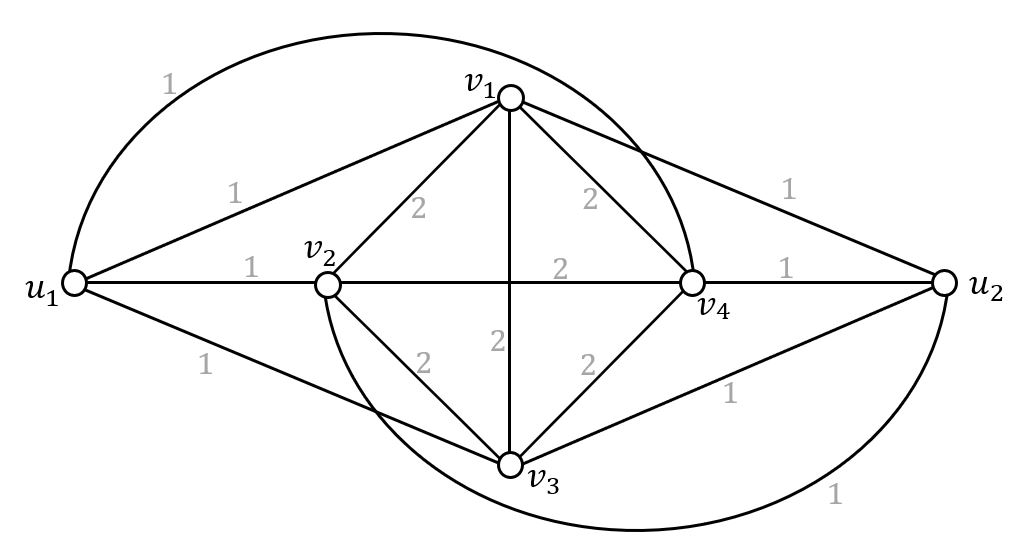}
    \caption{An example where the unique optimum partition for \mmkc for $k=5$ induces a disconnected part. The edge weights are as shown. Every $5$-partition in this example necessarily consists of one part with $2$ vertices and four singleton parts. If the part with $2$ vertices is $\{u_1,u_2\}$, then the objective value is $8$. If the part with 2 vertices is $\{u_i,v_j\}$ where $i\in[2]$ and $j\in[4]$, then the objective value is $10$. If the part with $2$ vertices is $\{v_i,v_j\}$ where $i,j\in[4]$, then the objective value is $12$. Hence the optimum partition for minmax $5$-cut is $(\{u_1,u_2\},\{v_1\},\{v_2\},\{v_3\},\{v_4\})$, where the first part induces a disconnected subgraph.
    }
    \label{fig:disconnected-part}
\end{figure}

\medskip
\noindent \emph{\mmkc for fixed $k$.} 
For fixed $k$, there is an easy approach to solve \mmkc based on the following observation: For a given instance, an optimum solution to \mmkc is a $k$-approximate optimum to \mskc. The randomized algorithm of Karger and Stein implies that the number of $k$-approximate solutions to \mskc is $n^{O(k^2)}$ and they can all be enumerated in polynomial time \cite{KS96, GLL20-STOC, GHLL20} (also see \cite{CQX19}). These two facts immediately imply that \mmkc can be solved in $n^{O(k^2)}$ time. 
We recall that the graph cut function is symmetric and submodular.\footnote{A function $f:2^V\rightarrow \R$ is symmetric if $f(S)= f(V\setminus S)$ for all $S\subseteq V$ and is submodular if $f(A)+f(B)\ge f(A\cap B) + f(A\cup B)$.}
In an upcoming work, Chandrasekaran and Chekuri \cite{CC21} show that the more general problem of min-max symmetric submodular $k$-partition\footnote{In the min-max symmetric submodular $k$-partition problem, the input is a symmetric submodular function $f:2^V\rightarrow \R$ given by an evaluation oracle, and the goal is to partition the ground set $V$ into $k$ non-empty parts $V_1, \ldots, V_k$ so as to minimize $\max_{i=1}^k f(V_i)$.} is also solvable in time $n^{O(k^2)}T$, where $n$ is the size of the ground set and $T$ is the time to evaluate the input submodular function on a given set. 


\subsection{Results}\label{sec:results}

In this work, we focus on \mmkc when $k$ is part of input. 
We first show that \mmkc is strongly NP-hard. Our reduction also implies that it is W[1]-hard when parameterized by $k$, i.e., there does not exist a $f(k)n^{O(1)}$-time algorithm for any function $f(k)$. 
\begin{restatable}{theorem}{thmHardness}
\label{thm:hardness}
\mmkc is strongly NP-hard and W[1]-hard when parameterized by $k$. 
\end{restatable}

Our hardness reduction also implies that \mmkc does not admit an algorithm that runs in time $n^{o(k)}$ assuming the exponential time hypothesis. Given the hardness result, it is natural to consider approximations and fixed-parameter tractability. Using the known $2$-approximation for \mskc and the observation that the optimum value of \mskc is at most $k$ times the optimum value of \mmkc, it is easy to get a $(2k)$-approximation for \mmkc. An interesting open  question is whether we can improve the approximability. 

The hardness results also raise the question of whether \mmkc 
admits a parameterized approximation algorithm when parameterized by $k$ or, going a step further, does it admit a parameterized approximation scheme when parameterized by $k$? We resolve this question affirmatively by designing a parameterized approximation scheme. Let $G=(V,E)$ be a graph with non-negative edge weights $w:E\rightarrow \R_+$. We write $G$ to denote the unit-cost version  of the graph (i.e., the unweighted graph) and $G_w$ to denote the graph with edge weights $w$. We emphasize that the unweighted graph could have parallel edges. 
For a partition $(V_1, \ldots, V_k)$ of $V$, we define 
\[
\cost_{G_w}(V_1, \ldots, V_k):=\max\set{w(\delta(V_i)): i\in [k]}.
\]
We will denote the minimum cost of a $k$-partition in $G_w$ by $\OPT(G_w,k)$. 
The following is our algorithmic result showing that \mmkc admits a parameterized approximation scheme when parameterized by $k$.

\begin{restatable}{theorem}{thmFPAS}
\label{thm:FPAS}
There exists a randomized algorithm that takes as input 
an instance of \mmkc, namely an $n$-vertex graph $G=(V,E)$ with edge weights $w:E\rightarrow \R_{\ge 0}$ and an integer $k\ge 2$, 
along with an $\epsilon\in (0,1)$, 
and runs in time $(k/\epsilon)^{O(k^2)}n^{O(1)}\log(\max_{e\in E}w(e))$
to return a partition $\mcp$ of the vertices of $G$ such that $\cost_{G_w}(\mcp)\le (1+\epsilon)\opt(G_w, k)$ with high probability. 
\end{restatable}

We note that $\log(\max_{e\in E}w(e))$ is polynomial in the size of the input. 
Theorem \ref{thm:FPAS} can be viewed as the counterpart of the parameterized approximation scheme for \mskc due to Lokshtanov, Saurabh, and Surianarayanan \cite{LSS} but for \mmkc. The central component of our parameterized-approximation scheme given in Theorem \ref{thm:FPAS} is the following result which shows a fixed-parameter algorithm for \mmkc in unweighted graphs when parameterized by $k$ and the solution size. 

\begin{restatable}{theorem}{thmDP}
\label{thm dp}
There exists an algorithm that takes as input an unweighted instance of \mmkc, namely an $n$-vertex graph $G=(V,E)$ and an integer $k\ge 2$, along with an integer $\lambda$, 
and runs in time $(k\lambda)^{O(k^2)}n^{O(1)}$
to determine if there exists a $k$-partition $(V_1, \ldots, V_k)$ of $V$ such that $\cost_G(V_1,\ldots, V_k)\le \lambda$ and if so, then finds an optimum. 
\end{restatable}
We emphasize that the algorithm in Theorem \ref{thm dp} is deterministic. 

\subsection{Outline of techniques}
Our NP-hardness and W[1]-hardness results for \mmkc are based on a reduction from the clique problem. Our reduction is an adaptation of the reduction from the clique problem to \mskc due to Downey et al \cite{DEFPR03}.

Our randomized algorithm for Theorem \ref{thm:FPAS} essentially reduces the input weighted instance of \mmkc to an instance where Theorem \ref{thm dp} can be applied: we reduce the instance to an unweighted instance with optimum value $O((k/\epsilon ^3))\log n$, i.e., the optimum value is logarithmic in the number of vertices. Moreover, the reduction runs in time $2^{O(k)}(n/\epsilon)^{O(1)}\log{\opt(G_w,k)}$. 
Applying Theorem \ref{thm dp} to the reduced instance yields a run-time of
\begin{align*}
    \left(\left(\frac{k^2}{\epsilon^3}\right)\log{n}\right)^{O(k^2)}n^{O(1)}
    = \left(\frac{k}{\epsilon}\right)^{O(k^2)}(\log{n})^{O(k^2)}n^{O(1)}
    &= \left(\frac{k}{\epsilon}\right)^{O(k^2)}(k^{O(k^2)}+n)n^{O(1)}\\
    &= \left(\frac{k}{\epsilon}\right)^{O(k^2)}n^{O(1)}.
\end{align*}
Hence, the total run-time (including the reduction time) is $(k/\epsilon)^{O(k^2)}n^{O(1)}\log{\opt(G_w, k)}$, thereby proving Theorem \ref{thm:FPAS}. 

We now briefly describe the reduction to an unweighted instance with logarithmic optimum: 
(i) Firstly, we do a standard knapsack PTAS-style rounding procedure to convert the instance to an unweighted instance with a $(1+\epsilon)$-factor loss. (ii) Secondly, we delete cuts with small value to ensure that all connected components in the graph have large min-cut value, i.e., have min-cut value at least $\epsilon \opt/k$---this deletion procedure can remove at most $\epsilon \opt$ edges and hence, a $(1+\epsilon)$-approximate solution in the resulting graph gives a $(1+O(\epsilon))$-approximate solution in the original graph. (iii) Finally, we do a random sampling of edges with probability $p:=\Theta(k \log n/(\epsilon ^3 \opt))$. This gives a subgraph that preserves all cut values within a $(1\pm \epsilon)$-factor when scaled by $p$ with high probability. The preservation of all cut values also implies that the optimum value to \mmkc is also preserved within a $(1 \pm \epsilon)$-factor. The scaling factor of $p$ allows us to conclude that the optimum in the subsampled graph is $O((k/\epsilon ^3))\log n$. We note that this three step reduction follows the same ideas as that of \cite{LSS} who designed a parameterized approximation scheme for \mskc. Our contribution to the reduction is simply showing that their reduction ideas also apply to \mmkc (see Section \ref{sec:reduction-to-unweighted-instances} for details). 

The main contribution of our work is in proving Theorem \ref{thm dp}, i.e., giving a fixed-parameter algorithm for \mmkc when parameterized by $k$ and the solution size. We discuss this now. At a high-level, we exploit the tools developed by \cite{LSS} who designed a dynamic program based fixed-parameter algorithm for \mskc when parameterized by $k$ and the solution size. Our algorithm for \mmkc is also based on a dynamic program. 
However, since we are interested in \mmkc, the subproblems in our dynamic program are completely different from that of \cite{LSS}. 
We begin with the observation that an optimum solution to \mmkc is a $k$-approximate optimum to \mskc. This observation and the tree packing approach for \mskc due to \cite{CQX19} allows us to obtain, in polynomial time, a spanning tree $T$ of the input graph such that the number of edges of the tree crossing a \mmkc optimum partition is $O(k^2)$. We will call a partition $\Pi$ with $O(k^2)$ edges of the tree $T$ crossing $\Pi$ to be a $T$-feasible partition. 
Next, we use the tools of \cite{LSS} to generate, in polynomial time, a suitable tree decomposition of the input graph---let us call this a \emph{good} tree decomposition. 
The central intuition underlying our algorithm is to use the spanning tree $T$ to guide a dynamic program on the good tree decomposition. 

As mentioned before, our dynamic program is different from that of \cite{LSS}. We now sketch the details of our dynamic program. For simplicity, we assume that we have a value $\lambda \ge \opt(G,k)$. 
The \emph{adhesion} of a tree node $t$ in the tree decomposition, denoted $A_t$, is the intersection of the bag corresponding to $t$ with that of its parent (the adhesion of the root node of the tree decomposition is the empty set). The good tree decomposition that we generate has low adhesion, i.e., the adhesion size is $O(\lambda k)$ for every tree node. In order to define our sub-problems for a tree node $t$, we consider all possible partitions $\mathcal{F}^{A_t}$ of the adhesion $A_t$ containing at most $k$ parts and which can be extended to a $T$-feasible partition of the entire vertex set. A simple counting argument shows that $|\mathcal{F}_{A_t}|=(\lambda k)^{O(k^2)}$ (see Lemma \ref{lemma F family size}). Now consider a Boolean function $f_t:\mathcal{F}^{A_t}\times \set{0,1, \ldots, \lambda}^k \times \set{0,1, \ldots, 2k^2}\rightarrow \set{0,1}$. We note that the domain of the function is small, i.e., $(\lambda k)^{O(k^2)}$. 
Let $(\mcp _{A_t}, \Bar{x}, d)$ denote an argument to the function. The function aims to determine if there exists a partition $\mcp$ of the union of the bags descending from $t$ in the tree decomposition (call this set of vertices to be $V_t$)
so that (i) the projection of the partition  $\mcp$ to $A_t$ is exactly $\mcp_{A_t}$, (ii) the number of edges crossing the $i$'th part of $\mcp$ in the subgraph $G[V_t]$ 
is exactly $x_i$ for all $i\in [k]$, and (iii) the number of tree edges crossing the partition $\mcp$ is at most $d$. It is easy to see that if we can compute such a function $f_r$ for the root node $r$ of the tree decomposition, then it can be used to find the optimum value of \mmkc, namely $\opt(G,k)$. 

However, we are unable to solve the sub-problem (i.e., compute such a function $f_t$) based on the sub-problem values of the children of $t$. 
We observe that instead of solving this sub-problem exactly, a weaker goal of finding a function that satisfies a certain \fc and \fs properties suffices (see Definition \ref{defn:fc-fs} for these properties and Proposition \ref{prop:OPT from fr function}). We show that this weaker goal of computing an \fc and \fs function $f_t$ based on \fc and \fs functions $f_{t'}$ for all children $t'$ of $t$ can be achieved in time $(\lambda k)^{O(k^2)}n^{O(1)}$ (see Lemma \ref{lemma:children f to parent f}). Since the domain of the function is of size $(\lambda k)^{O(k^2)}$ and the tree decomposition is polynomial in the size of the input, the total number of sub-problems that we solve in the dynamic program is $(\lambda k)^{O(k^2)}n^{O(1)}$, thus proving Theorem \ref{thm dp}. 

In order to achieve the weaker goal of computing a function $f_t$ for the tree node $t$ that is \fc and \fs, we progressively define sub-problems and note that it suffices to achieve a weaker goal for all these sub-problems. Consequently, our goal reduces to computing Boolean functions that satisfy certain weaker properties. We encourage the reader to trace towards the base case of the dynamic program during the first read of the dynamic program. 

One of the advantages of our dynamic program (in contrast to that of \cite{LSS}) is that it is also applicable for alternative norm-based measures of $k$-partitions: here, the goal is to find a $k$-partition of the vertex set of the given edge-weighted graph so as to minimize $(\sum_{i=1}^k w(\delta(V_i))^p)^{1/p}$---we call this as \mlpnormkc. We note that \mmkc is exactly \mlinfinitynormkc while \mskc is exactly \mlonenormkc.
Our dynamic program can also be used to obtain the counterpart of Theorem \ref{thm dp} for \mlpnormkc for every $p\ge 1$. This result in conjunction with the reduction to unweighted instances (which can be shown to hold for \mlpnormkc) also leads to a parameterized approximation scheme for \mlpnormkc for every $p\ge 1$.

\paragraph{Organization.} 
We set up the tools to prove Theorem \ref{thm dp} in Section \ref{sec:prelims}. We prove Theorem \ref{thm dp} in Section \ref{sec:DP}. 
We show a reduction from weighted instances to unweighted instances with logarithmic optimum value in Section \ref{sec:reduction-to-unweighted-instances}. 
We use Theorem \ref{thm dp} and the reduction to unweighted instances with logarithmic optimum value to prove Theorem \ref{thm:FPAS} in Section \ref{sec:FPAS}.
We prove the hardness results mentioned in Theorem \ref{thm:hardness} in Section \ref{sec:hardness}. 
We conclude with a few open questions in Section \ref{sec:conclusion}.

\section{Tools for the fixed-parameter algorithm}\label{sec:prelims}

In this section, we set up the background for the fixed-parameter algorithm of Theorem \ref{thm dp}. 
Let $G=(V,E)$ be a graph. 
Throughout this work, we consider a partition to be an ordered tuple of non-empty subsets. 
An ordered tuple of subsets $(S_1,\ldots,S_k)$, where $S_i\subseteq V$ for all $i\in[k]$, is a $k$-subpartition of $V$ if $S_1\cup\ldots\cup S_k=V$ and $S_i\cap S_j=\emptyset$ for every pair of distinct $i,j\in[k]$.
We emphasize the distinction between partitions and $k$-subpartitions---in a partition, all parts are required to be non-empty but the number of parts can be fewer than $k$ while a $k$-subpartition allows for empty parts but the number of parts is exactly $k$.

For a subgraph $H\subseteq G$, a subset $X\subseteq V$, and a partition/$k$-subpartition $\mcp$ of $X$, we use $\delta_H(\mcp)$ to denote the set of edges in $E(H)$ whose end-vertices are in different parts of $\mcp$.
For a subgraph $H$ of $G$ and a subset $S\subseteq V(H)$, we use $\delta_{H}(S)$ to denote the set of edges in $H$ with exactly one end-vertex in $S$. 
We will denote the set of (exclusive) neighbors of a subset $S$ of vertices in the graph $G$ by $N_G(S)$.
We need the notion of a tree decomposition.

\begin{definition}
Let $G=(V,E)$ be a graph.
A pair $(\tau,\chi)$, where $\tau$ is a tree and $\chi:V(\tau)\rightarrow 2^V$ is a mapping of the nodes of the tree to subsets of vertices of the graph, is a \emph{tree decomposition} of $G$ if the following conditions hold:
\begin{enumerate}[(i)]
    \item $\cup_{t\in V(\tau)}\chi(t)=V$,
    \item for every edge $e=uv\in E$, there exists some $t\in V(\tau)$ such that $u,v\in\chi(t)$, and
    \item for every $v\in V$, the set of nodes $\{t\in V(\tau):v\in\chi(t)\}$ induces a connected subtree of $\tau$.
\end{enumerate}
For each $t\in V(\tau)$, we call $\chi(t)$ to be a \emph{bag} of the tree decomposition.
\end{definition}

We now describe certain notations that will be helpful while working with the tree decomposition. Let $(\tau, \chi)$ be the tree decomposition of the graph $G=(V,E)$.
We root $\tau$ at an arbitrary node $r\in V(\tau)$. 
For a tree node $t\in V(\tau)\backslash\{r\}$, there is a unique edge between $t$ and its parent. 
Removing this edge disconnects $\tau$ into two subtrees $\tau_1$ and $\tau_2$, and we say that the set $A_t:=\chi(\tau_1)\cap\chi(\tau_2)$ is the \emph{adhesion associated with $t$}. 
For the root node $r$, we define $A_r:=\emptyset$.
For a tree node $t\in V(\tau)$, we denote the subgraph induced by all vertices in bags descending from $t$ as $G_t$ (here, the node  $t$ is considered to be a descendant of itself), i.e., 
\[G_t:=G\left[\bigcup_{t'\text{ is a descendant of }t}\chi(t')\right].\]
We need the notions of compactness and edge-unbreakability.

\begin{definition}
A tree decomposition $(\tau,\chi)$ of a graph $G$ is \emph{compact} if for every tree node $t\in V(\tau)$, the set of vertices $V(G_t)\backslash A_t$ induces a connected subgraph in $G$ and $N_G(V(G_t)\backslash A_t)=A_t$. 
\end{definition}

\begin{definition}
Let $G=(V,E)$ be a graph and let $S\subseteq V$. The subset $S$ is \emph{$(a,b)$-edge-unbreakable} if for every nonempty proper subset $S'$ of $V$ satisfying $|E[S',V\backslash S']|\leq b$, we have that either $|S\cap S'|\leq a$ or $|S\backslash S'|\leq a$.
\end{definition}
Informally, a subset $S$ is $(a,b)$-edge-unbreakable if every non-trivial $2$-partition of $G[S]$ either has large cut value or one side of the partition is small in size. With these definitions, we have the following result from \cite{LSS}.

\begin{lemma}\label{lemma tree decomposition} \cite{LSS}
There exists a polynomial time algorithm that takes a graph $G=(V,E)$, an integer $k\ge 2$, and an integer $\lambda$ as input and returns a compact tree decomposition $(\tau,\chi)$ of $G$ such that
\begin{enumerate}[(i)]
    \item each adhesion has size at most $\lambda k$, and
    \item for every tree node $t\in V(\tau)$, the bag $\chi(t)$ is $((\lambda k+1)^5,\lambda k)$-edge-unbreakable.
\end{enumerate}
\end{lemma}
We observe that since Lemma \ref{lemma tree decomposition} runs in polynomial time, the size of $\tau$ is necessarily $\poly(n,k,\lambda)$, where $n$ is the number of vertices in the input graph.
Next, we need the notion of $\alpha$-respecting partitions.

\begin{definition}
 Let $G=(V,E)$ be a graph and $G'$ be a subgraph of $G$. A partition $\mathcal{P}$ of $V$ \emph{$\alpha$-respects} $G'$ if $|\delta_{G'}(\mathcal{P})|\leq \alpha$.
\end{definition}

The following lemma shows that we can efficiently find a spanning tree $T$ of a given graph such that there exists an optimum $k$-partition that $2k^2$-respects $T$. It follows from Lemma 7 of \cite{CQX19} and the observation that an optimum solution to \mmkc is a $k$-approximate optimum to \mskc. 
\begin{lemma}\label{lemma thorup tree} \cite{CQX19}
There exists a polynomial time algorithm that takes a graph $G$ as input and returns a spanning tree $T$ of $G$ such that there exists an optimum min-max $k$-partition $\Pi$ that $(2k^2)$-respects $T$.
\end{lemma}

We will frequently work with refinements and coarsenings of partitions and also restrictions of partitions to subsets. 

\begin{definition}
Let $G=(V,E)$ be a graph and $S\subseteq V$ be a subset of vertices. 
\begin{enumerate}
\item Let $\mcq$ be a partition/$k$-subpartition of $S$. A partition/$k$-subpartition $\mcp$ of $S$ \emph{coarsens} $\mcq$ if each part of $\mcp$ is a union of parts of $\mcq$.

\item Let $\mcp$ be a partition/$k$-subpartition of $S$. A partition/$k$-subpartition $\mcq$ of $S$ \emph{refines} $\mcp$ if each part of $\mcp$ is a union of parts of $\mcq$.

\item Let $\mcp$ be a partition/$k$-subpartition of $V$. A partition/$k$-subpartition $\mcp'$ of $S$ is a \emph{restriction of $\mcp$ to $S$} if for every $u,v\in S$, $u$ and $v$ are in the same part of $\mcp'$ if and only if they are in the same part of $\mcp$. 
\end{enumerate}
\end{definition}
We note that a \emph{restriction of a partition/$k$-subpartition $\mcp$ to a subset $S$} is a reordering of the tuple obtained by taking the intersection of each part in $\mcp$ with $S$.
The following definition allows us to handle partitions of subsets that are crossed by a spanning tree at most $2k^2$ times.
\begin{definition}

Let $G=(V,E)$ be a graph, $T$ be a spanning tree of $G$, and $X\subseteq V$. A partition $\cP$ of $X$ is \emph{$T$-feasible} if there exists a partition $\cP'$ of $V$ such that
\begin{enumerate}[(i)]
    \item The restriction of $\cP'$ to $X$ is $\cP$, 
    and
    \item $\mcp'$ $(2k^2)$-respects $T$.
\end{enumerate}
Moreover, a $k$-subpartition $\mcp'$ of $X$ is $T$-feasible if the partition obtained from $\mcp'$ by discarding the empty parts of $\mcp'$ is $T$-feasible.
\end{definition}

The next definition and the subsequent lemmas will show a convenient way to work with $T$-respecting partitions of a subset $X$ of vertices, where $T$ is a spanning tree.
\begin{definition}
Let $G=(V,E)$ be a graph, $T$ be a spanning tree of $G$, and $X\subseteq V$. The graph $\proj(T,X)$ is the tree obtained from $T$ by
\begin{enumerate}
    \item repeatedly removing leaves of $T$ that are not in $X$ until there is none, and
    \item for every path in $T$ all of whose internal vertices are of degree $2$ and are in $V\setminus X$, contract this path, i.e. replace each such path with a single edge, until there is none.
\end{enumerate}
\end{definition}

\begin{observation}\label{obs proj tree size}
Every vertex of $\proj(T,X)$ that is not in $X$ has degree at least 3 in $\proj(T,X)$. Consequently, the number of vertices in $\proj(T,X)$ is $O(|X|)$.
\end{observation}

The next lemma gives a convenient way to work with $T$-feasible partitions of subsets of vertices. It is adapted and simplified from \cite{LSS}. We give a proof for the sake of completeness.

\begin{lemma}\label{lemma proj cut size}
Let $G=(V,E)$ be a graph and $T$ be a spanning tree of $G$.
\begin{enumerate}
    \item For a subset $X\subseteq V(G)$ and a partition $\mcp$ of $G$, if $\mcp_X$ is the restriction of $\mcp$ to $X$, then there exists a partition $\tilde{\mcp}$ of the vertices of $\proj(T,X)$ such that $|\delta_{\proj(T,X)}(\tilde{\mcp})|\leq |\delta_T(\mcp)|$ and the restriction of $\tilde{\mcp}$ to $X$ is $\mcp_X$. 
    \item For a subset $X\subseteq V$ and a partition $\tilde{\mcp}$ of the vertices of $\proj(T,X)$, if $\tilde{\mcp}_X$ is the restriction of $\tilde{\mcp}$ to $X$, then there exists a partition $\mcp$ of $V(G)$ such that $|\delta_T(\mcp)|\leq |\delta_{\proj(T,X)}(\tilde{\mcp})|$ and the restriction of $\mcp$ to $X$ is $\tilde{\mcp}_X$.
\end{enumerate}

\end{lemma}

\begin{proof}
We will start by proving the first statement. Let $X\subseteq V(G')$ and $\mcp$ be fixed as in the first statement. We will construct $\tilde{\mcp}$ by constructing $\delta_{\proj(T,X)}(\tilde{\mcp})$ as follows. For each edge $e\in \delta_T(\mcp)$, if $e$ remains in $\proj(T,X)$, then include $e$ into $\delta_{\proj(T,X)}(\tilde{\mcp})$; if $e$ is removed as part of a path that is replaced with an edge $e\in E(\proj(T,X))$, include $e'$ into $\delta_{\proj(T,X)}(\tilde{\mcp})$. By doing this, we can guarantee that $|\delta_{\proj(T,X)}(\tilde{\mcp})|\leq |\delta_T(\mcp)|$.

To see that $\delta(\tilde{\mcp})$ indeed yields a partition whose restriction to $X$ is $\mcp_X$, we claim that the partition $\tilde{\mcp}'$ whose parts are the connected components of $\proj(T,X)\backslash\delta_{\proj(T,X)}(\tilde{\mcp})$, when restricted to $X$, refines $\mcp_X$. 
For every pair of vertices $u,v\in X$ that are in different parts of $\mcp_X$, we know $u$ and $v$ are also in different parts of $\mcp$. 
This means some edge $e$ on the unique path in $T$ between $u$ and $v$ is in $\delta_T(\mcp)$. 
By the way we constructed $\delta_{\proj(T,X)}(\tilde{\mcp})$, either $e$ is contained in $\delta_{\proj(T,X)}(\tilde{\mcp})$, or $e'$ which replaces a path containing $e$ is contained in $\delta_{\proj(T,X)}(\tilde{\mcp})$. 
In either case, the unique path in $\proj(T,X)$ between $u$ and $v$ is disconnected. 
This proves our claim. 
To construct $\tilde{\mcp}$ from $\delta_{\proj(T,X)}(\tilde{\mcp})$, we simply group parts of $\tilde{\mcp}'$ together as necessary to comply with $\mcp_X$. 
This completes the proof of the first statement.

The proof of the second statement is similar to the preceding proof. 
Let $X\subseteq V$ and $\tilde{\mcp}$ be fixed as in the second statement. 
We start by constructing the edge set $\delta_T(\mcp)$. 
For each edge $e\in \delta_{\proj(T,X)}(\tilde{\mcp})$, if $e$ is originally an edge of $T$, then include $e$ into $\delta_T(\mcp)$; if $e$ is introduced to replace a path in $T$, then fix an arbitrary edge $e'$ in this path and include $e'$ into $\delta_T(\mcp)$. 
By doing this, clearly we can guarantee $|\delta_T(\mcp)|\leq |\delta_{\proj(T,X)}(\tilde{\mcp})|$.

By the same argument as in proof of the first statement, we can see that the connected components of $T\backslash \delta_T(\mcp)$ yields a partition that, when restricted to $X$, refines $\tilde{\mcp}_X$. Combining parts as necessary, we obtain a desired partition $\mcp$. This completes the proof of the second statement.
\end{proof}

\section{Fixed-parameter algorithm parameterized by $k$ and solution size}\label{sec:DP}
In this section we prove Theorem \ref{thm dp}. 
Let $(G=(V,E),k)$ be the input instance of \mmkc with $n$ vertices. The input graph $G$ could possibly have parallel edges.
We assume that $G$ is connected.
Let $\opt=\opt(G,k)$ (i.e., $\opt$ is the optimum objective value of \mmkc on input $G$) and let $\lambda$ be the input such that $\lambda\ge \opt$. We will design a dynamic programming algorithm that runs in time $(\lambda k)^{O(k^2)}n^{O(1)}$ to compute $\opt$. 

Given the input, we first use Lemma \ref{lemma tree decomposition} to obtain a tree decomposition $(\tau,\chi)$ of $G$ satisfying the conditions of the lemma. Since the algorithm in the lemma runs in polynomial time, the size of the tree decomposition $(\tau,\chi)$ is polynomial in the input size. Next, we use Lemma \ref{lemma thorup tree} to obtain a spanning tree $T$ such that there exists an optimum min-max $k$-partition $\Omega=(\Omega_1,\ldots,\Omega_k)$ of $V$ that $(2k^2)$-respects $T$, and moreover $T$ is a subgraph of $G$. We fix the tree decomposition $(\tau, \chi)$, the spanning tree $T$, and the optimum solution $\Omega$ with these choices in the rest of this section. We note that $\Omega_i\neq\emptyset$ for all $i\in[k]$ and $\max_{i\in[k]}|\delta_G(\Omega_i)|=\OPT$. We emphasize that the choice of $\Omega$ is fixed only for the purposes of the correctness of the algorithm and is not known to the algorithm explicitly. 

Our algorithm is based on dynamic program (DP). We will describe the subproblems of the DP in Section \ref{sec:dp-subproblems}. We will need the notion of a nice decomposition of the bags corresponding to the tree decomposition. We describe this notion in Section \ref{sec:nice-decomposition} and give an algorithm to generate them in Section \ref{sec:generating-nice-decompositions}. We will show the recursion to solve the dynamic program in Section \ref{sec:recursion}. We encourage the reader to trace towards the base case of the dynamic program on first read.

\subsection{Subproblems of the DP}\label{sec:dp-subproblems}
In this section, we state the subproblems in our dynamic program (DP), bound the number of subproblems in the DP, and prove Theorem \ref{thm dp}. For a tree node $t\in V(\tau)$, let $\mcf^{A_t}$ be the collection of partitions of the adhesion $A_t$ that are (i) $T$-feasible and (ii) have at most $k$ parts.
We emphasize that elements of $\mcf_{A_t}$ are of the form $\mcp_{A_t}=(\tilde{P}_1,\ldots,\tilde{P}_{k'})$ for some $k'\in\{0,1,\ldots,k\}$, where $\tilde{P}_i\neq \emptyset$ for all $i\in[k']$.
The following lemma bounds the size of $\mathcal{F}^{A_t}$, which in turn, will be helpful in bounding the number of subproblems to be solved in our dynamic program.

\begin{lemma}\label{lemma F family size}
For every tree node $t\in V(\tau)$, we have $|\mathcal{F}^{A_t}|=(\lambda k)^{O(k^2)}$. Moreover, the collection $\mathcal{F}^{A_t}$ can be enumerated in $(\lambda k)^{O(k^2)}$ time.
\end{lemma}
\begin{proof}
First we claim that a partition $\mathcal{P}_{A_t}$ of $A_t$ is $T$-feasible if and only if it is $\proj(T,A_t)$-feasible. 

Assume a partition $\mathcal{P}_{A_t}$ of $A_t$ is $T$-feasible, realized by a partition $\mcp$ of $V$. It follows that $|\delta_T(\mcp)|\leq 2k^2$. By Lemma \ref{lemma proj cut size}, there exists a partition $\tilde{\mcp}$ of $\proj(T,A_t)$ such that $|\delta_{\proj(T,A_t)}(\tilde{\mcp})|\leq 2k^2$ and the restriction of $\tilde{\mcp}$ to $A_t$ is $\mcp_{A_t}$. This is equivalent to saying $\mcp_{A_t}$ is $\proj(T,A_t)$-feasible.

The other direction is similar. If a partition $\mcp_{A_t}$ of $A_t$ is $\proj(T,A_t)$-feasible, realized by a partition $\tilde{\mcp}$ of $V(\proj(T,A_t))$, it follows that $|\delta_\proj(T,A_t)(\tilde{\mcp})|\leq 2k^2$. By Lemma \ref{lemma proj cut size} there exists a partition $\mcp$ of $G$ such that $|\delta_T(\mcp)|\leq 2k^2$ and the restriction of $\mcp$ to $A_t$ is $\mcp_{A_t}$. Hence $\mcp_{A_t}$ is $T$-feasible.

It remains to bound the number of $\proj(T,A_t)$-feasible partitions of $A_t$. By Observation \ref{obs proj tree size}, the size of $\proj(T,A_t)$ is $O(\lambda k)$. We notice that partitions with at most $k$ parts that $2k^2$-respect $\proj(T,A_t)$ can be enumerated by removing up to $2k^2$ edges of $\proj(T,A_t)$, and putting the resulting connected components (there are at most $2k^2+1$ of them) into $k$ bins. 
Therefore, combining the previous observation, we conclude that
\[|\mathcal{F}^{A_t}|= O(\lambda k)^{2k^2}\cdot k^{2k^2+1}=(\lambda k)^{O(k^2)}.\]
Moreover, the time required to compute $\mathcal{F}^{A_t}$ (by enumerating all eligible partitions as above) is also $(\lambda k)^{O(k^2)}$. 
\end{proof}

The following definition will be useful in identifying the subproblems of the DP. 
\begin{definition}\label{defn:fc-fs}
Let $t\in V(\tau)$ be a tree node, and $f_t:\mathcal{F}^{A_t}\times\{0,1,\ldots,\lambda\}^k\times\{0,1,\ldots,2k^2\}\to\{0,1\}$ be a Boolean function.
\begin{enumerate}
    \item
    \textbf{(Correctness)}
    The function $f_t$ is \emph{\fc} if we have $f_t(\mathcal{P}_{A_t},\Bar{x},d)=1$ for all $\mathcal{P}_{A_t}=(\Tilde{P}_1,\ldots ,\Tilde{P}_{k'})\in \mathcal{F}^{A_t}$, 
    $\Bar{x}\in \lrange^k$, and $d\in\krange$ for which there exists a $k$-subpartition $\mathcal{P}=(P'_1,\ldots,P'_k)$ of $V(G_t)$ satisfying the following conditions:
    \begin{enumerate}[(i)]
        \item $P'_i\cap A_t=\Tilde{P}_i$ for all $i\in [k']$, 
        \item $|\delta_{G_t}(P'_i)|= x_i$ for all $i\in [k]$,
        \item $|\delta_T(\mathcal{P})|\leq d$, and
        \item $\mcp$ is a restriction of $\Omega$ to $V(G_t)$.
    \end{enumerate}
    A $k$-subpartition of $V(G_t)$ satisfying the above four conditions is said to \emph{witness} \fcness of $f_t(\mathcal{P}_{A_t},\Bar{x},d)$. 
    
    \item 
    \textbf{(Soundness)} The function $f_t$ is \emph{\fs} if for all $\mathcal{P}_{A_t}=(\Tilde{P}_1,\ldots ,\Tilde{P}_{k'})\in \mathcal{F}^{A_t}$, $\Bar{x}\in \lrange^k$ and $d\in\krange$, we have $f_t(\mcp_{A_t},\Bar{x},d)=1$ only if there exists a $k$-subpartition $\mcp=(P'_1,\ldots,P'_k)$ of $V(G_t)$ satisfying conditions (i), (ii) and (iii) above. A $k$-subpartition of $V(G_t)$ satisfying (i), (ii) and (iii) is said to \emph{witness} \fsness of $f_t(\mathcal{P}_{A_t},\Bar{x},d)$.
\end{enumerate}
\end{definition}

We emphasize the distinction between correctness and soundness: correctness relies on all four conditions while soundness relies only on three conditions. We discuss the need for distinct correctness and soundness definitions after Lemma \ref{lemma:nice decomp}.

The next proposition shows that an \fc and \fs function for the root node of the tree decomposition can be used to recover the optimum value. 
\begin{proposition}\label{prop:OPT from fr function}
If we have a function $f_r:\mathcal{F}^{A_r}\times\{0,1,\ldots,\lambda\}^k\times\{0,1,\ldots,2k^2\}\to\{0,1\}$ that is both \fc and \fs, where $r$ is the root of the tree decomposition $\tau$, then 
\[\OPT=\min\left\{\max_{i\in[k]}\{x_i\}:f_r(\mathcal{P}_{\emptyset},\Bar{x},2k^2)=1,\Bar{x}\in[\lambda]^k\right\},\]
where $\mathcal{P}_\emptyset$ is the 0-tuple that denotes the trivial partition of $A_r=\emptyset$. 
\end{proposition}
\begin{proof}
The optimum partition $\Omega$ is a $k$-subpartition of $V=V(G_r)$ that witnesses \fcness of $f_r(\mcp_\emptyset,(|\delta_G(\Omega_1)|,\ldots,|\delta_G(\Omega_k)|),2k^2)$. 
Hence, $f_r(\mcp_\emptyset,(|\delta_G(\Omega_1)|,\ldots,|\delta_G(\Omega_k)|),2k^2)=1$, where $|\delta_G(\Omega_i)|\in[\lambda]$ for every $i\in [k]$. Consequently, 
\[\min\left\{\max_{i\in[k]}\{x_i\}:f_r(\mathcal{P}_{\emptyset},\Bar{x},2k^2)=1,\Bar{x}\in[\lambda]^k\right\}\leq\max_{i\in[k]}\{|\delta_G(\Omega_i)|\}=\OPT.\]

We now show the reverse inequality. Suppose that $f_r(\mathcal{P}_{\emptyset},\Bar{x},2k^2)=1$ for some $\Bar{x}\in[\lambda]^k$. Then there exists a $k$-subpartition $\mcp'$ of $V$ witnessing \fsness of $f_r(\mathcal{P}_{\emptyset},\Bar{x},2k^2)$. Since $\Bar{x}\in[\lambda]^k$, we know that $x_i\ge 1$ for every $i\in [k]$. Since $G$ is connected, this implies that each part of $\mcp'$ is non-empty. Therefore, $\mcp'$ is also a $k$-partition of $V$ and is hence, feasible for \mmkc. This implies that
\[
\max_{i\in [k]} \{x_i\} \ge \opt.
\]
\end{proof}

By Proposition \ref{prop:OPT from fr function}, it suffices to compute an \fc and \fs function $f_r$, where $r$ is the root of the tree decomposition $\tau$. We will compute this in a bottom-up fashion on the tree decomposition using the following lemma. 

\begin{lemma}\label{lemma:children f to parent f}
There exists an algorithm that takes as input 
$(\tau,\chi)$, a tree node $t\in V(\tau)$, Boolean functions $f_{t'}:\mathcal{F}^{A_{t'}}\times\{0,1,\ldots,\lambda\}^k\times\{0,1,\ldots,2k^2\}\to\{0,1\}$ for every child $t'$ of $t$ that are \fc and \fs, and runs in time $(\lambda k)^{O(k^2)}n^{O(1)}$ to return a function  $f_t:\mathcal{F}^{A_t}\times\{0,1,\ldots,\lambda\}^k\times\{0,1,\ldots,2k^2\}\to\{0,1\}$ that is \fc and \fs.
\end{lemma}

We now complete the proof of Theorem \ref{thm dp} using Lemma \ref{lemma:children f to parent f} and Proposition \ref{prop:OPT from fr function}. 
\begin{proof}[Proof of Theorem \ref{thm dp}]
In order to compute a function $f_r:\mathcal{F}^{A_r}\times\{0,1,\ldots,\lambda\}^k\times\{0,1,\ldots,2k^2\}\to\{0,1\}$ that is both \fc and \fs, we can apply Lemma \ref{lemma:children f to parent f} on each tree node $t\in V(\tau)$ in a bottom up fashion starting from the leaf nodes of the tree decomposition. Therefore, using Lemmas \ref{lemma F family size} and \ref{lemma:children f to parent f}, the total run time to compute $f_r$ is 
\[(\lambda k)^{O(k^2)}n^{O(1)}\cdot|V(\tau)|=(\lambda k)^{O(k^2)}n^{O(1)}\cdot\poly(n,\lambda,k)=(\lambda k)^{O(k^2)}n^{O(1)}.\]
Using Proposition \ref{prop:OPT from fr function}, we can compute OPT from the function $f_r$. Consequently, the total time to compute OPT is $(\lambda k)^{O(k^2)}n^{O(1)}$.
\end{proof}

We will prove Lemma \ref{lemma:children f to parent f} in the following subsections. 
We fix the tree node $t\in V(\tau)$ for the rest of the subsections. 

\subsection{Nice decomposition}\label{sec:nice-decomposition}

We need the notion of a nice decomposition that we define below. Our definition differs from the notion of the nice decomposition defined by \cite{LSS} in property (ii). 

\begin{definition}\label{def nice decomp}
A \emph{nice decomposition} of $\chi(t)$ is a triple $(\mathcal{P}_{\chi(t)},\mathcal{Q}_{\chi(t)},O)$ where $\mathcal{P}_{\chi(t)}$ and $\mathcal{Q}_{\chi(t)}$ 
are partitions of $\chi(t)$, $\mathcal{Q}_{\chi(t)}$ refines $\mathcal{P}_{\chi(t)}$, and $O$ is either a part of $\mathcal{P}_{\chi(t)}$ or $\emptyset$. Additionally, the following conditions need to be met:
\begin{enumerate}[(i)]
    \item If $O\neq\emptyset$, then $O$ is a part of $\mathcal{Q}_{\chi(t)}$.
    \par If $O=\emptyset$, then $\mathcal{P}_{\chi(t)}$ has only one part.
    \item For every part $P$ of $\mathcal{P}_{\chi(t)}$, $P$ contains at most $4k^2+1$ parts of $\mathcal{Q}_{\chi(t)}$.
    \item For every pair of distinct parts $P,P'$ of $\mathcal{P}_{\chi(t)}$ other than $O$, there are no edges between $P$ and $P'$.
    \item If $t'\in V(\tau)$ is a child of $t$ or $t$ itself, then $A_{t'}$ intersects with at most one part of $\mathcal{P}_{\chi(t)}$ other than $O$.
\end{enumerate}
\end{definition}
 In order to compute an \fc and \fs function $f_t:\mathcal{F}^{A_t}\times\{0,1,\ldots,\lambda\}^k\times\{0,1,\ldots,2k^2\}\to\{0,1\}$, we will compute a family $\mcd$ of nice decompositions of $\ct$ such that if there exists a $k$-subpartition $\Pi$ of $V(G_t)$ that realizes \fcness of $f_t(\mcp_{A_t}, \Bar{x}, d)$ for some $\mcp_{A_t}\in\mathcal{F}^{A_t}$, $\bar{x}\in\lrange^k$ and $d\in\krange$, then there exists a nice decomposition $(\mcp_{\ct},\mcq_{\ct},O)$ in $\mathcal{D}$ such that $\mcq_{\ct}$ refines a restriction of $\Pi$ to $\ct$.
A formal statement is given in Lemma \ref{lemma:nice decomp}.

\begin{restatable}{lemma}{lemGeneratingNiceDecomps}
\label{lemma:nice decomp}
There exists an algorithm that takes as input 
the spanning tree $T$, 
the tree decomposition $(\tau, \chi)$, 
a tree node $t\in V(\tau)$, 
and runs in time $(\lambda k)^{O(k^2)}n^{O(1)}$ 
to return a family $\mathcal{D}$ of nice decompositions of $\ct$ with $|\mathcal{D}|=(\lambda k)^{O(k^2)}n^{O(1)}$. Additionally, if a $k$-subpartition $\Pi$ of $V(G_t)$ realizes \fcness of $f_t(\mathcal{P}_{A_t},\Bar{x},d)$ for some $\mathcal{P}_{A_t}\in\mathcal{F}^{A_t}$, $\Bar{x}\in\lrange^k$ and $d\in\krange$, then $\mathcal{D}$ contains a nice decomposition $(\mcp_{\chi(t)},\mathcal{Q}_{\chi(t)},O)$ where $\mathcal{Q}_{\chi(t)}$ refines a restriction of $\Pi$ to $\chi(t)$.
\end{restatable}

We defer the proof of Lemma \ref{lemma:nice decomp} to Section \ref{sec:generating-nice-decompositions}. 

We now discuss the need for distinct definitions for correctness and soundness. 
If a $k$-subpartition $\Pi$ of $V(G_t)$ realizes \fcness of $f_t(\mcp_{A_t},\Bar{x},d)$ for some $\mcp_{A_t}\in\mcf^{A_t}$, $\Bar{x}\in\lrange^k$ and $d\in\krange$, then $\Pi$ is a restriction of $\Omega$, the optimal partition, to $V(G_t)$.
The family $\mcd$ of nice decompositions then serves to provide a partition $\mcq_{\ct}$ of $\ct$ that refines a restriction of the optimal partition $\Omega$ to $\ct$, which will later be used to identify an optimal partition. We note that if a $k$-subpartition $\Pi$ of $V(G_t)$ only witnesses $f$-soundness of $f_t(\mcp_{A_t},\bar{x},d)$ for some $\mcp_{A_t}\in\mcf^{A_t}$, $\bar{x}\in\lrange^k$ and $d\in\krange$, then the family $\mcd$ is not guaranteed to provide a refinement of a restriction of the $k$-subpartition $\Pi$ to $\ct$. This motivates the two distinct definitions for correctness and soundness.

\subsection{Computing an \fc and \fs function $f_t$}\label{sec:recursion}
In this section, we will prove Lemma \ref{lemma:children f to parent f}. For a fixed tree node $t\in V(\tau)$, we will describe an algorithm to assign values to $f_t(\mcp_{A_t},\Bar{x},d)$ for all $\mcp_{A_t}\in\mcf^{A_t}$, $\bar{x}\in\lrange^k$, and $d\in\krange$ based on the value of $f_{t'}(\mcp_{A_{t'}},\Bar{\alpha},\beta)$ for all children $t'$ of $t$, and all $\mcp_{A_{t'}}\in\mcf^{A_{t'}}$, $\bar{\alpha}\in\lrange^k$, and $\beta\in\krange$ so that the resulting function $f_t$ is \fc and \fs. 

For the fixed $t\in V(\tau)$, we use Lemma \ref{lemma:nice decomp} to obtain a family $\mcd$ of nice decompositions of $\chi(t)$. 
Our plan to compute the function $f_t$ involves working with each nice decomposition $D\in \mcd$. 
The following definition will be helpful in transforming our goal of computing the function $f_t$. 

\begin{definition}
Let $D:=(\mcp_{\chi(t)},\mcq_{\chi(t)},O)\in\mcd$, and $r^D:\mcf^{A_t}\times[\lambda]^k\times\{0,1,\ldots,2k^2\}\to\{0,1\}$ be a Boolean function.
\begin{enumerate}
    \item
    \textbf{(Correctness)}
    The function $r^D$ is \emph{\rc} if we have $r^D(\mcp_{A_t},\Bar{x},d)=1$ for all $\mcp_{A_t}=(\tilde{P}_1, \ldots, \tilde{P}_{k'})\in\mcf^{A_t}$, $\Bar{x}\in \lrange^k$ and $d\in\krange$ for which there exists a $k$-subpartition $\mcp=(P'_1,\ldots,P'_k)$ of $V(G_t)$ satisfying the following conditions:
    \begin{enumerate}[(i)]
        \item $P'_i\cap A_t=\Tilde{P}_i$ for all $i\in [k']$, 
        \item $|\delta_{G_t}(P'_i)|= x_i$ for all $i\in [k]$,
        \item $|\delta_T(\mcp)|\leq d$,
        \item $\mcp$ restricted to $\ct$ coarsens $\mcq_{\chi(t)}$, and
        \item $\mcp$ is a restriction of $\Omega$ to $V(G_t)$.
    \end{enumerate}
    A $k$-subpartition of $V(G_t)$ satisfying the above five conditions is said to \emph{witness} \rcness of $r^D(\mcp_{A_t},\Bar{x},d)$.
    \item 
    \textbf{(Soundness)} The function $r^D$ is \emph{\rs} if for all $\mcp_{A_t}=(\tilde{P}_1, \ldots, \tilde{P}_{k'})\in\mcf^{A_t}$, $\Bar{x}\in \lrange^k$ and $d\in\krange$, we have $r^D(\mcp_{A_t},\Bar{x},d)=1$ only if there exists a $k$-subpartition $\mcp=(P'_1,\ldots,P'_k)$ of $V(G_t)$ satisfying conditions (i), (ii), (iii) and (iv) above. A $k$-subpartition of $V(G_t)$ satisfying (i), (ii), (iii) and (iv) is said to \emph{witness} \rsness of $r^D(\mcp_{A_t},\Bar{x},d)$.
\end{enumerate}
\end{definition}

The next proposition shows that \rc and \rs functions can be used to recover an \rc and \fc function. 
\begin{proposition}\label{prop:rD to ft}
Suppose that we have functions $r^D:\mcf^{A_t}\times[\lambda]^k\times\{0,1,\ldots,2k^2\}\to\{0,1\}$ for every $D\in \mcd$ such that all of them are both \rc and \rs. Then, the function $f_t$ obtained by setting 
\[f_t(\mcp_{A_t},\Bar{x},d):=\max\set{r^D(\mcp_{A_t},\Bar{x},d):\ D\in\mcd}\]
for every $\mcp_{A_t}\in\mcf^{A_t}$, $\Bar{x}\in \lrange^k$ and $d\in\krange$ is both \fc and \fs.
\end{proposition}
\begin{proof}
We first show \fcness. For $\mcp_{A_t}\in\mcf^{A_t}$, $\bar{x}\in\lrange^k$ and $d\in\krange$, suppose that there exists a $k$-subpartition $\Pi$ of $V(G_t)$ witnessing \fcness of $f_t(\mcp_{A_t},\Bar{x},d)$. By Lemma \ref{lemma:nice decomp}, we know that $\mcd$ contains a nice decomposition $D^\ast=(\mcp^\ast_{\ct},\mcq^\ast_{\ct},O^\ast)$ such that $\mcq^\ast_{\ct}$ refines $\Pi$ restricted to $\ct$. Then by \rcness of $r^{D^{\ast}}$, we know that $r^{D^{\ast}}(\mcp_{A_t},\bar{x},d)=1$. This implies $f_t(\mcp_{A_t},\Bar{x},d)=1$.

Next we show \fsness. Suppose that $f_t(\mcp_{A_t},\Bar{x},d)=1$ for some $\mcp_{A_t}\in\mcf^{A_t}$, $\bar{x}\in\lrange^k$ and $d\in\krange$. Then, there exists $D\in \mcd$ such that $r^D(\mcp_{A_t},\bar{x},d)=1$. By \rsness of the function $r^D$, there exists a $k$-subpartition $\Pi'$ of $V(G_t)$ witnessing \rsness of $r^D(\mcp_{A_t},\bar{x},d)$. It follows that $\Pi'$ also witnesses \fsness of $f_t(\mcp_{A_t},\Bar{x},d)$. 
\end{proof}
Our goal now is to compute an \rc and \rs function $r^D$ for each $D\in \mcd$. 

\begin{lemma}\label{lemma:rD update}
There exists an algorithm that takes as input 
$(\tau,\chi)$, a tree node $t\in V(\tau)$, a nice decomposition $D=(\mcp_{\chi(t)},\mcq_{\chi(t)},O)\in\mcd$, together with Boolean functions $f_{t'}:\mcf^{A_{t'}}\times\lrange^k\times\krange\to\{0,1\}$ for every child $t'$ of $t$ that are \fc and \fs, and 
runs in time $(\lambda k)^{O(k^2)}n^{O(1)}$
to return a function $r^D:\mcf^{A_t}\times[\lambda]^k\times\{0,1,\ldots,2k^2\}\to\{0,1\}$ that is \rc and \rs.
\end{lemma}

We now use Lemma \ref{lemma:rD update} to complete the proof of Lemma \ref{lemma:children f to parent f}.
\begin{proof}[Proof of Lemma \ref{lemma:children f to parent f}]
Using Lemma \ref{lemma:nice decomp}, we compute a family $\mcd$ of nice decompositions in time $(\lambda k)^{O(k^2)}n^{O(1)}$, where $|\mcd|=(\lambda k)^{O(k^2)}n^{O(1)}$. For each $D\in \mcd$, we use Lemma \ref{lemma:rD update} to compute a function $r^D:\mcf^{A_t}\times[\lambda]^k\times\{0,1,\ldots,2k^2\}\to\{0,1\}$ that is both \rc and \rs in time $(\lambda k)^{O(k^2)}n^{O(1)}$. Finally, we use Proposition \ref{prop:rD to ft} to compute the desired function $f_t$ that is both \fc and \fs. The total run-time is 
\[(\lambda k)^{O(k^2)}n^{O(1)}+(\lambda k)^{O(k^2)}n^{O(1)}\cdot (\lambda k)^{O(k^2)}n^{O(1)}=(\lambda k)^{O(k^2)}n^{O(1)}.\]
\end{proof}
The rest of the section is devoted to proving Lemma \ref{lemma:rD update}.
We fix the inputs specified in Lemma \ref{lemma:rD update} for the rest of this section. In particular, we fix $D=(\mcp_{\chi(t)},\mcq_{\chi(t)},O)\in\mcd$. 

\paragraph{Notations.} 
Let $\mathcal{P}_{\chi(t)}=(P_1,\ldots,P_p,O)$. If $O=\emptyset$, we will abuse notation and use $\mcp_{\ct}$ to refer to the partition of $\ct$ containing only one part, namely $\mcp_{\ct}=(P_1)$. We note that $p\le n$ since $D$ is a nice decomposition. 
We define $P_{\leq \ell}:=\cup_{i=1}^\ell P_i$, where $\ell\in\prange$. It follows that $P_{\leq 0}=\emptyset$. We specially define $P_0:=O$ for indexing convenience.
For every $\ell\in\prange$, we define $\mathcal{A}(\ell)$ to be the set of children of $t$ whose adhesion is contained in $O\cup P_\ell$ and intersects $P_\ell$, i.e.,
\[\mathcal{A}(\ell):=\{t'\in V(\tau):t'\text{ is a child of }t, A_{t'}\subseteq O\cup P_\ell, A_{t'}\cap P_\ell\neq\emptyset\}.\]
Moreover, let $\mathcal{A}(\ell):=\{t_1,t_2,...,t_{|\mathcal{A}(\ell)|}\}$. For each $\ell\in \set{0, 1, \ldots, p}$ and $a\in\{0,1,\ldots,|\mathcal{A}(\ell)|\}$, let
\begin{align*}
G_\leq(\ell,a)&:=G\left[O\cup P_\ell\cup\bigcup_{1\leq i\leq a}V(G_{t_i})\right], \\
G(\ell)&:=G\left[O\cup P_\ell\cup\bigcup_{t\in\mathcal{A}(\ell)}V(G_{t'})\right], \text{ and }\\
G_\leq(\ell)&:=G\left[\bigcup_{0\leq i\leq \ell}V(G(i))\right].
\end{align*}

\begin{figure}
    \centering
    \includegraphics[width=0.6\textwidth]{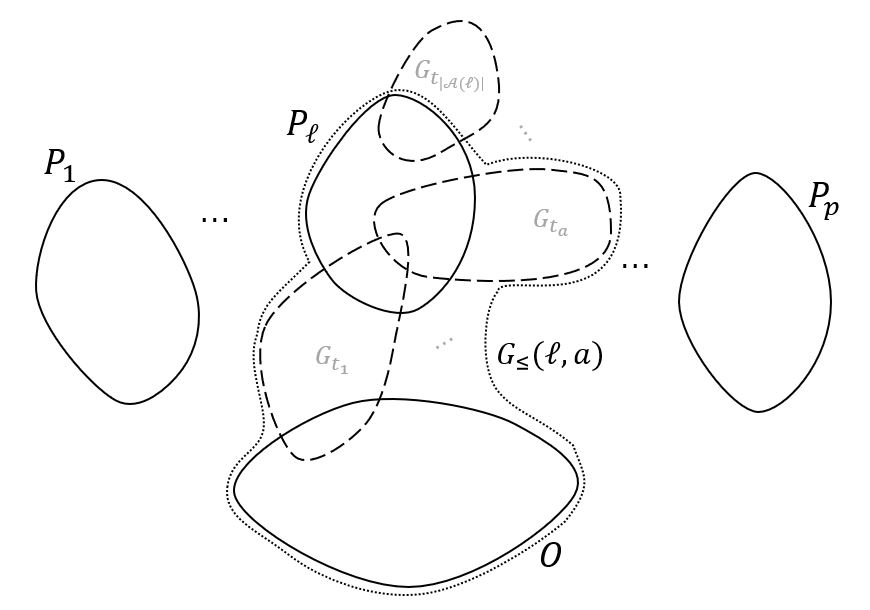}
    \caption{An example of $G_\leq(\ell,a)$. Here the regions enclosed by dashed lines represent $V(G_{t'})$ that intersects $P_\ell$, where $t'$ runs over $\mathcal{A}(\ell)$. The region enclosed by the dotted line is $G_\leq(\ell,a)$.}
    \label{fig:G_leq(l,a)}
\end{figure}

\begin{figure}
    \centering
    \includegraphics[width=0.6\textwidth]{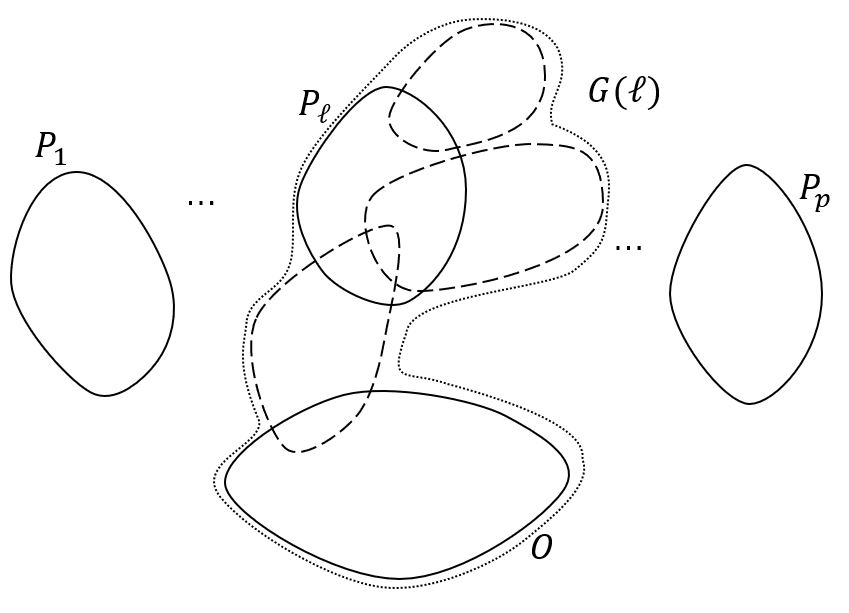}
    \caption{An example of $G(\ell)$. Here the regions enclosed by dashed lines represent $V(G_{t'})$ that intersects $P_\ell$, where $t'$ runs over children of $t$. The region enclosed by the dotted line is $G(\ell)$.}
    \label{fig:G(l)}
\end{figure}


These subgraphs are illustrated in Figures \ref{fig:G_leq(l,a)} and \ref{fig:G(l)}. In order to compute an \rc and \rs function $r^D$, we will employ new sub-problems that we define below. 
\begin{definition}
Let $g^{D,\mcp_{A_t}}:\{0,1,\ldots,p\}\times\{0,1,\ldots,2k^2\}\times\lrange^k\times[k]\to\{0,1\}$ be a Boolean function, where $D=(\mcp_{\ct}=(P_1,\ldots,P_p,O),\mcq_{\ct},O)\in\mathcal{D}$ and $\mcp_{A_t}=(\tilde{P}_1, \ldots, \tilde{P}_{k'})\in\mathcal{F}^{A_t}$.
\begin{enumerate}
    \item
    \textbf{(Correctness)}
    The function $g^{D,\mcp_{A_t}}$ is \emph{\gc} if we have $g^{D,\mcp_{A_t}}(\ell,d,\bar{y},q)=1$ for all $\ell\in\prange, d\in\krange, y \in\lrange^k$ and $q\in[k]$ for which there exists a $k$-subpartition $\mathcal{P}=(P'_1,\ldots,P'_k)$ of $V(G_\leq(\ell))$ satisfying the following conditions:
    \begin{enumerate}[(i)]
        \item $\mcp$ restricted to $O\cup P_{\le \ell}$ is a coarsening of $\mathcal{Q}_{\chi(t)}$ restricted to $O\cup P_{\le \ell}$,
        \item $|\delta_T(\mathcal{P})|\leq d$, 
        \item $|\delta_{G_\leq(\ell)}(P'_i)|= y_i$  for all $i\in [k]$,
        \item $O\subseteq P'_q$,
        \item if $A_t\subseteq O\cup P_{\leq \ell}$, then $P'_i\cap A_t=\Tilde{P}_i$, for all $i\in [k']$, and
        \item $\mcp$ is a restriction of $\Omega$ to $V(G_\leq(\ell))$.
    \end{enumerate}
    A $k$-subpartition of $V(G_\leq(\ell))$ satisfying the above six conditions is said to \emph{witness} \gcness of $g^{D,\mcp_{A_t}}(\ell,d,\bar{y},q)$.
    \item 
    \textbf{(Soundness)} The function $g^{D,\mcp_{A_t}}$ is \emph{\gs} if for all $\ell\in\prange, d\in\krange, y \in\lrange^k$ and $q\in[k]$, we have $g^{D,\mcp_{A_t}}(\ell,d,\bar{y},q)=1$ only if there exists a $k$-subpartition $\mathcal{P}=(P'_1,\ldots,P'_k)$ of $V(G_\leq(\ell))$ satisfying conditions (i), (ii), (iii), (iv) and (v) above. A $k$-subpartition of $V(G_\leq(\ell))$ satisfying (i), (ii), (iii), (iv) and (v) is said to \emph{witness} \gsness of $g^{D,\mcp_{A_t}}(\ell,d,\bar{y},q)$.
\end{enumerate}
\end{definition}

\begin{definition}
Let $h^{D,\mcp_{A_t}}:\{0,1,\ldots,p\}\times\{0,1,\ldots,2k^2\}\times\lrange^k\times[k]\to\{0,1\}$ be a Boolean function, where $D\in\mathcal{D}$ and $\mcp_{A_t}=(\tilde{P}_1, \ldots, \tilde{P}_{k'})\in\mathcal{F}^{A_t}$.
\begin{enumerate}
    \item
    \textbf{(Correctness)}
    The function $h$ is \emph{\hc} if we have $h^{D,\mcp_{A_t}}(\ell,d,\bar{y},q)=1$ for all $\ell\in\prange, d\in\krange, y \in\lrange^k$ and $q\in[k]$ for which there exists a $k$-subpartition $\mathcal{P}=(P'_1,\ldots,P'_k)$ of $V(G(\ell))$ satisfying the following conditions:
    \begin{enumerate}[(i)]
        \item $\mcp$ restricted to $O\cup P_{\ell}$ is a coarsening of 
13
 $\mathcal{Q}_{\chi(t)}$ restricted to $O\cup P_{\ell}$,
        \item $|\delta_T(\mathcal{P})|\leq d$, 
        \item $|\delta_{G(\ell)}(P'_i)|= y_i$  for all $i\in [k]$,
        \item $O\subseteq P'_q$,
        \item if $A_t\subseteq O\cup P_{\ell}$, then $P'_i\cap A_t=\Tilde{P}_i$, for all $i\in [k']$, and
        \item $\mcp$ is a restriction of $\Omega$ to $V(G(\ell))$.
    \end{enumerate}
    A $k$-subpartition of $V(G(\ell))$ satisfying the above six conditions is said to \emph{witness} \hcness of $h^{D,\mcp_{A_t}}(\ell,d,\bar{y},q)$.
    \item 
    \textbf{(Soundness)} The function $h^{D,\mcp_{A_t}}$ is \emph{\hs} if for all $\ell\in\prange, d\in\krange, y \in\lrange^k$ and $q\in[k]$, we have $h^{D,\mcp_{A_t}}(\ell,d,\bar{y},q)=1$ only if there exists a $k$-subpartition $\mathcal{P}=(P'_1,\ldots,P'_k)$ of $V(G(\ell))$ satisfying conditions (i), (ii), (iii), (iv) and (v) above. Such $k$-subpartition $\mcp$ is said to \emph{witness} \hsness of $h^{D,\mcp_{A_t}}(\ell,d,\bar{y},q)$.
\end{enumerate}
\end{definition}

The following proposition outlines how to compute a function $r^D$ that is both \rc and \rs using functions $g^{D,\mcp_{A_t}}$ for every $\mcp_{A_t}\in \mathcal{F}^{A_t}$.
\begin{proposition}\label{prop:g-to-rD}
Suppose that we have functions $g^{D,\mcp_{A_t}}:\{0,1,\ldots,p\}\times\{0,1,\ldots,2k^2\}\times\lrange^k\times[k]\to\{0,1\}$ for every $\mcp_{A_t}\in\mathcal{F}^{A_t}$ such that all of them are both \gc and \gs. Then, the function $r^D$ obtained by setting 
\[r^D(\mcp_{A_t},\Bar{x},d):=\max \set{g^{D,\mcp_{A_t}}(p,d,\Bar{x},q): q\in[k]}\]
for every $\mcp_{A_t}\in\mathcal{F}^{A_t}$, $\Bar{x}\in \lrange^k$ and $d\in\krange$ is both \rc and \rs.
\end{proposition}

\begin{proof}
We first show \rcness. For $\mcp_{A_t}\in\mathcal{F}^{A_t}$, $\bar{x}\in\lrange^k$ and $d\in\krange$, suppose that there exists a $k$-subpartition $\Pi=(\pi_1,\ldots,\pi_k)$ of $V(G_t)$ witnessing \rcness of $r^D(\mcp_{A_t},\bar{x},d)$. Then, there exists $q'\in [k]$ such that $O\subseteq \pi_{q'}$. It follows that $\Pi$ also witnesses \gcness of $g^{D,\mcp_{A_t}}(p,d,\bar{x},q')$, and hence $g^{D,\mcp_{A_t}}(p,d,\bar{x},q')=1$ since the function $g^{D,\mcp_{A_t}}$ is \gc. This implies that $r^D(\mcp_{A_t}, \Bar{x}, d)=1$. 

Next, we show \rsness. Suppose that $r^D(\mcp_{A_t},\Bar{x},d)=1$ for some $\mcp_{A_t}\in\mathcal{F}^{A_t}$, $\bar{x}\in\lrange^k$ and $d\in\krange$. Then, there exists $q'\in[k]$ such that $g^{D,\mcp_{A_t}}(p,d,\bar{x},q')=1$. 
By \gsness of the function $g^{D,\mcp_{A_t}}$, there exists a $k$-subpartition $\Pi'$ of $V(G_\leq(p))=V(G_t)$ that witnesses \gsness of $g^{D,\mcp_{A_t}}(p,d,\bar{x},q')$. It follows that $\Pi'$ also witnesses \rsness of $r^D(\mcp_{A_t},\bar{x},d)$.
\end{proof}

Our goal now is to compute a function $g^{D,\mcp_{A_t}}$ that is both \gc and \gs.

\begin{lemma}\label{lemma:compute g}
There exists an algorithm that takes as input 
$(\tau,\chi)$, 
a tree node $t\in V(\tau)$, 
a partition $\mathcal{P}_{A_t}\in\mathcal{F}^{A_t}$, 
a nice decomposition $(\mcp_{\chi(t)}=(P_1,\ldots,P_p,O),\mcq_{\chi(t)},O)\in \mcd$, 
together with Boolean functions $f_{t'}:\mathcal{F}^{A_{t'}}\times\lrange^k\times\krange\to\{0,1\}$ for every child $t'$ of $t$ that are \fc and \fs,
and runs in time $(\lambda k)^{O(k^2)}n^{O(1)}$
to return a function $g^{D,\mcp_{A_t}}:\{0,1,\ldots,p\}\times\{0,1,\ldots,2k^2\}\times\lrange^k\times[k]\to\{0,1\}$ that is \gc and \gs.
\end{lemma}

Lemma \ref{lemma:rD update} follows from Lemma \ref{lemma:compute g} and Proposition \ref{prop:g-to-rD}. We will prove Lemma \ref{lemma:compute g} in the following subsections.

\subsubsection{Computing $g^{D,\mcp_{A_t}}$ assuming $h^{D,\mcp_{A_t}}$ is available}
In this section, for a given pair of $D=(\mcp_{\ct}=(P_1,\ldots,P_p,O),\mcq_{\ct},O)\in\mathcal{D}$ and $\mcp_{A_t}\in\mathcal{F}^{A_t}$, we will show how to construct a function $g^{D,\mcp_{A_t}}:\{0,1,\ldots,p\}\times\{0,1,\ldots,2k^2\}\times\lrange^k\times[k]\to\{0,1\}$ that is both \gc and \gs using a function $h^{D,\mcp_{A_t}}:\{0,1,\ldots,p\}\times\{0,1,\ldots,2k^2\}\times\lrange^k\times[k]\to\{0,1\}$ that is \hc and \hs.

\begin{lemma}\label{lemma:g update from h}
There exists an algorithm that takes as input 
a partition $\mcp_{A_t}\in\mathcal{F}^{A_t}$, 
a nice decomposition $D=(\mcp_{\chi(t)}=(P_1,\ldots,P_p,O),\mcq_{\chi(t)},O)\in\mathcal{D}$, 
together with a Boolean function $h^{D,\mcp_{A_t}}:\{0,1,\ldots,p\}\times\{0,1,\ldots,2k^2\}\times\lrange^k\times[k]\to\{0,1\}$ that is \hc and \hs, and 
runs in time $(\lambda k)^{O(k^2)}n$
to return a function $g^{D,\mcp_{A_t}}:\{0,1,\ldots,p\}\times\{0,1,\ldots,2k^2\}\times\lrange^k\times[k]\to\{0,1\}$ that is \gc and \gs.
\end{lemma}

\begin{proof} 
Let the inputs be fixed as in the lemma. We will iteratively assign values to $g^{D,\mcp_{A_t}}(\ell,d,\Bar{w}, q)$ for $\ell=0,1,\ldots, p$. 

For $\ell=0$, we set $g^{D,\mcp_{A_t}}(0,d,\Bar{w}, q):=h^{D,\mcp_{A_t}}(0,d,\Bar{w},q)$ for every $d\in\krange$, $\bar{w}\in\lrange^k$ and $q\in[k]$. We observe that if there exists a $k$-subpartition of $V(G_{\le}(0))$ that witnesses \gcness of $g^{D,\mcp_{A_t}}(0,d,\Bar{w}, q)$, then it also witnesses \hcness of $h^{D,\mcp_{A_t}}(0,d,\Bar{w}, q)$ and by \hcness of the function $h^{D, \mcp_{A_t}}$, we should have that $h^{D,\mcp_{A_t}}(0,d,\Bar{w}, q)=1$, which in turn implies that we have indeed set $g^{D,\mcp_{A_t}}(0,d,\Bar{w}, q)=1$. 
Moreover, if $g^{D,\mcp_{A_t}}(0,d,\Bar{w}, q)=1$, then $h^{D,\mcp_{A_t}}(0,d,\Bar{w}, q)=1$ and by \hsness of the function $h^{D, \mcp_{A_t}}$, there exists a $k$-subpartition of $V(G_{\le}(0))$ that witnesses \hsness of $h^{D,\mcp_{A_t}}(0,d,\Bar{w}, q)$ which also witnesses \gsness of $g^{D,\mcp_{A_t}}(0,d,\Bar{w}, q)=1$

Now, we will assume that for some $\ell\in [p]$, we have assigned values to $g^{D,\mcp_{A_t}}(\ell-1,d,\Bar{w}, q)$ for every $d\in\krange$, $\bar{w}\in\lrange^k$ and $q\in[k]$ and describe an algorithm to assign values to $g^{D,\mcp_{A_t}}(\ell,d,\Bar{w}, q)$ for every $d\in\krange$, $\bar{w}\in\lrange^k$ and $q\in[k]$. 


Let 
\begin{align*}
    \mathcal{Y}^\ell:=\big\{&((\ell-1,d_1,\Bar{y},q_1),(\ell,d_2,\Bar{z},q_2)):
    \\&\qquad d_1,d_2\in\krange,\bar{y},\bar{z}\in\lrange^k,q_1,q_2\in[k],
    \\&\qquad g^{D,\mcp_{A_t}}(\ell-1,d_1,\Bar{y},q_1)=1,h^{D,\mcp_{A_t}}(\ell,d_2,\Bar{z},q_2)=1\big\}.
\end{align*}
Here, $\mathcal{Y}^\ell$ is the collection of pairs of inputs to $g^{D,\mcp_{A_t}}$ and $h^{D,\mcp_{A_t}}$ such that $g^{D,\mcp_{A_t}}$ evaluates to $1$ on the first input and $h^{D,\mcp_{A_t}}$ evaluates to $1$ on the second input. 
For each pair $((\ell-1,d_1,\Bar{y},q_1),(\ell,d_2,\Bar{z},q_2))\in\mathcal{Y}^\ell$ and each pair of permutations $\sigma_1, \sigma_2:[k]\rightarrow [k]$ of permutations (i.e., bijections) satisfying the following conditions, 
\begin{enumerate}
    \item [(G1)] $\sigma_1(q_1)=\sigma_2(q_2)$, 
    \item [(G2)] If $A_t\subseteq O\cup P_\ell$, then $\sigma_2(i)=i$ for all $i\in[k']$.
    \item [(G3)] If $A_t\subseteq O\cup P_{\leq \ell-1}$, then $\sigma_1(i)=i$ for all $i\in[k']$.
\end{enumerate}
our algorithm will set 
$g^{D,\mcp_{A_t}}(\ell,d_3,\sigma_1(\Bar{y})+\sigma_2(\Bar{z}),\sigma_1(q_1))=1$ for all $d_3\geq d_1+d_2$. Here the notation $\sigma_1(\bar{y})$ refers to the $k$-dimensional vector whose $i$th entry is $y_{\sigma_1^{-1}(i)}$. 
Finally, we set $g^{D,\mcp_{A_t}}(\ell, d, \Bar{w}, q)=0$ for all $d\in \set{0,1,\ldots, 2k^2}$, $\bar{w}\in\lrange^k$ and $q\in[k]$ for which the algorithm has not set the $g^{D,\mcp_{A_t}}$ value so far. 
This algorithm is described in Algorithm \ref{algo updating g values}. 

\begin{algorithm}[ht]
\begin{algorithmic}
\STATE Set $g^{D,\mcp_{A_t}}(0,d,\bar{w},q)=h^{D,\mcp_{A_t}}(0,d,\bar{w},q)$ for every $d\in\krange$, $\bar{w}\in\lrange^k$ and $q\in[k]$.
\FORALL{$\ell=1,\ldots,p$}
    \FORALL{pair $((\ell-1,d_1,\Bar{y},q_1),(\ell,d_2,\Bar{z},q_2))\in\mathcal{Y}^\ell$}
        \FORALL{permutation pairs $(\sigma_1,\sigma_2)$ satisfying (G1), (G2) and (G3)}
            \STATE Set $g^{D,\mcp_{A_t}}(\ell,d_3,\sigma_1(\Bar{y})+\sigma_2(\Bar{z}),\sigma_1(q_1))=1$ for all $d_3\geq d_1+d_2$.
        \ENDFOR
    \ENDFOR
    \STATE For $d_3\in\krange$, $\bar{w}\in\lrange^k$ and $q\in[k]$ such that $g^{D,\mcp_{A_t}}(\ell,d_3,\bar{w},q)$ is not yet set to 1, set $g^{D,\mcp_{A_t}}(\ell,d_3,\bar{w},q)=0$.
\ENDFOR
\end{algorithmic}
\caption{Computing function $g^{D,\mcp_{A_t}}$}\label{algo updating g values}
\end{algorithm}

We first bound the run-time of the algorithm. The size of $\mathcal{Y}^{\ell}$ is $O(k^6)\lambda^{2k}$ for every $\ell\in [p]$. We recall that $p\le n$. Thus, the run-time of the algorithm is 
\[
p\cdot O(k^6)\lambda^{2k}\cdot(k!)^2 O(k^2) = (\lambda k)^{O(k^2)}n. 
\]
We now prove the correctness of the algorithm by induction on $\ell$. We recall that we have already proved the base case. We now prove the induction step. 

By induction hypothesis, if there exists a $k$-subpartition of $V(G_{\le}(\ell-1))$ witnessing \gcness of $g^{D, \mcp_{A_t}}(\ell-1, d, \Bar{w}, q)$ for some $d\in\krange$, $\bar{w}\in\lrange^k$ and $q\in[k]$, then $g^{D, \mcp_{A_t}}(\ell-1, d, \Bar{w}, q)=1$. Furthermore, if $g^{D, \mcp_{A_t}}(\ell-1, d, \Bar{w}, q)=1$ for some $d\in\krange$, $\bar{w}\in\lrange^k$ and $q\in[k]$, then there exists a $k$-subpartition of $V(G_{\le}(\ell-1))$  witnessing \gsness of $g^{D, \mcp_{A_t}}(\ell-1, d, \Bar{w}, q)$. 

Suppose that there exists a $k$-subpartition $\mcq^0=(Q^0_1,\ldots,Q^0_k)$ of $V(G_\leq(\ell))$ witnessing \gcness of $g^{D,\mcp_{A_t}}(\ell, d, \Bar{w}, q)$ for some $d\in\krange$, $\bar{w}\in\lrange^k$ and $q\in[k]$. We now show that the algorithm will correctly set $g^{D,\mcp_{A_t}}(\ell, d, \Bar{w}, q)=1$.

Let $\mcq^1:=(Q^1_1,\ldots,Q^1_k)$ and $\mcq^2:=(Q^2_1,\ldots,Q^2_k)$ be the restriction of $\mcq^0$ to the vertices of $G_\leq(\ell-1)$ and $G(\ell)$, respectively, given by $Q^1_i:=Q^0_i\cap V(G_\leq(\ell-1))$ and $Q^2_i:=Q^0_i\cap V(G(\ell))$ for every $i\in[k]$.
It follows that $\mcq^1$ and $\mcq^2$ are both restrictions of $\Omega$.
We note that $\mcq^1$ witnesses \gcness of $g^{D,\mcp_{A_t}}(\ell-1,d_1,\bar{y},q)$, where $d_1:=|\delta_{T}(\mcq^1)|$ and $y_i:=|\delta_{G_\leq(\ell-1)}(Q^1_i)|$ for all $i\in[k]$.
Furthermore, $\mcq^2$ witnesses \hcness of $h^{D,\mcp_{A_t}}(\ell,d_2,\bar{z},q)$ where $d_2:=|\delta_{T}(\mcq^2)|$ and $z_i:=|\delta_{G(\ell)}(Q^1_i)|$ for all $i\in[k]$. We note that $d_1, d_2\in \krange$ and $\Bar{y}, \Bar{z}\in \lrange^k$ since the number of crossing edges in the subgraph is at most the number of crossing edges in the graph $G_{\le}(\ell)$.

Consequently, the pair $((\ell-1,d_1,\bar{y},q),(\ell,d_2,\bar{z},q))$ is present in $\mathcal{Y}^\ell$. We now consider the case when $\sigma_1$ and $\sigma_2$ are both the identity permutation on $[k]$, i.e., $\sigma_1(i)=i=\sigma_2(i)$ for every $i\in [k]$. These two permutations satisfy the conditions (G1), (G2), and (G3). Moreover, since there are no edges between any two distinct parts among $P_1, \ldots, P_{\ell}$ due to the nice decomposition property, we know that $d_1+d_2\le d$ and $y_i+z_i=w_i$ for each $i\in [k]$. This also implies that $\sigma_1(\Bar{y})+\sigma_2(\Bar{z})=\Bar{w}$. We also have that $\sigma_1(q)=q$. Consequently, the algorithm will set $g^{D,\mcp_{A_t}}(\ell, d, \Bar{w}, q)$ to be $1$. 

Next suppose that $((\ell-1,d_1,\Bar{y},q_1),(\ell,d_2,\Bar{z},q_2))\in\mathcal{Y}^\ell$ and let $\sigma_1, \sigma_2:[k]\rightarrow [k]$ be permutations satisfying conditions (G1), (G2), and (G3). Then, we will exhibit a $k$-subpartition $\mcq^0$ of the vertices of $G_{\le}(\ell)$ that witnesses \gsness of $g^{D,\mcp_{A_t}}(\ell,d_3,\sigma_1(\Bar{y})+\sigma_2(\Bar{z}),\sigma_1(q_1))$ for every $d_3\ge d_1 + d_2$.
Let $\mcq^1:=(Q^1_1, \ldots, Q^1_k)$ and $\mcq^2:=(Q^2_1, \ldots, Q^2_k)$ be a $k$-subpartition of the vertices of $G_{\le}(\ell-1)$ and $G(\ell)$ respectively, that witnesses \gsness of $g^{D,\mcp_{A_t}}(\ell-1,d_1,\Bar{y},q_1)$ and \hsness of $h^{D,\mcp_{A_t}}(\ell,d_1,\Bar{z},q_2)$, respectively. 

Consider the $k$-subpartition $\mcq^0:=(Q^0_1, \ldots, Q^0_k)$ of the vertices of $G_\leq(\ell)$ obtained by setting $Q^0_i:=Q^1_{\sigma_1^{-1}(i)}\cup Q^2_{\sigma_2^{-1}(i)}$. Let $d_3\ge d_1+d_2$.
We will show that $\mcq^0$ witnesses \gsness of $g^{D,\mcp_{A_t}}(\ell,d_3,\sigma_1(\Bar{y})+\sigma_2(\Bar{z}),\sigma_1(q_1))$.

Since $\sigma_1(q_1)=\sigma_2(q_2)$, we know that the parts containing $O$ in $\mcq^1$ and $\mcq^2$, i.e. $Q^1_{q_1}$ and $Q^2_{q_2}$, are both in $Q^0_{\sigma_1(q_1)}$, thus proving condition (iv) needed to witness \gsness of $g^{D,\mcp_{A_t}}(\ell,d_3,\sigma_1(\Bar{y})+\sigma_2(\Bar{z}),\sigma_1(q_1))$.
As a consequence of the fact that $O\subseteq Q^1_{q_1}$ and $O\subseteq Q^2_{q_2}$, we obtain that $\mcq^0$ is indeed a $k$-subpartition of the vertices of $G_{\le} (\ell)$. Furthermore, the $k$-subpartition $\mcq^0$ restricted to $O\cup P_{\le \ell}$ is a coarsening of $\mcq_{\chi(t)}$ restricted to $O\cup P_{\le \ell}$, thus proving condition (i) needed to witness \gsness of $g^{D,\mcp_{A_t}}(\ell,d_3,\sigma_1(\Bar{y})+\sigma_2(\Bar{z}),\sigma_1(q_1))$. By the nice decomposition property, there are no edges between two distinct parts among $P_1,\ldots,P_\ell$. Hence, 
\begin{align*}
    |\delta_T(\mcq^0)|&=|\delta_T(\mcq^1)|+|\delta_T(\mcq^2)|\leq d_1+d_2\text{, and}
    \\|\delta_{G_\leq(\ell)}(Q^0_i)|&=|\delta_{G_\leq(\ell-1)}(Q^1_{\sigma_1^{-1}(i)})|+|\delta_{G(\ell)}(Q^2_{\sigma_2^{-1}(i)})|
    =y_{\sigma_1^{-1}(i)}+z_{\sigma_2^{-1}(i)}\text{ for all }i\in[k],
\end{align*}
thus proving conditions (ii) and (iii) needed to witness \gsness of $g^{D,\mcp_{A_t}}(\ell,d_3,\sigma_1(\Bar{y})+\sigma_2(\Bar{z}),\sigma_1(q_1))$.

Suppose that $A_t\subseteq O\cup P_\ell$. Then, we know that $\sigma_2(i)=i$ for all $i\in[k']$. We also know that $Q^2_i\cap A_t=\Tilde{P}_i$ for all $i\in[k']$. Therefore, $\Tilde{P}_i\subseteq Q^2_i\subseteq Q^0_{\sigma_2(i)}=Q^0_i$ for all $i\in[k']$. Next, suppose that $A_t\subseteq O\cup P_{\leq\ell}$. Then, we know that $\sigma_1(i)=i$ and $Q^1_i\cap A_t=\Tilde{P}_i$ for all $i\in[k']$. Therefore, $\Tilde{P}_i\subseteq Q^1_i\subseteq Q^0_{\sigma_1(i)}=Q^0_i$ for all $i\in[k']$. Hence, condition (v) needed to witness \gsness of $g^{D,\mcp_{A_t}}(\ell,d_3,\sigma_1(\Bar{y})+\sigma_2(\Bar{z}),\sigma_1(q_1))$ also holds. This shows that $\mcq^0$ witnesses \gsness of $g^{D,\mcp_{A_t}}(\ell,d_3,\sigma_1(\bar{y})+\sigma_2(\bar{z}),\sigma_1(q_1))$ for all $d_3\geq d_1+d_2$.


\end{proof}


\subsubsection{Computing $h^{D,\mcp_{A_t}}$ in leaf nodes of the tree decomposition}
In this section we will describe an algorithm to compute an \hc and \hs function $h^{D,\mcp_{A_t}}:\prange\times\krange\times\lrange^k\times[k]\to\{0,1\}$ when $t\in V(\tau)$ is a leaf node of $\tau$. This corresponds to the base case of our dynamic program.

\begin{lemma}\label{lemma:h in leaf nodes}
If $t\in V(\tau)$ is a leaf node of $\tau$, then there exists an algorithm that takes as input 
a partition $\mcp_{A_t}\in\mathcal{F}^{A_t}$, 
a nice decomposition  $D=(\mcp_{\chi(t)}=(P_1,\ldots,P_p,O),\mcq_{\chi(t)},O)\in\mathcal{D}$, 
and runs in time $(\lambda k)^{O(k^2)}n^{O(1)}$
to return a function $h^{D,\mcp_{A_t}}:\prange\times\krange\times\lrange^k\times[k]\to\{0,1\}$ that is \hc and \hs.
\end{lemma}

\begin{proof}
Let the input be fixed as in the lemma. We will iteratively assign values to $h^{D, \mcp_{A_t}}(\ell, d, \Bar{w}, q)$ for $\ell=0,1,\ldots, p$. 

Let $\ell \in \set{0, 1, \ldots, p}$. 
By the definition of nice decomposition, we know that the part $P_\ell$ contains $O(k^2)$ parts of $\mcq_{\chi(t)}$. Hence, $O\cup P_\ell$ contains $O(k^2)$ parts from $\mcq_{\chi(t)}$. Hence, we can enumerate all $k^{O(k^2)}$ $k$-subpartitions of $O\cup P_\ell$ that coarsen $\mcq_{\chi(t)}$ and explicitly verify if one of them satisfies the required conditions to witness \hsness of $h^{D,\mcp_{A_t}}(\ell,d,\Bar{y}, q)$. If so, then we set the corresponding $h^{D,\mcp_{A_t}}(l,d,\Bar{y}, q)=1$ and otherwise set $h^{D,\mcp_{A_t}}(l,d,\Bar{y},q)=0$.
Thus, the time to compute $h^{D,\mcp_{A_t}}(\ell,d,\Bar{y}, q)$ for all $d\in\krange$, $\bar{y}\in\lrange^k$ and $q\in[k]$ is 
\[
(2k^2+1)\cdot\lambda^{k}\cdot k\cdot k^{O(k^2)}n^{O(1)} = (\lambda k)^{O(k^2)}n^{O(1)}.
\]

In order to compute $h^{D,\mcp_{A_t}}(\ell,d,\Bar{y},q)$ for all $\ell\in \set{0,1,\ldots, p}$, $d\in\krange$, $\bar{y}\in\lrange^k$ and $q\in[k]$, the total time required is $p(\lambda k)^{O(k^2)}n^{O(1)}=(\lambda k)^{O(k^2)}n^{O(1)}$, since $p\leq n$.
The resulting function $h^{D,\mcp_{A_t}}$ is \hs as well as \hc by construction. 
\end{proof}

\subsubsection{Computing $h^{D,\mcp_{A_t}}$ in non-leaf nodes of the tree decomposition}
In this section, we will describe an algorithm to compute a function $h^{D,\mcp_{A_t}}:\prange\times\krange\times\lrange^k\times[k]\to\{0,1\}$ that is \hc and \hs when $t$ is a non-leaf node of the tree decomposition $\tau$.

\begin{lemma}\label{lemma:h in non-leaf node}
There exists an algorithm that takes as input 
$(\tau, \chi)$, 
a non-leaf tree node $t\in V(\tau)$, 
a partition $\mcp_{A_t}\in\mathcal{F}^{A_t}$, 
a nice decomposition $D=(\mcp_{\chi(t)}=(P_1,\ldots,P_p,O),\mcq_{\chi(t)},O)\in\mathcal{D}$, 
together with Boolean functions $f_{t'}:\mathcal{F}^{A_{t'}}\times\lrange^k\times\krange\to\{0,1\}$ for each child $t'$ of $t$ that are \fc and \fs, 
and runs in time $(\lambda k)^{O(k^2)}n^{O(1)}$
to return a function $h^{D,\mcp_{A_t}}:\prange\times\krange\times\lrange^k\times[k]\to\{0,1\}$ that is \hc and \hs.
\end{lemma}

\begin{proof}
Given the inputs and an integer $\ell\in\prange$, let us define $\Tilde{\mathcal{R}}(D,\mcp_{A_t},\ell)$ to be the set of $k$-subpartitions $\mcp=(P'_1,\ldots,P'_k)$ of $O\cup P_\ell$ that satisfy the following conditions:
\begin{enumerate}[(i)]
    \item $\mcp$ coarsens $\mcq_{\chi(t)}$ restricted to $O\cup P_{\ell}$,
    \item $|\delta_T(\mcp)|\leq 2k^2$, and
    \item if $A_t\subseteq O\cup P_\ell$, then $P'_i\cap A_t=\tilde{P}_i$ for all $i\in[k']$.
\end{enumerate}
Since every $k$-subpartition in $\Tilde{\mathcal{R}}(D,\mcp_{A_t},\ell)$ necessarily coarsens $\mathcal{Q}_{\chi(t)}$, we have the size bound $|\Tilde{\mathcal{R}}(D,\mcp_{A_t},\ell)|= k^{O(k^2)}$.

In order to compute an \hc and \hs function $h^{D, \mcp_{A_t}}$, we will employ a new sub-problem that we define below.
\begin{definition}
Let $\mu^{D,\mcp_{A_t},\ell}:\{0,1,...,2k^2\}\times\lrange^k\times[k]\times\Tilde{\mathcal{R}}(D,\mcp_{A_t},\ell)\to\{0,1\}$ be a Boolean function, where $D=(\mcp_{\ct}=(P_1,\ldots,P_p,O),\mcq_{\ct},O)\in\mathcal{D}$, $\mcp_{A_t}=(\tilde{P}_1, \ldots, \tilde{P}_{k'})\in\mathcal{F}^{A_t}$ and $\ell\in\prange$.
\begin{enumerate}
    \item
    \textbf{(Correctness)}
    The function $\mu^{D,\mcp_{A_t},\ell}$ is \emph{\hpc} if we have $\mu^{D,\mcp_{A_t},\ell}(d,\Bar{y},q,\mathcal{R})=1$ for all $d\in\krange$, $\bar{y}\in\lrange^k$, $q\in[k]$ and $\mathcal{R}\in\Tilde{\mathcal{R}}(D,\mcp_{A_t},\ell)$ for which there exists a $k$-subpartition $\mathcal{P}=(P'_1,\ldots,P'_k)$ of $V(G(\ell))$ satisfying the following conditions:
    \begin{enumerate}[(i)]
        \item $\mathcal{P}$ restricted to $O\cup P_\ell$ is $\mathcal{R}$,
        \item $|\delta_{G(\ell)}(P'_i)|= y_i$ for all $i\in [k]$,
        \item $|\delta_T(\mcp)|\leq d$, 
        \item $O\subseteq P'_q$,
        \item if $A_t\subseteq O\cup P_\ell$, then $P'_i\cap A_t=\Tilde{P}_i$ for all $i\in[k']$, and
        \item $\mcp$ is a restriction of $\Omega$ to $V(G(\ell))$.
    \end{enumerate}
    A $k$-subpartition of $V(G(\ell))$ satisfying the above six conditions is said to \emph{witness} \hpcness of $\mu^{D,\mcp_{A_t},\ell}(d,\Bar{y},q,\mathcal{R})$.
    \item 
    \textbf{(Soundness)} The function $\mu^{D,\mcp_{A_t},\ell}$ is \emph{\hps} if for all $d\in\krange$, $\bar{y}\in\lrange^k$, $q\in[k]$ and $\mathcal{R}\in\Tilde{\mathcal{R}}(D,\mcp_{A_t},\ell)$, we have $\mu^{D,\mcp_{A_t},\ell}(d,\Bar{y},q,\mathcal{R})=1$ only if there exists a $k$-subpartition $\mathcal{P}=(P'_1,\ldots,P'_k)$ of $V(G(\ell))$ satisfying conditions (i), (ii), (iii), (iv) and (v) above. A $k$-subpartition of $V(G(\ell))$ satisfying (i), (ii), (iii), (iv) and (v) is said to \emph{witness} \hpsness of $\mu^{D,\mcp_{A_t},\ell}(d,\Bar{y},q,\mathcal{R})$.
\end{enumerate}
\end{definition}

A function $\mu^{D,\mcp_{A_t},\ell}$ that is \hpc and \hps helps compute a function $h^{D,\mcp_{A_t}}$ that is \hc and \hs by the following proposition.
\begin{proposition}\label{prop:mu to h}
Suppose that we have functions $\mu^{D,\mcp_{A_t},\ell}:\{0,1,...,2k^2\}\times\lrange^k\times[k]\times\Tilde{\mathcal{R}}(D,\mcp_{A_t},\ell)\to\{0,1\}$ for all $\ell\in\prange$ such that all of them are \hpc and \hps. Then the function $h^{D,\mcp_{A_t}}:\{0,1,\ldots,p\}\times\{0,1,\ldots,2k^2\}\times\lrange^k\times[k]\to\{0,1\}$ obtained by setting 
\[h^{D,\mcp_{A_t}}(\ell,d,\Bar{y},q):=\max\big\{\mu^{D,\mcp_{A_t},\ell}(d,\Bar{y},q,\mathcal{R}):\mathcal{R}\in\Tilde{\mathcal{R}}(D,\mcp_{A_t},\ell)\big\}\]
for every $\ell\in\prange$, $d\in\krange$, $\bar{y}\in\lrange^k$ and $q\in[k]$ is both \hc and \hs.
\end{proposition}
\begin{proof}
We first show \hcness. For $\ell\in\prange$, $d\in\krange$, $\bar{y}\in\lrange^k$ and $q\in[k]$, suppose that there exists a $k$-subpartition $\Pi=(\pi_1,\ldots,\pi_k)$ of $V(G(\ell))$ withnessing \hcness of $h^{D,\mcp_{A_t}}(\ell,d,\Bar{y},q)$. Then, $\Pi$ also witnesses \hpcness of $\mu^{D,\mcp_{A_t},\ell}(d,\bar{y},q,\Pi^0)$, where $\Pi^0=(\pi_1\cap (O\cup P_\ell),\ldots,\pi_k\cap (O\cup P_\ell))\in\tilde{\mcr}(D,\mcp_{A_t},\ell)$. Since the function $\mu^{D,\mcp_{A_t},\ell}$ is \hpc, we know that $\mu^{D,\mcp_{A_t},\ell}(d,\bar{y},q,\Pi^0)=1$. This implies that $h^{D,\mcp_{A_t}}(\ell,d,\Bar{y},q)=1$.


Next, we show \hsness. Suppose that $h^{D,\mcp_{A_t}}(\ell,d,\Bar{y},q)=1$ for some $\ell\in\prange$, $d\in\krange$, $\bar{y}\in\lrange^k$ and $q\in[k]$. Then, there exists $\mcr'\in\tilde{\mcr}(D,\mcp_{A_t},\ell)$ such that  $\mu^{D,\mcp_{A_t},\ell}(d,\bar{y},q,\mcr')=1$. Since the function $\mu^{D,\mcp_{A_t},\ell}$ is \hps, we know that there exists some $k$-subpartition $\Pi'$ of $V(G(\ell))$ witnessing \hpsness of $\mu^{D,\mcp_{A_t},\ell}(d,\bar{y},q,\mcr')$. It follows that $\Pi'$ also witnesses \hsness of $h^{D,\mcp_{A_t}}(\ell,d,\Bar{y},q)$.

\end{proof}

By the above proposition, it suffices to assign values to $\mu^{D,\mcp_{A_t},\ell}(d,\Bar{y},q,\mathcal{R})$ for every $\ell\in\prange$, $\mathcal{R}\in \tilde{\mcr}(D,\mcp_{A_t},\ell)$, $d\in\krange$ $\Bar{y}\in \lrange^k$, and $q\in[k]$ so that the resulting function is \hpc and \hps. 
We define another sub-problem.

\begin{definition}
Let $\nu^{D,\mcp_{A_t},\ell,\mcr}:\krange\times\{0,1,\ldots,|\mathcal{A}(\ell)|\}\times\lrange^k\times[k]\to\{0,1\}$ be a Boolean function, where $D=(\mcp_{\ct}=(P_1,\ldots,P_p,O),\mcq_{\ct},O)\in\mathcal{D}$, $\mcp_{A_t}=(\tilde{P}_1, \ldots, \tilde{P}_{k'})\in\mathcal{F}^{A_t}$, $\ell\in\prange$ and $\mcr=(R_1,\ldots,R_k)\in\tilde{\mcr}(D,\mcp_{A_t},\ell)$.
\begin{enumerate}
    \item
    \textbf{(Correctness)}
    The function $\nu^{D,\mcp_{A_t},\ell,\mcr}$ is \emph{\nuc} if we have $\nu^{D,\mcp_{A_t},\ell,\mcr}(d,a,\Bar{z},q)=1$ for all $d\in\krange$, $a\in\{0,1,...,|\mathcal{A}(\ell)|\}$, $\bar{z}\in\lrange^k$ and $q\in[k]$ for which there exists a $k$-subpartition $\mathcal{P}=(P'_1,\ldots,P'_k)$ of $V(G_\leq(\ell,a))$ satisfying the following conditions:
    \begin{enumerate}[(i)]
        \item $P'_i\cap (O\cup P_\ell)=R_i$ for all $i\in [k]$,
        \item $|\delta_{G_\leq(\ell,a)}(P'_i)|= z_i$ for all $i\in [k]$,
        \item $|\delta_T(\mcp)|\leq d$,
        \item $O\subseteq P'_q$, and
        \item $\mcp$ is a restriction of $\Omega$ to $V(G_\leq(\ell,a))$.
    \end{enumerate}
    A $k$-subpartition of $V(G_\leq(\ell,a))$ satisfying the above five conditions is said to \emph{witness} \nucness of $\nu^{D,\mcp_{A_t},\ell,\mcr}(d,a,\Bar{z},q)$.
    \item 
    \textbf{(Soundness)} The function $\nu^{D,\mcp_{A_t},\ell,\mcr}$ is \emph{\nus} if for all $d\in\krange$, $a\in\{0,1,...,|\mathcal{A}(\ell)|\}$, $\bar{z}\in\lrange^k$ and $q\in[k]$, we have $\nu^{D,\mcp_{A_t},\ell,\mcr}(d,a,\Bar{z},q)=1$ only if there exists a $k$-subpartition $\mathcal{P}=(P'_1,\ldots,P'_k)$ of $V(G_\leq(\ell))$ satisfying conditions (i), (ii), (iii) and (iv) above. A $k$-subpartition of $V(G_\leq(\ell,a))$ satisfying (i), (ii), (iii) and (iv) is said to \emph{witness} \nusness of $\nu^{D,\mcp_{A_t},\ell,\mcr}(d,a,\Bar{z},q)$.
\end{enumerate}
\end{definition}
A function $\nu^{D,\mcp_{A_t},\ell,\mcr}:\krange\times\{0,1,...,|\mathcal{A}(\ell)|\}\times\lrange^k\times[k]\to\{0,1\}$ that is \nuc and \nus helps compute a function $\mu^{D,\mcp_{A_t},\ell}:\{0,1,...,2k^2\}\times\lrange^k\times[k]\times\Tilde{\mathcal{R}}(D,\mcp_{A_t},\ell)\to\{0,1\}$ that is \hpc and \hps by the following observation. 
\begin{observation}\label{obs:nu to mu}
Suppose that we have functions $\nu^{D,\mcp_{A_t},\ell,\mcr}:\krange\times\{0,1,\ldots,|\mathcal{A}(\ell)|\}\times\lrange^k\times[k]\to\{0,1\}$ for every $\mcr\in\tilde{\mcr}(D,\mcp_{A_t},\ell)$ such that all of them are \nuc and \nus. Then the function $\mu^{D,\mcp_{A_t},\ell}:\{0,1,...,2k^2\}\times\lrange^k\times[k]\times\Tilde{\mathcal{R}}(D,\mcp_{A_t},\ell)\to\{0,1\}$ obtained by setting
\[\mu^{D,\mcp_{A_t},\ell}(d,\Bar{y},q,\mathcal{R}):=\nu^{D,\mcp_{A_t},\ell,\mcr}(d,|\mathcal{A}(\ell)|,\Bar{y},q)\]
for every $d\in\krange$, $\bar{y}\in\lrange^k$, $q\in[k]$ and $\mcr\in\tilde{\mcr}(D,\mcp_{A_t},\ell)$ is both \hpc and \hps.
\end{observation}

We note that a function $\nu^{D,\mcp_{A_t},\ell,\mcr}:\krange\times\{0,1,\ldots,|\mathcal{A}(\ell)|\}\times\lrange^k\times[k]\to\{0,1\}$ that is \nuc and \nus can be computed in $(\lambda k)^{O(k^2)}$ time using Lemma \ref{lemma:compute nu}, which we state and prove after this proof.

Our algorithm to prove Lemma \ref{lemma:h in non-leaf node} starts by computing $\Tilde{\mathcal{R}}(D,\mcp_{A_t},\ell)$ for every $\ell\in\prange$, which can be done in $k^{O(k^2)}n$ time  (since the size of $\Tilde{\mathcal{R}}(D,\mcp_{A_t},\ell)$ is $k^{O(k^2)}$ for every $\ell\in\prange$). For each $\ell\in\prange$ and $\mcr\in\tilde{\mcr}(D,\mcp_{A_t},\ell)$, our algorithm assigns values to $\nu^{D,\mcp_{A_t},\ell,\mcr}(d,|\mathcal{A}(\ell)|,\Bar{y},q)$ for all $d\in\krange$, $\bar{y}\in\lrange^k$ and $q\in[k]$ using Lemma \ref{lemma:compute nu}. The algorithms uses these values to next assign values to $\mu^{D,\mcp_{A_t},\ell}(d,\Bar{y},q,\mathcal{R})$ for all $\ell\in\prange$, $d\in\krange$, $\bar{y}\in\lrange^k$, $q\in[k]$ and $\mcr\in\tilde{\mcr}(D,\mcp_{A_t},\ell)$ using Observation \ref{obs:nu to mu}. 
Finally, the algorithm uses these values to assign values to $h^{D,\mcp_{A_t}}(\ell,d,\bar{y},q)$ for all $\ell\in\prange$, $d\in\krange$, $\bar{y}\in\lrange^k$ and $q\in[k]$ using Proposition \ref{prop:mu to h}. The resulting function $h^{D, \mcp_{A_t}}$ is \hc and \hs. 

The total time to compute $h^{D,\mcp_{A_t}}(\ell,d,\bar{y},q)$ for all $\ell\in\prange$, $d\in\krange$, $\bar{y}\in\lrange^k$ and $q\in[k]$ is
\[(\lambda k)^{O(k^2)}n^{O(1)}+(p+1)\cdot|\Tilde{\mathcal{R}}(D,\mcp_{A_t},\ell)|\cdot(\lambda k)^{O(k^2)}n^{O(1)}=(\lambda k)^{O(k^2)}n^{O(1)}.\]
This completes the proof of Lemma \ref{lemma:h in non-leaf node}.
\end{proof}

\begin{lemma}\label{lemma:compute nu}
There exists an algorithm that takes as input 
$(\tau, \chi)$, 
a non-leaf tree node $t\in V(\tau)$, 
a partition $\mcp_{A_t}\in\mathcal{F}^{A_t}$,
a nice decomposition $D=(\mcp_{\ct}=(P_1,\ldots,P_p,O),\mcq_{\ct},O)\in \mcd$, 
an integer $\ell\in\prange$, 
a $k$-subpartition $\mcr=(R_1,\ldots,R_k)\in\tilde{\mcr}(D,\mcp_{A_t},\ell)$, 
together with Boolean functions  $f_{t'}:\mathcal{F}^{A_{t'}}\times\lrange^k\times\krange\to\{0,1\}$ for each child $t'$ of $t$ that are \fc and \fs, 
and runs in $(\lambda k)^{O(k^2)}n^{O(1)}$ time 
to return a function $\nu^{D,\mcp_{A_t},\ell,\mcr}:\krange\times\{0,1,...,|\mathcal{A}(\ell)|\}\times\lrange^k\times[k]\to\{0,1\}$ that is \nuc and \nus.
\end{lemma}

\begin{proof}

Let the input be as stated in the lemma. We will iteratively assign values to $\nu^{D,\mcp_{A_t},\ell,\mcr}(d,a,\bar{z},q)$ for $a=0,1,\ldots,p$.

For $a=0$, we observe that in order to assign values to $\nu^{D,\mcp_{A_t},\ell,\mcr}(d,0,\bar{z},q)$ for all $d\in\krange$, $\bar{z}\in\lrange^k$ and $q\in[k]$, the $k$-subpartition $\mcr$ is the only $k$-subpartition of $V(G_\leq(\ell,0))=O\cup P_l$ whose restriction to $O\cup P_\ell$ is $\mathcal{R}$ (i.e., it is the only $k$-subpartition that can satisfy condition (i) in the definition of the $\nu^{D,\mcp_{A_t},\ell,\mcr}$ sub-problem). Therefore, we set $\nu^{D,\mcp_{A_t},\ell,\mcr}(d,0,\bar{z},q)=1$ if and only if $\mcr$ satisfies the remaining \nusness conditions, which can be verified in $n^{O(1)}$ time. The run-time to assign values to $\nu^{D,\mcp_{A_t},\ell,\mcr}(d,0,\bar{z},q)$ for all $d\in\krange$, $\bar{z}\in\lrange^k$ and $q\in[k]$ is $\lambda^{O(k)} k^{O(1)}n^{O(1)}$. If $\mcr$ witnesses \nucness of $\nu^{D,\mcp_{A_t},\ell,\mcr}(d,0,\bar{z},q)$ for some $d\in\krange$, $\bar{z}\in\lrange^k$ and $q\in[k]$, then our algorithm sets $\nu^{D,\mcp_{A_t},\ell,\mcr}(d,0,\bar{z},q)=1$. If our algorithm sets $\nu^{D,\mcp_{A_t},\ell,\mcr}(d,0,\bar{z},q)=1$ for some $d\in\krange$, $\bar{z}\in\lrange^k$ and $q\in[k]$, then $\mcr$ witnesses \nusness of $\nu^{D,\mcp_{A_t},\ell,\mcr}(d,0,\bar{z},q)$.

Now, we will assume that for some $a\in [|\mathcal{A}(\ell)|]$, we have assigned values to $\nu^{D,\mcp_{A_t},\ell,\mcr}(d,a-1,\bar{z},q)$ for all $d\in\krange$, $\bar{z}\in\lrange^k$ and $q\in[k]$ and describe an algorithm to assign values to $\nu^{D,\mcp_{A_t},\ell,\mcr}(d,a,\bar{z},q)$ for all $d\in\krange$, $\bar{z}\in\lrange^k$ and $q\in[k]$. 

Let $\mathcal{R}_{A_{t_a}}=(R'_1,...,R'_{k_a})$ be the restriction of $\mathcal{R}$ to $A_{t_a}$, where $k_a\leq k$ and $R'_i\neq\emptyset$ for all $i\in [k_a]$ (i.e., $\mcr_{A_{t_a}}$ is a partition with at most $k$ parts). Additionally, in the case $t_a=t$, we order $\mcr_{A_{t_a}}$ so that $\mcr_{A_{t_a}}=\mcp_{A_t}$. Moreover, let $\gamma:[k_a]\to [k]$ be the injection such that $R'_i=R_{\gamma(i)}\cap A_{t_a}$.

We start by defining a set
\begin{align*}
    \mathcal{Z}^a:=\big\{&((d_1,a-1,\Bar{z},q),(\mathcal{R}_{A_{t_a}},\bar{x},d_2)):
    \\&\qquad d_1,d_2\in\krange,\bar{z},\bar{x}\in\lrange^k,q\in[k],
    \\&\qquad \nu^{D,\mcp_{A_t},\ell,\mcr}(d_1,a-1,\Bar{z},q)=1,f_{A_{t_a}}(\mathcal{R}_{A_{t_a}},\bar{x},d_2)=1\big\}.
\end{align*}
For each pair $((d_1,a-1,\Bar{z},q),(\mathcal{R}_{A_{t_a}},\Bar{x},d_2))\in\mathcal{Z}^a$, and each pair of permutations $(\sigma_1,\sigma_2):[k]\to[k]$ satisfying the following two conditions,
\begin{enumerate}
    \item[(N1)] $\sigma_1(i)=\sigma_2(\gamma^{-1}(i))$ for all $i\in[k]$ with $R_i\cap A_{t_a}\neq\emptyset$, and
    \item[(N2)] $\sigma_1(i)=i$ for all $i\in[k]$ for which $R_i\neq\emptyset$,
\end{enumerate}
our algorithm will set $\nu^{D,\mcp_{A_t},\ell,\mcr}(d_3,a,\Bar{w},\sigma_1(q))=1$ for all $d_3\geq d_1+d_2-|\delta_T(\mathcal{R}_{A_{t_a}})|$, where 
\[w_i:= z_{\sigma_1^{-1}(i)}+x_{\sigma_2^{-1}(i)}-\mathds{1}[\sigma_2^{-1}(i)\leq k_a]\cdot|\delta_{G[A_{t_a}]}(R'_{\sigma_2^{-1}(i)})|\]
for all $i\in[k]$.

Finally, we set $\nu^{D,\mcp_{A_t},\ell,\mcr}(d',a,\bar{w},q)=0$ 
for all $d'\in\krange$, $\bar{w}\in\lrange^k$ and $q\in[k]$ for which the algorithm has not set the $\nu^{D,\mcp_{A_t},\ell,\mcr}$ value so far. 
This algorithm is described in Algorithm \ref{algo updating nu values}. 

\begin{algorithm}
\begin{algorithmic}
\STATE 
\FORALL{$d\in\krange$, $\bar{z}\in\lrange^k$ and $q\in[k]$}
    \IF{$\mcr$ witnesses \nusness of $\nu^{D,\mcp_{A_t},\ell,\mcr}(d,0,\bar{z},q)$}
        \STATE Set $\nu^{D,\mcp_{A_t},\ell,\mcr}(d,0,\bar{z},q)=1$.
    \ELSE
        \STATE Set $\nu^{D,\mcp_{A_t},\ell,\mcr}(d,0,\bar{z},q)=0$.
    \ENDIF
\ENDFOR
\FORALL{$a=1,2,\ldots,|\mathcal{A}(\ell)|$}
    \FORALL{pair $((d_1,a-1,\Bar{z},q),(\mathcal{R}_{A_{t_a}},\Bar{x},d_2))\in\mathcal{Z}^a$}
        \FORALL{$(\sigma_1,\sigma_2)$ satisfying (N1) and (N2)}
            \STATE Set $\nu^{D,\mcp_{A_t},\ell,\mcr}(d_3,a,\Bar{w},\sigma_1(q))=1$ for all 
            $d_3\geq d_1+d_2-|\delta_T(\mathcal{R}_{A_{t_a}})|$, 
            \par where $w_i:= z_{\sigma_1^{-1}(i)}+x_{\sigma_2^{-1}(i)}-\mathds{1}[\sigma_2^{-1}(i)\leq k_a]\cdot|\delta_{G[A_{t_a}]}(R'_{\sigma_2^{-1}
        (i)})|$ for all $i\in[k]$.
        \ENDFOR
    \ENDFOR
    \STATE For all $d'\in\krange$, $\bar{w}\in\lrange^k$ and $q\in[k]$ with $\nu^{D,\mcp_{A_t},\ell,\mcr}(d',a,\bar{w},q)$ not yet set to 1, set $\nu^{D,\mcp_{A_t},\ell,\mcr}(d',a,\bar{w},q)=0$.
\ENDFOR
\end{algorithmic}
\caption{Computing function $\nu^{D,\mcp_{A_t},\ell,\mcr}$}\label{algo updating nu values}
\end{algorithm}
We first bound the run-time of the algorithm. The size of $\mathcal{Z}^a$ is $O(k^5)\lambda^{2k}$ for every $a\in[|\mathcal{A}(\ell)|]$. We note that $|\mathcal{A}(\ell)|\leq|V(\tau)|=\poly(n,\lambda,k)$. Thus, the run-time of the algorithm is
\[|\mathcal{A}(\ell)|\cdot O(k^5)\lambda^{2k}\cdot (k!)^2 O(k^2)=(\lambda k)^{O(k^2)}n^{O(1)}.\]

We now prove the correctness of Algorithm \ref{algo updating nu values}. We recall that we have already proved the base case. We now prove the induction step.

By induction hypothesis, if there exists a $k$-subpartition of $V(G_\leq(\ell,a-1))$ witnessing \nucness of $\nu^{D,\mcp_{A_t},\ell,\mcr}(d,a-1,\bar{w},q)$ for some $d\in\krange$, $\bar{w}\in\lrange^k$ and $q\in[k]$, then $\nu^{D,\mcp_{A_t},\ell,\mcr}(d,a-1,\bar{w},q)=1$. Furthermore, if $\nu^{D,\mcp_{A_t},\ell,\mcr}(d,a-1,\bar{w},q)=1$ for some $d\in\krange$, $\bar{w}\in\lrange^k$ and $q\in[k]$, then there exists a $k$-subpartition of $V(G_\leq(\ell,a-1))$ witnessing \nusness of $\nu^{D,\mcp_{A_t},\ell,\mcr}(d,a-1,\bar{w},q)$.

Suppose that there exists a $k$-subpartition $\mcq^0=(Q^0_1,\ldots,Q^0_k)$ of $G_\leq(\ell,a)$ that witnesses \nucness of $\nu^{D,\mcp_{A_t},\ell,\mcr}(d,a,\bar{w},q)$ for some $d\in\krange$, $\bar{w}\in\lrange^k$ and $q\in[k]$. We now show that the algorithm will correctly set $\nu(d,a,\bar{w},q)=1$.

Let $\mcq^1=(Q^1_1,\ldots,Q^1_k)$ be the restriction of $\mcq^0$ to $V(G_\leq(\ell,a))$ given by $Q^1_i:=Q^0_i\cap V(G_{\le}(\ell,a-1))$ for all $i\in[k]$. Let $\mcq^2=(Q^2_1,\ldots,Q^2_k)$ be the restriction of $\mcq^0$ to $V(G_{t_a})$ given by $Q^2_i\cap A_{t_a}=R'_i$ for all $i\in[k_a]$. It follows that $\mcq^1$ and $\mcq^2$ are both restrictions of $\Omega$.
Let $\theta:[k]\to[k]$ be the permutation of $[k]$ such that $Q^2_i=Q^0_{\theta(i)}\cap V(G_{t_a})$ for every $i\in[k]$.

We note that $\mcq^1$ witnesses \nucness of $\nu^{D,\mcp_{A_t},\ell,\mcr}(d_1,a-1,\bar{z},q)=1$, where $d_1:=|\delta_T(\mcq^1)|$ and $z_i:=|\delta_{G_\leq(\ell,a-1)}(Q^1_i)|$ for every $i\in[k]$. Furthermore, $\mcq^2$ witnesses \nucness of $f_{t_a}(\mathcal{R}_{A_{t_a}},\bar{x},d_2)=1$, where $x_i:=|\delta_{G_{t_a}}(\mcq^2_i)|$ for all $i\in[k]$ and $d_2:=|\delta_T(\mcq^2)|$. We note that $d_1,d_2\in\krange$ and $\bar{z},\bar{x}\in\lrange^k$ since the number of crossing edges in the corresponding subgraph is at most the number of crossing edges in the graph $G_\leq(\ell,a)$.

Hence, the pair $((d_1,a-1,\bar{z},q),(\mathcal{R}_{A_{t_a}},\bar{x},d_2))$ is present in $\mathcal{Z}^a$.
Consider the pair of permutations $(\sigma_1,\theta)$, where $\sigma_1$ is the identity permutation on $[k]$, i.e. $\sigma_1(i)=i$ for every $i\in[k]$. 

We will first prove that $(\sigma_1, \theta)$ is an eligible pair of permutations for the algorithm. We note that by definition of $\theta$, for all $i\in[k]$ such that $R_i\cap A_{t_a}\neq\emptyset$, the part $Q^0_{i}$ consists of $Q^2_{\theta^{-1}(i)}$ and $Q^1_i$, and $Q^1_i$ further contains $R_i$ and $R'_{\gamma^{-1}(i)}$. Since each part of $\mcq^0$ intersects at most one part of $\mathcal{R}_{A_{t_a}}$, we know that $R'_{\gamma^{-1}(i)}= Q^0_i\cap V(G_{t_a})= Q^2_{\theta^{-1}(i)}$. By definition of $\mcq^2$, we know that $\gamma^{-1}(i)=\theta^{-1}(i)$. This implies that $\theta(\gamma^{-1}(i))=i=\sigma_1(i)$ for all $i\in[k]$ such that $R_i\cap A_{t_a}\neq\emptyset$, which proves (N1). Since $\sigma_1$ is the identity permutation, it also satisfies (N2).

For this choice of permutations, we will show that our algorithm will set $\nu^{D,\mcp_{A_t},\ell,\mcr}(d,a,\bar{w},q)=1$. By compactness of the tree decomposition $(\tau,\chi)$, there are no edges between any two distinct members among $V(G_{t_1})\backslash A_{t_1},\ldots,V(G_{t_a})\backslash A_{t_a}$. Consequently, our algorithm indeed sets $\nu^{D,\mcp_{A_t},\ell,\mcr}(d,a,\bar{w},q)=1$ due to the following relationships: 
\begin{align*}
    d&\geq|\delta_T(\mcq^0)|= d_1+d_2-|\delta_T(\mathcal{R}_{A_{t_a}})|\text{, and}
    \\w_i&=|\delta_{G_\leq(\ell,a)}(Q^0_i)|
    \\&=|\delta_{G\leq(\ell,a-1)}(Q^1_i)|+|\delta_{G_{t_a}}(Q^2_{\theta^{-1}(i)})|-\mathds{1}[\theta^{-1}(i)\leq k_a]\cdot|\delta_{G[A_{t_a}]}(R'_{\theta^{-1}(i)})|
    \\&=z_i+x_{\theta^{-1}(i)}-\mathds{1}[\theta^{-1}(i)\leq k_a]\cdot|\delta_{G[A_{t_a}]}(R'_{\theta^{-1}(i)})|\quad\text{for all }i\in[k].
\end{align*}

Next suppose that $((d_1,a-1,\Bar{z},q),(\mathcal{R}_{A_{t_a}},\Bar{x},d_2))\in\mathcal{Z}^a$ and let $(\sigma_1,\sigma_2)$ be a pair of permutations satisfying (N1) and (N2). We will exhibit a $k$-subpartition $\mcq^0$ of $V(G_\leq(\ell,a))$ that witnesses \nusness of $\nu(d_3,a,\bar{w},\sigma_1(q))=1$ for all $d_3\ge d_1+d_2$ and $\bar{w}$ as described above. Let $\mcq^1=(Q^1_1,\ldots,Q^1_k)$ and $\mcq^2=(Q^2_1,\ldots,Q^2_k)$ be $k$-subpartitions of $V(G_\leq(\ell,a-1))$ and $V(G_{t_a})$ that witnesses \nusness of $\nu^{D,\mcp_{A_t},\ell,\mcr}(d_1,a-1,\Bar{z},q)$ and \fsness of $f_{t_a}(\mathcal{R}_{A_{t_a}},\Bar{x},d_2)$, respectively. Then, by definition, we have $Q^2_i\cap A_{t_a}=R'_i$ for all $i\in[k_a]$.

Consider the $k$-subpartition $\mcq^0:=(Q^0_1,\ldots,Q^0_k)$ of $V(G_\leq(\ell,a))$ obtained by setting $Q^0_i:=Q^1_{\sigma_1^{-1}(i)}\cup Q^2_{\sigma_2^{-1}(i)}$. Let $d_3\ge d_1+d_2$ and $ \Bar{w}$ be as described above. We will show that $\mcq^0$ witnesses \nusness of $\nu^{D,\mcp_{A_t},\ell,\mcr}(d_3,a,\bar{w},\sigma_1(q))$ by showing that it satisfies the five conditions in the definition of \nusness of $\nu^{D,\mcp_{A_t},\ell,\mcr}(d_3,a,\bar{w},\sigma_1(q))$.

We first show that $\mcq^0$ is indeed a $k$-subpartition of $V(G_\leq(\ell,a))$. It suffices to show that if a part $R_i$ of $\mathcal{R}$ has non-empty intersection with $A_{t_a}$, then the part containing $R_i$ in $\mcq^1$, i.e., the part $Q^1_i$, is in the same part in $\mcq^0$ with the part containing $R_i\cap A_{t_a}=R'_{\gamma^{-1}(i)}$ in $\mcq^2$, i.e. the part $Q^2_{\gamma^{-1}(i)}$. This is equivalent to requiring $\sigma_1(i)=\sigma_2(\gamma^{-1}(i))$ for all $i\in[k]$ with $R_i\cap A_{t_a}\neq\emptyset$, which is satisfied due to condition (N1). Additionally, for all $i\in[k]$ such that $R_i\neq\emptyset$, we have that $Q^0_i\cap (O\cup P_\ell)=Q^1_{\sigma_1^{-1}(i)}\cap (O\cup P_\ell)=R_i$ due to (N2). For all $i\in[k]$ such that $R_i=\emptyset$, we have $Q^0_i\cap (O\cup P_\ell)=Q^1_{\sigma_1^{-1}(i)}\cap (O\cup P_\ell)=\emptyset=R_i$, where  $Q^1_{\sigma_1^{-1}(i)}\cap (O\cup P_\ell)=\emptyset$ follows from condition (i) of $\mcq^1$ witnessing \nusness of $\nu^{D,\mcp_{A_t},\ell,\mcr}(d_1,a-1,\bar{z},q)$. Hence, this implies condition (i) needed to witness \nusness of $\nu^{D,\mcp_{A_t},\ell,\mcr}(d_3,a,\bar{w},\sigma_1(q))$.

By compactness of the tree decomposition $(\tau,\chi)$, we know that there is no edge between any two distinct members of $V(G_{t_1})\backslash A_{t_1},\ldots,V(G_{t_a})\backslash A_{t_a}$. This implies that for a part $Q^0_i$ of $\mcq^0$, if $Q^2_{\sigma_2^{-1}(i)}$, which is contained in $Q^0_i$, does not intersect $A_{t_a}$, then
\[|\delta_{G_\leq(\ell,a)}(Q^0_i)|=|\delta_{G_\leq(\ell,a-1)}(Q^0_{\sigma_1^{-1}(i)})|+|\delta_{G_{t_a}}(Q^0_{\sigma_2^{-1}(i)})|=z_{\sigma_1^{-1}(i)}+x_{\sigma_2^{-1}(i)}.\]
If $Q^2_{\sigma_2^{-1}(i)}$ intersects $A_{t_a}$, then
\[|\delta_{G_\leq(\ell,a)}(Q^0_i)|=z_{\sigma_1^{-1}(i)}+x_{\sigma_2^{-1}(i)}-|\delta_{G[A_{t_a}]}(R'_{\sigma_2^{-1}(i)})|.\]
Since $Q^2_{\sigma_2^{-1}(i)}$ intersects $A_{t_a}$ if and only if $\sigma_2^{-1}(i)\leq k_a$, combining these two equations, we get
\[|\delta_{G_\leq(\ell,a)}(Q^0_i)|=z_{\sigma_1^{-1}(i)}+x_{\sigma_2^{-1}(i)}-\mathds{1}[\sigma_2^{-1}(i)\leq k_a]\cdot|\delta_{G[A_{t_a}]}(R'_{\sigma_2^{-1}(i)})|=w_i.\]
The above holds for every $i\in [k]$, thus implying condition (ii) needed to witness \nusness of $\nu^{D,\mcp_{A_t},\ell,\mcr}(d_3,a,\bar{w},\sigma_1(q))$. The part in $\mcq^0$ that contains $O$ is $Q^0_{\sigma_1(q)}$, thus implying condition (iv) needed to witness \nusness of $\nu^{D,\mcp_{A_t},\ell,\mcr}(d_3,a,\bar{w},\sigma_1(q))$.

Again due to compactness of the tree decomposition $(\tau,\chi)$, we know that
\[|\delta_T(\mcq^0)|=|\delta_T(\mcq^1)|+|\delta_T(\mcq^2)|-|\delta_T(\mathcal{R}_{A_{t_a}})|\leq d_1+d_2-|\delta_T(\mathcal{R}_{A_{t_a}})|,\]
thus implying condition (iii) needed to witness \nusness of $\nu^{D,\mcp_{A_t},\ell,\mcr}(d_3,a,\bar{w},\sigma_1(q))$.
This shows that $\mcq^0$ indeed witnesses \nusness of $\nu^{D,\mcp_{A_t},\ell,\mcr}(d_3,a,\bar{w},\sigma_1(q))=1$, where $d_3$ and $\bar{w}$ are in the range given in our algorithm.
\end{proof}

\subsubsection{Proof of Lemma \ref{lemma:compute g}}
In this subsection, we complete the proof of Lemma \ref{lemma:compute g}. 

\begin{proof}[Proof of Lemma \ref{lemma:compute g}]
Let the inputs be as stated in Lemma \ref{lemma:compute g}. First we will consider the case where $t\in V(\tau)$ is a leaf node. By Lemma \ref{lemma:h in leaf nodes}, we can compute $h^{D,\mcp_{A_t}}:\prange\times\krange\times\lrange^k\times[k]\to\{0,1\}$ that is \hc and \hs in time $(\lambda k)^{O(k^2)}n^{O(1)}$. 

If $t\in V(\tau)$ is not a leaf node, then by Lemma \ref{lemma:h in non-leaf node}, we can compute a function  $h^{D,\mcp_{A_t}}:\prange\times\krange\times\lrange^k\times[k]\to\{0,1\}$ that is \hc and \hs in $(\lambda k)^{O(k^2)}n^{O(1)}$ time.

Therefore, in either case, to compute a desired function $h^{D,\mcp_{A_t}}:\prange\times\krange\times\lrange^k\times[k]\to\{0,1\}$ that is \hc and \hs takes total time $(\lambda k)^{O(k^2)}n^{O(1)}$. By Lemma \ref{lemma:g update from h}, we can compute a function $g^{D,\mcp_{A_t}}:\prange\times\krange\times\lrange^k\times[k]\to\{0,1\}$ that is \gc and \gs with an additional $(\lambda k)^{O(k^2)}n$ time. This completes the proof of Lemma \ref{lemma:compute g}.

\end{proof}

\subsection{Generating Nice Decompositions and Proof of Lemma \ref{lemma:nice decomp}}
\label{sec:generating-nice-decompositions}
In this section, we restate and prove Lemma \ref{lemma:nice decomp}.

\lemGeneratingNiceDecomps*

Our definition of nice decomposition closely resembles the definition of \cite{LSS}. Our way to generate nice decompositions and thereby prove Lemma \ref{lemma:nice decomp} will also closely resemble the proof approach of \cite{LSS}.
We need the following lemma. 
\begin{lemma} [Lemma 2.1 of \cite{LSS}] \label{lemma:2.1 LSS}
There exists an algorithm that takes as input a set $S$ and positive integers $s_1, s_2\le |S|$, and runs in time $O((s_1+s_2)^{O(s_1)}|S|^{O(1)})$ to return a family $\mathcal{S}\subseteq 2^S$ of size $O((s_1+s_2)^{O(s_1)}\log|S|)$ such that for every pair of disjoint subsets $X_1,X_2\subseteq S$ where $|X_1|\leq s_1$ and $|X_2|\leq s_2$, there exists a set $X\in\mathcal{S}$ with $X_1\subseteq X\subseteq S\backslash X_2$.
\end{lemma} 

Let the inputs be as stated in Lemma \ref{lemma:nice decomp}. We will start with notations followed by the algorithm with bound on the runtime and a proof of correctness.

\paragraph{Notations.} 
Let $A\subseteq E(\proj(T,\ct))$. We let $\mcr_A$ denote the partition of $V(\proj(T,\ct))$ whose parts are the connected components of $\proj(T,\ct)\backslash A$. 
Let $P$ and $P'$ be disjoint subsets of vertices of $\proj(T, \ct)$. We say that $P$ \emph{shares adhesion} with $P'$ if there exists some descendant $t'$ of $t$ (inclusive) such that $P\cap A_{t'}\neq\emptyset$ and $P'\cap A_{t'}\neq\emptyset$. When it is necessary to specify which adhesion is shared by $P$ and $P'$, we will say that \emph{$P$ shares adhesion with $P'$ via $A_{t'}$} if $t'$ is a descendant of $t$ (inclusive) such that $A_{t'}$ intersects both $P$ and $P'$. Moreover, we say that $P$ \emph{shares $G_t$-edge} with $P'$ if there is an edge in $G_t$ with one end-vertex in $P$ and another end-vertex in $P'$. 
For a partition $\mcr$ of $V(\proj(T,\ct))$, we define the graph $H(\mcr)$ whose vertices correspond to the parts of $\mcr$ and two vertices are adjacent in $H(\mcr)$ if and only if the corresponding parts in $\mcr$ share either an adhesion or a $G_t$-edge with each other.


\paragraph{Algorithm.} We now describe the algorithm. 
We initialize $\mcd$ to be the empty set.
The first step of our algorithm is to generate a family $\mathcal{C}$ using Lemma \ref{lemma:2.1 LSS} on set $E(\proj(T,\ct))$ with $s_1=\min\{|E(\proj(T,\ct))|,4k^2\}$ and $s_2=\min\{|E(\proj(T,\ct))|,(4k^2+1)\cdot 2((\lambda k+1)^5+4k^2+1)\}$.
The next step of our algorithm depends on the size of the bag $\chi(t)$.

\noindent
\emph{Case 1:} Suppose $|\ct|\leq (\lambda k+1)^5$. For each $C'\in\mathcal{C}$, if $\mcr_{C'}$ has at most $4k^2+1$ parts, then we add the following triple to $\mathcal{D}$:
\[D_{C'}:=\left(\mcp_{\chi(t)}:=(\ct),\mcq_{\chi(t)}:=\mcr_{C'}|_{\ct},O:=\emptyset\right),\]
where the partition $\mcr_{C'}|_{\ct}$ is the restriction of $\mcr_{C'}$ to $\ct$. If $\mcr_{C'}$ has more than $4k^2+1$ parts, then we do not add anything into $\mathcal{D}$.

\noindent
\emph{Case 2:} Suppose $|\ct|>(\lambda k+1)^5$. 
For each $C'\in\mathcal{C}$, 
we generate a family $\mathcal{S}_{C'}$ using Lemma \ref{lemma:2.1 LSS} on the set $\{P:P\text{ is a part of }\mcr_{C'}\}$ with $s_1=\min\{|\mcr_{C'}|,4k^2+1\}$ and $s_2=\min\{|\mcr_{C'}|,(4k^2+1)(\lambda^2 k^2+2\lambda k+4k^2+1)\}$.
For each $C'\in\mathcal{C}$ and each set $X'\in\mathcal{S}_{C'}$, we use the following steps to update $\mathcal{D}$. 
\begin{enumerate}
    \item 
    Starting from the partition $\mcq_0:=\mcr_{C'}$, we merge all parts that are not in $X'$ together to be one part called $O_1$. Call the resulting partition $\mcq_1$.
    \item  In the graph $H(\mcq_1)$, for each connected component of $H(\mcq_1)\backslash\{O_1\}$ that has more than $4k^2+1$ vertices, we merge the parts corresponding to the vertices of the component with $O_1$. Let $\mcq_2$ be the resulting partition and let $O_2$ be the part of $\mcq_2$ that contains $O_1$.
    \item In the graph $H(\mcq_2)$, for each connected component of $H(\mcq_2)\backslash\{O_2\}$, we merge the parts corresponding to vertices in the component. Let $\mcq_3$ be the resulting partition.
    \item Add the triple $D_{C',X'}:=(\mcp_{\ct},\mcq_{\ct},O)$ to $\mathcal{D}$, where $\mcp_{\ct},\mcq_{\ct}$ and $O$ are restrictions of $\mcq_3$, $\mcq_2$, and $O_2$ to $\ct$, respectively.
\end{enumerate}
This completes the description of our algorithm.

\paragraph{Run-time.} 
We now bound the run time of this algorithm. The family $\mathcal{C}$ can be computed in $O((4k^2+O(k^2(\lambda k+1)^5))^{4k^2}n^{O(1)})=(\lambda k)^{O(k^2)}n^{O(1)}$ time, and the size of $\mathcal{C}$ is $(\lambda k)^{O(k^2)}\log n$ using Lemma \ref{lemma:2.1 LSS}. In the next step, Case 1 runs in $n^{O(1)}$ time for each $C'\in\mathcal{C}$. Case 2 requires $O((4k^2+1+(4k^2)((\lambda k)^2+2\lambda k+4k^2+1))^{4k^2+1}n^{O(1)})=(\lambda k)^{O(k^2)}n^{O(1)}$ time to generate $\mathcal{S}_{C'}$ for each $C'\in\mathcal{C}$, and the size of each $\mathcal{S}_{C'}$ is bounded by $(\lambda k)^{O(k^2)}\log n$. The rest of the steps in Case 2 take $n^{O(1)}$ time. Therefore, the total time needed for this algorithm is
\[(\lambda k)^{O(k^2)}n^{O(1)}+ (\lambda k)^{O(k^2)}\log n\cdot (\lambda k)^{O(k^2)}n^{O(1)} =(\lambda k)^{O(k^2)}n^{O(1)}.\]

\paragraph{Correctness.} 
We now prove the correctness of the algorithm. We first show that the algorithm indeed outputs a family of nice decompositions.
\begin{claim}\label{claim D is family of nice decomp}
\begin{enumerate}
    \item If $|\ct|\leq (\lambda k+1)^5$, then the triple $D_{C'}$ is indeed a nice decomposition for every $C'\in\mathcal{C}$ for which $D_{C'}$ is defined.
    \item If $|\ct|> (\lambda k+1)^5$, then, for every $C'\in\mathcal{C}$ and $X'\in\mathcal{S}_{C'}$, the triple $D_{C',X'}$ is indeed a nice decomposition.
\end{enumerate}

\end{claim}
\begin{proof}
Suppose $|\ct|\leq(\lambda k+1)^5$ and $\mcr_{C'}$ has at most $4k^2+1$ parts. Then $\mcr_{C'}|_{\ct}$ also has at most $4k^2+1$ parts, which verifies condition 2 in Definition \ref{def nice decomp}. The remaining conditions in the definition hold immediately, and hence, $D_{C'}$ is indeed a nice decomposition.

We henceforth consider the case where $|\ct|>(\lambda k+1)^5$. 
Let $C'\in\mathcal{C}$ be fixed and $X'\in\mathcal{S}_{C'}$, yielding a triple $D_{C',X'}:=(\mcp_{\ct},\mcq_{\ct},O)$ to be included into $\mathcal{D}$. 
We will prove that $D_{C', X'}$ satisfies the conditions in Definition \ref{def nice decomp}.
We note that $\mcq_2$ refines $\mcq_3$, so $\mcq_{\ct}$ refines $\mcp_{\ct}$. 

For $i=0,1,2$, let $H'_i$ denote the subgraph of $H(\mcq_i)$ induced by vertices whose corresponding parts intersect $\ct$. By compactness of the tree decomposition $(\tau,\chi)$, the subgraph $G_t$ is connected. Therefore, there exists a path in $G_t$ between every pair of vertices of $\ct$. This implies that there exists a path in $G_t$ between every pair of parts of $\mcq_0$. Hence, the graph $H'_0$ is connected.

Let us first consider the case where $O=O_2\cap\ct=\emptyset$. In this case the parts merged to become $O_1$ do not intersect $\ct$, and hence $H'_0=H'_1$. If $H'_0=H'_1$ has more than $4k^2+1$ vertices, then $H'_1$ will be merged into $O_2$. Therefore, we conclude that $H'_0=H'_1$ has no more than $4k^2+1$ vertices and $H'_0=H'_1=H'_2$.  In step 3, $H'_2$ is in one component of $H(\mcq_2)\backslash\{O_2\}$. Therefore, in $\mcq_3$, only one part intersects $\ct$ and this part consists of at most $4k^2+1$ parts in $\mcq_2$ intersecting $\ct$. This implies $\mcp_{\ct}$ has only one part, and this part contains at most $4k^2+1$ parts of $\mcq_{\ct}$. The rest of the conditions in the definition of nice decomposition hold immediately.

Next we consider the case where $O=O_2\cap\ct\neq\emptyset$. Since $O_2$ is a part of $\mcq_2$, we know that $O$ is a part of $\mcq_{\ct}$, which proves condition (i) of the definition of nice decomposition. Every part of $\mcq_3$ consists of at most $4k^2+1$ parts of $\mcq_2$ due to step 3, so every part of $\mcp_{\ct}$ consists of at most $4k^2+1$ parts of $\mcq_{\ct}$, proving condition (ii). For two distinct parts $P$ and $P'$ of $\mcq_3$, $P$ and $P'$ share either an adhesion or a $G_t$-edge only if the corresponding vertices in $H(\mcq_3)$ are adjacent. In graph $H(\mcq_3)$, each edge has one end-vertex being $O_2$. This implies that one of $P$ and $P'$ has to be $O_2$. Therefore, no two parts other than $O$ in $\mcp_{\ct}$ share an adhesion or a $G_t$-edge, thus proving conditions (iii) and (iv) in the definition of nice decomposition. This completes the proof of Claim \ref{claim D is family of nice decomp}. 

\end{proof}

The following lemma completes the proof of Lemma \ref{lemma:nice decomp}. 

\begin{lemma}\label{lemma:D contains refinement}
If a $k$-subpartition $\Pi$ of $V(G_t)$ witnesses \fcness of $f_t(\mcp_{A_t},\bar{x},d)$ for some $\mcp_{A_t}\in\mcf^{A_t}$, $\bar{x}\in\lrange^k$ and $d\in\krange$, then $\mathcal{D}$ contains a nice decomposition $D=(\mcq_{\ct},\mcq_{\ct},O)$ such that $\mcq_{\ct}$ refines a restriction of $\Pi$ to $\ct$. 
\end{lemma}

\begin{proof}
We begin with some notations. Let $\Pi$ be a $k$-subpartition of $V(G_t)$ that witnesses \fcness of $f_t(\mcp_{A_t},\bar{x},d)$ for some $\mcp_{A_t}\in\mcf^{A_t}$, $\bar{x}\in\lrange^k$, and $d\in\krange$. 
Then, $\Pi$ is a restriction of $\Omega$ to $V(G_t)$ and $|\delta_T(\Pi)|\leq d\leq 2k^2$.
Let $\Pi_{\ct}$ 
denote a partition of $\ct$ that is a restriction of $\Pi$ to $\chi(t)$. 
By definition, the partition $\mcp_{A_t}$ is a restriction of $\Pi$ to $A_t$. 

\begin{claim}\label{claim C size bound}
There exists a partition $\Tilde{\Pi}_{\chi(t)}$ of $\proj(T,\chi(t))$ such that $\Pi_{\ct}$ is a restriction of $\Tilde{\Pi}_{\chi(t)}$ to $\ct$ and $|\delta_{\proj(T,\ct)}(\tpi_{\ct})|\leq 4k^2$.
\end{claim}
\begin{proof}

Since $\mcp_{A_t}$ is $T$-feasible, there exists a partition $\mathcal{B}$ of $V$ that $2k^2$-respects $T$ and its restriction to $A_t$ is $\mcp_{A_t}$. Without loss of generality we may assume that $\mathcal{B}$ has exactly $k'$ parts $B_1,\ldots,B_{k'}$ (recall that $\mathcal{P}_{A_t}=(\Tilde{P}_1,\ldots,\Tilde{P}_{k'})$) such that $B_i\cap A_t=\tilde{P}_i$ for all $i\in[k']$. Moreover, we may assume that $\Pi$ has exactly $k'$ parts $\pi_1,\ldots,\pi_{k'}$ such that $\pi_i\cap A_t=\tilde{P}_i$ for all $i\in[k']$. With these two partitions, using compactness of the tree decomposition $(\tau,\chi)$, we can define a partition $\mathcal{B}=(B'_1,\ldots,B'_{k'})$ of $V$ by $B'_i=(B_i\backslash V(G_t))\cup \pi_i$ for each $i\in[k']$. The restriction of $\mathcal{B}'$ to $\ct$ is $\Pi_{\ct}$. Moreover, we have $|\delta_T(\mathcal{B}')|\leq |\delta_T(\mathcal{B})|+|\delta_T(\Pi)|\leq 4k^2$. 

Now by Lemma \ref{lemma proj cut size}, given partition $\mathcal{B}'$ of $V$, there exists a partition $\tpi_{\ct}$ of $\proj(T,\ct)$ such that its restriction to $\ct$ is $\Pi_{\ct}$ and $|\delta_{\proj(T,\ct)}(\tpi_{\ct})|\leq |\delta_T(\mathcal{B}')|\leq 4k^2$. This completes the proof of Claim \ref{claim C size bound}. 
\end{proof}

From now on we will fix $\tpi_{\ct}$ to be a partition of $V(\proj(T,\ct))$ that satisfies the conditions of Claim \ref{claim C size bound}. Let $\Tilde{C}$ denote the set of edges $\delta_{\proj(T,\ct)}(\Tilde{\Pi}_{\chi(t)})$, which is a subset of $E(\proj(T,\chi(t)))$. It follows that $|\tilde{C}|\leq 4k^2$.

Removing $\Tilde{C}$ from $\proj(T,\chi(t))$ yields a partition of $V(\proj(T,\chi(t)))$ whose parts are connected components of $\proj(T,\chi(t))\backslash C$. We denote this partition as $\tpi'_{\ct}$, and observe that $\tpi'_{\ct}$ is necessarily a refinement of $\tpi_{\ct}$. Moreover, we may assume that each part of $\tpi'_{\ct}$ intersects $\ct$. If some parts of $\tpi'_{\ct}$ do not intersect $\ct$, then there exist two parts $P$ and $P'$ of $\tpi'_{\ct}$ such that $P$ intersects $\ct$ while $P'$ does not, and there is an edge in $\Tilde{C}\subseteq E(\proj(T, \ct))$ with one end-vertex in $P$ and the other end-vertex in $P'$. We can modify $\tpi_{\ct}$ so that $P'$ belongs to the part containing $P$. After such modification, the size of $\delta_{\proj(T,\ct)}(\tpi_{\ct})$ does not increase (because grouping $P'$ into the part containing $P$ does not require us to cut any edge not in $\tilde{C}$) and $\Pi_{\ct}$ is still a restriction of $\tpi_{\ct}$ to $\ct$. By doing this repeatedly, we may assume that each part of $\tpi'_{\ct}$ intersects $\ct$ while the size of $\tilde{C}$ does not increase.

Since $\Omega$ is an optimum \mmkc in $G$, we have that $|\delta_G(\Omega)|\leq k\OPT \leq k\lambda$. Since $\Pi_{\ct}$ is a restriction of $\Omega$ to $\ct$, by the edge-unbreakability property of $\ct$, we know that at most one part of $\Pi_{\ct}$ has size exceeding $(\lambda k+1)^5$. Now, let $S_{\ct}$ denote the union of the parts of $\Pi_{\ct}$ whose sizes are at most $(\lambda k+1)^5$, i.e.,
\[S_{\ct}:=\bigcup_{\substack{\pi\text{ is a part of }\Pi_{\ct},\\|\pi|\leq (\lambda k+1)^5}}\pi.\]
Furthermore, we use $\Tilde{S}$ to denote the union of parts of $\tpi_{\ct}$ whose intersection with $S_{\ct}$ is non-empty. Each part of $\tpi_{\ct}$ included in $\Tilde{S}$ induces a set of edges in $\proj(T,\ct)$ whose both end-vertices are in this part. We use $\Tilde{S}_E$ to denote the union of such edges, i.e.,
\[\Tilde{S}:=\bigcup_{\substack{\pi\text{ is a part of }\tpi_{\ct},\\\pi\cap S_{\ct}\neq\emptyset}}\pi \text{ and}\]
\[\Tilde{S}_E:=\bigcup_{\substack{\pi\text{ is a part of }\tpi_{\ct},\\\pi\cap S_{\ct}\neq\emptyset}}E(\proj(T,\ct)[\pi]).\]
We remark that $\tilde{S}_E\cap\tilde{C}=\emptyset$ because every edge in $\tilde{C}$ has end-vertices in different parts of $\tpi_{\ct}$.
For convenience, we summarize the nature of notations introduced here in Table \ref{table:notations}.

\begin{table}[ht]
    \centering
    \begin{tabular}{c|c}
      $\Pi$ & $k$-subpartition of $V(G_t)$  \\
      \hline
        $\Pi_{\ct}$ & Partition of $\ct$\\
        \hline
        $\tpi_{\ct}$ & Partition of $V(\proj(T,\chi(t)))$\\
        \hline
        $\tpi'_{\ct}$ & Partition of $V(\proj(T,\chi(t)))$\\
        \hline
        $\Tilde{C}$ & Subset of $E(\proj(T,\ct))$\\
        \hline
        $S_{\ct}$ & Subset of $\ct$\\
        \hline
        $\Tilde{S}$ & Subset of $V(\proj(T,\ct))$\\
        \hline
        $\Tilde{S}_E$ & Subset of $E(\proj(T,\ct))$
    \end{tabular}
    \caption{Nature of notations for the proof of Lemma \ref{lemma:D contains refinement}.}
    \label{table:notations}
\end{table}

Now that we have introduced these notations and definitions, 
our next goal is to bound the size of $\tilde{S}_E$. We will use the following claim. 
\begin{claim}\label{claim:part of tpi' size bound}
For every part $P$ of $\tpi'_{\ct}$, we have that $|P|\leq 2(|P\cap\ct|+4k^2+1)$.
\end{claim}
\begin{proof}
If $P$ is a part of $\tpi'_{\ct}$, then by definition of $\tpi'_{\ct}$, the subgraph induced by $P$ in $\proj(T,\ct)$, i.e., $\proj(T,\ct)[P]$, is a connected subtree of $\proj(T,\ct)$. To bound the size of $P$, we will bound the sizes of the following types of vertices:
\begin{enumerate}
    \item vertices of $P$ with degree at least 3 in $\proj(T,\ct)[P]$,
    \item vertices of $P$ that are not in $\ct$ and have degree at most $2$ in the subtree $\proj(T,\ct)[P]$, and
    \item vertices in $P\cap\ct$.
\end{enumerate}
These three types together form a superset of the vertices of $P$.

In order to bound the number of Type 2 vertices, we note that these vertices are in $V(\proj(T,\ct))\backslash\ct$, which means they are of degree at least 3 in $\proj(T,\ct)$. For each Type 2 vertex, some edge in $\proj(T,\ct)$ adjacent to it which connects $P$ to some other component is not included in $\proj(T,\ct)[P]$, resulting it to have degree at most 2 in $\proj(T,\ct)[P]$. Since each part of $\tpi'_{\ct}$ induces a connected subtree of $\proj(T,\ct)$, each Type 2 vertex serves to connect $P$ to a unique part of $\tpi'_{\ct}$. Here $\tpi'_{\ct}$ has at most $4k^2+1$ parts as $|\tilde{C}|\leq 4k^2$, and hence the number of Type 2 vertices is at most $4k^2+1$.

In order to bound the number of Type 1 vertices, we will first bound the number of leaves of $\proj(T,\ct)[P]$. Leaves of $\proj(T,\ct)[P]$ are either in $\ct$ or not in $\ct$. The number of leaves in $\ct$ is at most $|P\cap\ct|$. The leaves of $\proj(T,\ct)[P]$ that are not in $\ct$ are Type 2 vertices. So the total number of leaves of $\proj(T,\ct)[P]$ is at most $|P\cap\ct|+4k^2+1$.

Next we bound the number of Type 1 vertices. We have the following inequality, where all degrees are with respect to the subgraph $\proj(T,\ct)[P]$:
\begin{align*}
    2|P|-2&=2|E(\proj(T,\ct)[P])|=\sum_{v\in P}\deg(v)
    \\&\geq |\{v\in P:\deg(v)=1\}|+2|\{v\in P:\deg(v)=2\}|+3|\{v\in P:\deg(v)\geq 3\}|
    \\&=2|P|-|\{v\in P:\deg(v)=1\}|+|\{v\in P:\deg(v)\geq 3\}|.
\end{align*}
where the last equation holds due to the fact that $|\{v\in P:\deg(v)=1\}|+|\{v\in P:\deg(v)=2\}|+|\{v\in P:\deg(v)\geq 3\}|=|P|$.
This implies that the number of Type 1 vertices can be bounded by the following relationship:
\[|\{v\in P:\deg(v)\geq 3\}|\leq |\{v\in P:\deg(v)=1\}|\leq |P\cap\ct|+4k^2+1.\]

The number of Type 3 vertices is exactly $|P\cap\ct|$, and hence the size of $P$ is at most 
\[(4k^2+1)+(|P\cap\ct|+4k^2+1)+|P\cap\ct|=2(|P\cap\ct|+4k^2+1).\]
\end{proof}

\begin{claim}\label{claim:SE size bound}
We have that $|\tilde{S}|,|\tilde{S}_E|\leq (4k^2+1)\cdot 2((\lambda k+1)^5+4k^2+1)$. 
\end{claim}
\begin{proof}
We will start by bounding the size of $\tilde{S}$. Let us fix $\pi$ to be a part of $\tpi_{\ct}$ such that $\pi\cap S_{\ct}\neq\emptyset$, and $\pi'$ to be a part of $\tpi'_{\ct}$ contained in $\pi$. Here, we notice that $\pi\cap\ct$ is a part of $\Pi_{\ct}$ and $|\pi\cap\ct|\leq (\lambda k+1)^5$. Then by Claim \ref{claim:part of tpi' size bound}, we know that
\[|\pi'|\leq 2(|\pi'\cap\ct|+4k^2+1)\leq 2(|\pi\cap\ct|+4k^2+1)\leq 2((\lambda k+1)^5+4k^2+1).\]
The set $\tilde{S}$ is the union of all such parts $\pi'$, i.e., it is the union of parts $\pi'$ of $\tpi'_{\ct}$ where $\pi'$ is contained in some part $\pi$ of $\tpi_{\ct}$ such that $\pi \cap S_{\ct}\neq \emptyset$. There are at most $4k^2+1$ such candidates for $\pi'$ because $\tpi'_{\ct}$ has at most $4k^2+1$ parts. Hence, 
\[|\tilde{S}|\leq(4k^2+1)\cdot 2((\lambda k+1)^5+4k^2+1).\]
To bound the size of $\tilde{S}_E$, we observe that $\tilde{S}_E$ forms a forest over the vertex set $\tilde{S}$, and thus $|\tilde{S}_E|\leq |\tilde{S}|\leq (4k^2+1)\cdot 2((\lambda k+1)^5+4k^2+1)$.
\end{proof}

The first step of our algorithm generates a family $\mathcal{C}$ using Lemma \ref{lemma:2.1 LSS} on the set $E(\proj(T,\ct))$ with $s_1=\min\{|E(\proj(T,\ct))|,4k^2\}$ and $s_2=\min\{|E(\proj(T,\ct))|,(4k^2+1)\cdot 2((\lambda k+1)^5+4k^2+1)\}$. By Claim \ref{claim:SE size bound} and Lemma \ref{lemma:2.1 LSS}, this implies that $\mathcal{C}$ contains a set $C$ such that $\tilde{C}\subseteq C$ and $C\cap \tilde{S}_E=\emptyset$. For the rest of the proof, let us fix $C\in\mathcal{C}$ such that $\tilde{C}\subseteq C$ and $C\cap\tilde{S}_E=\emptyset$. We now introduce some more more notations and prove certain useful properties of $C$.

Here, we note that the partition $\mcr_C$ refines $\tpi'_{\ct}$ because $\tilde{C}\subseteq C$. We use $L_C$ to denote the set of parts of $\mcr_C$ that are contained in $\tilde{S}$. 
We use $N_C$ to denote the set of parts of $\mcr_C$ outside $\Tilde{S}$ that either share adhesion with some part in $L_C$ or share $G_t$-edges with some part in $L_C$. We will use the following observation and claim to bound the size of $N_C$.

\begin{observation}\label{obs:RC}
Every edge in $C\backslash\tilde{C}$ necessarily has both end-vertices in $V(\proj(T,\ct))\backslash \tilde{S}$. This is because every edge between $V(\proj(T,\ct))\backslash \tilde{S}$ and $\tilde{S}$ belongs to $\tilde{C}$, and every edge whose both end-vertices are in $\tilde{S}$ are either in $\tilde{S}_E$ (which does not intersect $C$) or in $\tilde{C}$.
This implies that when restricted to $\tilde{S}$, the partition $\mcr_C$ and $\tpi'_{\ct}$ are the same.
\end{observation}

\begin{claim}\label{claim:bounding nbrs in R}
$|N_C|\leq (4k^2+1)((\lambda k)^2+2\lambda k+4k^2+1).$
\end{claim}
\begin{proof}
Let us fix one part $R$ in $L_C$ and bound the number of parts of $\mcr_C$ that could share adhesion or $G_t$-edge with $R$. By Observation \ref{obs:RC}, we know that $R$ is also a part of $\tpi'_{\ct}$. We will use $\pi$ to denote the part of $\tpi_{\ct}$ that contains $R$. The parts of $\mcr_C$ that share either an adhesion or a $G_t$-edge with $R$ can be enumerated by the following four types:
\begin{enumerate}
    \item parts outside $\pi$ that share adhesion with $R$ via $A_{t'}$, where $t'$ is a child of $t$,
    \item parts outside $\pi$ that share adhesion with $R$ via $A_t$,
    \item parts outside $\pi$ that share $G_t$-edge with $R$, and
    \item parts in $\pi$.
\end{enumerate}
We bound the number of parts of Type 1. For this, we will first bound the number of children $t'$ of $t$ such that $A_{t'}$ intersects both $R$ and some part outside $\pi$.
Let $t'$ be a child of $t$ such that $A_{t'}$ intersects $R$ and a part $R'$ that is outside $\pi$. Then $R\cap A_{t'}$ and $R'\cap A_{t'}$ are in different parts of $\Pi$. 
By compactness of $\tau$, we know that $R\cap A_{t'}$ has a neighbor $v$ in $V(G_{t'})\backslash A_{t'}$, and $R'\cap A_{t'}$ has a neighbor $v'$ in $V(G_{t'})\backslash A_{t'}$. 
Since $V(G_{t'})\backslash A_{t'}$ induces a connected subgraph in $G$, there is a path between $v$ and $v'$ in $G[V(G_{t'})\backslash A_{t'}]$. 
Hence, in order to separate $R\cap A_{t'}$ and $R'\cap A_{t'}$, the $k$-subpartition $\Pi$ must cut some edge with at least one end-vertex in $V(G_{t'})\backslash A_{t'}$. We fix one such edge and denote it as $e_{t'}$. Then $e_{t'}$ is contained in $\delta_G(\Pi)$. We associate one such edge $e_{t'}$ for each child $t'$ of $t$ such that $A_{t'}$ intersects $R$ and a part $R'$ that is outside $\pi$.

We now show that the edge $e_{t'}$ associated the child $t'$ is unique. For the sake of contradiction, suppose that $e_{t'}=e_{t''}=e$ for two children $t'$ and $t''$ of $t$. Let $e=\set{u', u''}$ with $u'\in V(G_{t'})\setminus A_{t'}$ and $u''\in V(G_{t''})\setminus A_{t''}$. 
Then, $e$ is contained in some bag $\chi(t_0)$. The bags containing $u'$ induce a connected subtree, and $u'\notin\ct$, so $t_0$ must be a descendant of $t'$ (inclusive). Similarly $t_0$ must be a descendant of $t''$ (inclusive). This is a contradiction because $t'$ and $t''$ are distinct children of $t$.

Therefore, the number of children $t'$ of $t$ such that $A_{t'}$ intersects both $R$ and some other part outside $\pi$ is at most $|\delta_G(\Pi)|\leq\lambda k$. Each such adhesion has size at most $\lambda k$, so it contributes at most $\lambda k$ to the number of adjacent parts that $R$ could have. Hence the size of type 1 is bounded above by $(\lambda k)^2$.

In order to bound the number of parts of Type 2, we use the fact that $|A_t|\leq\lambda k$ and conclude that the the number of parts of Type 2 is at most $\lambda k$.

In order to bound the number of parts of Type 3, we observe that a $G_t$-edge with one end-vertex in $\pi$ and the other end-vertex not in $\pi$ is necessarily in $\delta_G(\Pi)$. Each part outside $\pi$ that shares $G_t$-edge with $R$ connects $R$ to a unique edge in $\delta_G(\Pi)$, and hence the number of parts of Type 3 is at most $|\delta_G(\Pi)|\leq\lambda k$.

Lastly, we bound the number of parts of Type 4. Since $R$ is a part in $L_C$, by definition we know that $\pi$ is contained in $\tilde{S}$. This means that every part in $\pi$ is also a part of $\tpi'_{\ct}$ by Observation \ref{obs:RC}. Therefore, the size of type 4 is at most $4k^2+1$.

We conclude that $R$ shares an adhesion or a $G_t$-edge with at most $(\lambda k)^2+\lambda k+\lambda k+4k^2+1$ parts of $\mcr_C$. Hence, the size of $N_C$ is at most
\[|L_C|\cdot((\lambda k)^2+\lambda k+\lambda k+(4k^2+1)).\]
Since parts in $L_C$ are also parts in $\tpi'_{\ct}$, we know that  $|L_C|\leq 4k^2+1$. This yields the desired bound: 
\[|N_C|\leq (4k^2+1)((\lambda k)^2+2\lambda k+4k^2+1).\]
\end{proof}

We now have the ingredients to show that $\mcd$ contains a nice decomposition $D=(\mcp_{\ct},\mcq_{\ct},O)$ such that $\mcq_{\ct}$ refines $\Pi_{\ct}$. We begin with the easier case where the size of the bag is small.

\begin{claim}\label{claim:D yields refinement small case}
If $|\ct|\leq (\lambda k+1)^5$, then the nice decomposition $D_C=((\ct),\mcr_C|_{\ct},\emptyset)$ is contained in $\mathcal{D}$ and the partition $\mcr_C|_{\ct}$ refines $\Pi_{\ct}$.
\end{claim}
\begin{proof}

If $|\ct|\leq (\lambda k+1)^5$, then we note that $S_{\ct}=\ct$ and hence $\tilde{S}=V(\proj(T,\ct))$. By Observation \ref{obs:RC}, we know that $\mcr_C=\tpi'_{\ct}$. This guarantees that $\mcr_C$ has at most $4k^2+1$ parts, and thus the triple $((\ct),\mcr_C|_{\ct},\emptyset)$ is defined and added into $\mathcal{D}$. Moreover, we know that $\mcr_C=\tpi'_{\ct}$ refines $\tpi_{\ct}$, and hence $\mcr_C|_{\ct}$ refines $\Pi_{\ct}$.
\end{proof}

We now handle the case where the size of the bag is large. For the rest of the proof, suppose that $|\ct|>(\lambda k+1)^5$. We recall that $\tilde{C}\subseteq C$ and $C\cap \tilde{S}_E=\emptyset$. Since $|L_C|\le 4k^2+1$, by Claim \ref{claim:bounding nbrs in R} and Lemma \ref{lemma:2.1 LSS}, the family $\mathcal{S}_C$ is guaranteed to contain a set $X$ such that $L_C\subseteq X$ and $X\cap N_C=\emptyset$. Let us fix such a set $X$ in the remainder of the proof. The next claim states that the nice decomposition $D_{C,X}$ yields a refinement of $\Pi_{\ct}$ as desired, thereby completing the proof of Lemma \ref{lemma:D contains refinement}.

\end{proof}

\begin{claim}\label{claim D contains refinement}
If $|\ct|> (\lambda k+1)^5$, then the nice decomposition $D_{C,X}=(\mcp_{\ct},\mcq_{\ct},O)$ yields a partition $\mcq_{\ct}$ of $\ct$ that refines $\Pi_{\ct}$.
\end{claim}
\begin{proof}
By definition of the set $X$, at the end of step 1 of the algorithm, the part $O_1$ contains all parts in $N_C$ and no parts in $L_C$. Moreover, in the graph $H(\mcq_1)$, we have $N_{H(\mcq_1)}(L_C)=\{O_1\}$. Since $L_C$ contains at most $4k^2+1$ parts, we know that no part in $L_C$ is merged into $O_1$ in step 2. Therefore, every part of $\mcr_C$ in $L_C$ remains a single part in $\mcq_2$. This implies that every part of $\tpi'_{\ct}$ in $\tilde{S}$ remains a single part in $\mcq_2$. If a part $\pi$ of $\tpi_{\ct}$ is contained in $\tilde{S}$, then $\pi$ is the union of some parts of $\tpi'_{\ct}$ in $\tilde{S}$, and hence the union of some parts of $\mcq_2$. If a part $\pi$ of $\tpi_{\ct}$ is not contained in $\tilde{S}$, then $\pi$ is the unique part of $\tpi_{\ct}$ that is not contained in $\tilde{S}$. This implies that $\pi$ is the union of parts of $\mcq_2$ that are not contained in $L_C$. Thus, we conclude that $\mcq_2$ refines $\tpi_{\ct}$, and hence $\mcq_{\ct}$ refines $\Pi_{\ct}$.
\end{proof}


\section{Reduction to unweighted instances with logarithmic optimum value}
\label{sec:reduction-to-unweighted-instances}
In this section, we show a $(1+\epsilon)$-approximation preserving reduction to unweighted instances with optimum value $O((k/\epsilon^3)\log n)$. The ideas in this section are standard and are also the building blocks for the $(1+\epsilon)$-approximation for \mskc. 
Our contribution to the reduction is simply showing that the ideas also apply to \mmkc. 

\begin{theorem}
\label{thm:reduction-to-unweighted-instances-with-small-opt}
There exists a randomized algorithm that takes as input 
an $n$-vertex graph $G=(V,E)$ with edge weights $w:E\rightarrow \R_+$, 
an integer $k\ge 2$, 
and an $\epsilon\in (0,1)$, 
and runs in time $2^{O(k)}\poly(n,1/\epsilon)\log(\OPT(G_w,k))$ 
to return a collection $\mcc$ of unweighted instances of \mmkc such that with high probability 
\begin{enumerate}[(i)]
\item the size of every instance in $\mcc$ is polynomial in the size of the input instance, 
\item the number of instances in $\mcc$ is $2^{O(k)}\poly(n,1/\epsilon)\log(\OPT(G_w,k))$, and 
\item the optimum solution among all instances in $\mcc$ with costs bounded by $O((k/\epsilon^3)\log n)$ can be used to recover a $k$-partition $\mcp$ of the vertices of $G$ such that $\cost_{G_w}(\mcp)=(1+O(\epsilon))\opt(G_w, k)$
in time $2^{O(k)}\poly(n,1/\epsilon)\log(\OPT(G_w,k))$. 
\end{enumerate}
\end{theorem}

\begin{proof}
Let the input instance be $G=(V,E)$ with edge costs $w:E\rightarrow \R_+$ along with an integer $k\ge 2$ and $\epsilon \in (0,1)$. 
We assume that $G$ is connected (if not, then guess the number of parts that each connected component will split into and solve the sub-problem for each connected component - this contributes a $2^{O(k)}$ overhead to the run-time). 
We may now guess a value $\lambda$ such that $\lambda\in [\OPT(G_w,k), 2\OPT(G_w,k)]$ by doing a binary search -- this contributes a $\log(\OPT(G_w,k))$ overhead to the run-time. 

We use the guess $\lambda$ to do a knapsack PTAS-style rounding to reduce the problem in a $(1+\epsilon)$-approximation preserving fashion to an unweighted multigraph with $O(m^2/\epsilon)$ edges, where $m$ is the number of edges in the input graph $G$ (see Lemma \ref{lem:knapsack-style-rounding} for the complete arguments of the reduction). 
Due to this reduction, we may henceforth assume that the input instance $(G,k)$ of \mmkc is an unweighted instance. 
Our goal now is to get a $(1+O(\epsilon))$-approximation for the unweighted instance $G$ for a given constant $\epsilon \in (0,1)$. We still have a guess $\lambda\in [\OPT(G,k), 2\OPT(G,k)]$.

Our next step in the reduction is to repeatedly remove cut sets $F$ (i.e., a subset $F$ of edges that cross some $2$-partition of $G$) that have $0<|F|\le \epsilon \lambda/2(k-1)$. If we can remove $k-1$ such cut sets, then we would have removed at most $\epsilon \lambda/2\le \epsilon\OPT(G,k)$ edges while creating $k$ connected components, thus contradicting optimality. Thus, we can remove at most $k-2$ such cut sets and the number of edges removed is less than $\epsilon \lambda/2\le \epsilon \OPT(G,k)$. At the end, we have a subgraph $H_1=(V,E_1)$ of $G$ such that 
(i) $|E-E_1|\le \epsilon \OPT(G,k)$ and 
(ii) the min-cut in each connected component of $H_1$ is at least $\epsilon \opt(G,k)/k\ge \epsilon \opt(H_1,k)/k$. 
We note that if we can find a $(1+\epsilon)$-approximate optimum \minmax $k$-partition $(S_1, \ldots, S_k)$ for the unweighted instance $H_1$, then $\cost_G(S_1, \ldots, S_k)\le \cost_{H_1}(S_1, \ldots, S_k) + \epsilon \OPT(G,k)\le (1+2\epsilon)\opt(G,k)$. Hence, our goal now is to compute a $(1+\epsilon)$-approximation for the instance $(H_1,k)$ of \mmkc. We note that the components of $H_1$ are \emph{well-connected}: each component of $H_1$ has min-cut value at least $\epsilon \opt(H_1,k)/k$. With a run-time overhead of $2^{O(k)}$, we again assume that the instance $(H_1,k)$ is connected. 

The final step of our reduction is to subsample the edges of $H_1$: we pick each edge with probability $p=100\log n/(\epsilon^2\opt(H_1,2))$, where $n$ is the number of vertices in $H$. Let $E_2$ be the set of sampled edges and let $H_2:=(V,E_2)$. By Benczur-Karger \cite{benczur-karger}, we know that with probability at least $1-1/n^{26}$, (i) the scaled cut-value of every $2$-partition is preserved within a $(1\pm \epsilon)$-factor, i.e., $|\delta_{H_2}(S)|/p\in [(1-\epsilon)|\delta_{H_1}(S)|, (1+\epsilon)|\delta_{H_1}(S)|]$ for every $S\subseteq V$ and (ii) $|E_2|=O(\log n/\epsilon^2))|E_1|$. The preservation of cut values immediately implies that $\opt(H_2,k)/p\in [(1-\epsilon)\opt(H_1,k), (1+\epsilon)\opt(H_1,k)]$ and moreover, 
\begin{align*}
    \opt(H_2,k)
    &\le p(1+\epsilon)\opt(H_1,k)\\
    &= O(1)\frac{\log n}{\epsilon^2 \opt(H_1,2)}\cdot \opt(H_1,k)\\
    &=O(1)\frac{k\log n}{\epsilon^3 } \quad \quad \text{(since $\opt(H_1,2)\ge \frac{\epsilon\opt(H_1,k)}{k}$)}.
\end{align*}
Thus, we obtain an instance $H_2$ whose optimum cost is $O((k/\epsilon^3)\log n)$ and the optimum cost of the instance $(H_2,k)$ can be used to obtain a $(1+O(\epsilon))$-approximation to the optimum cost of the instance $(H_1,k)$. We note that all reduction steps can be implemented to run in time $2^{O(k)}n^{O(1)}\log(\opt(G_w,k))$. 
\end{proof}

For the sake of completeness, we now give the details of the knapsack PTAS-style rounding procedure to reduce the problem in a $(1+\epsilon)$-approximation preserving fashion to an unweighted instance. 
\begin{lemma}
\label{lem:knapsack-style-rounding}
There exists a polynomial-time algorithm that takes as input a graph $G=(V,E)$ with edge weights $w:E\rightarrow \R_+$, an $\epsilon\in (0, 1)$, and a value $\lambda \in [\opt(G_w,k), 2\opt(G_w,k)]$,  
and returns an unweighted multigraph $H=(V',E')$ 
such that 
$|E'|\le 2|E|^2/\epsilon$ and 
an $\alpha$-approximate \minmax $k$-partition in $H$ can be used to recover an $\alpha(1+\epsilon)$-approximate \minmax $k$-partition in $G_w$ in polynomial time. 
\end{lemma}
\begin{proof}
We may assume that $w(e)\le \lambda$ for every $e\in E$ (otherwise, contract the edge). Let $m:=|E|$, $\theta = \epsilon \lambda/m$ and $w'(e):=\lceil w(e)/\theta \rceil$ for every $e\in E$. Let $H=(V,E')$ be the graph obtained by creating $w'(e)$ copies of each edge $e\in E$. We first bound the number of edges in $H$. We have 
\begin{align*}
    |E'|
    &= \sum_{e\in E} \left\lceil \frac{w(e)}{\theta} \right\rceil
    \le \sum_{e\in E} \left(\frac{w(e)}{\theta} +1\right)
    = \frac{w(E)m}{\epsilon \lambda} + m
    \le \frac{m^2}{\epsilon} + m
    \le \frac{2m^2}{\epsilon}.
\end{align*}
The last but one inequality above is because $w(e)\le \lambda$ for every $e\in E$. 
Next, we show that $\theta \opt(H,k)\le (1+4\epsilon)\opt(G_w,k)$. Let $(S_1^*, \ldots, S_k^*)$ be an optimum \minmax $k$-partition in $G_w$. Then, for every $i\in [k]$, we have that 
\begin{align*}
    \theta w'(\delta(S_i^*))
    &= \sum_{e\in \delta(S_i^*)} \left\lceil \frac{w(e)}{\theta} \right \rceil \theta
    \le \sum_{e\in \delta(S_i^*)} \left( \frac{w(e)}{\theta}  + 1\right)\theta
    = w(\delta(S_i^*)) + \theta|\delta(S_i^*)|\\
    &\le \opt(G_w,k) + \frac{\epsilon \lambda |\delta(S_i^*)|}{m}
    \le \opt(G_w,k)(1+4\epsilon).
\end{align*}
Hence, $(S_1^*, \ldots, S_k^*)$ is a $k$-partition of the vertices of $H$ with $\cost_H(S_1^*, \ldots, S_k^*)\le \opt(G_w,k)(1+4\epsilon)$. 

Let $(S_1,\ldots, S_k)$ be an $\alpha$-approximate \minmax $k$-partition in $H$. For every $i\in [k]$, we have
\begin{align*}
    w(\delta(S_i))
    &= \sum_{e\in \delta(S_i)} w(e)
    = \sum_{e\in \delta(S_i)} \left(\frac{w(e)}{\theta w'(e)}\right) \theta w'(e)
    \le \sum_{e\in \delta(S_i)} \theta w'(e)\\
    &\le \alpha \opt(H,k) \theta
    \le \alpha(1+4\epsilon) \opt(G_w, k).
\end{align*}
Thus, $\cost_{G_w}(S_1, \ldots, S_k)\le \alpha(1+4\epsilon)\opt(G_w,k)$. 
\end{proof}

\section{Proof of Theorem \ref{thm:FPAS}}
\label{sec:FPAS}
In this section, we restate and prove Theorem \ref{thm:FPAS}. 
\thmFPAS*
\begin{proof}
Given the input, we compute a family $\mathcal{C}$ as described in Theorem \ref{thm:reduction-to-unweighted-instances-with-small-opt}. For each instance $H\in\mathcal{C}$, we use Theorem \ref{thm dp} with $\lambda = O((k/\epsilon^3)\log n)$ to find an optimum to the instance or obtain that the optimum is $\Omega((k/\epsilon^3)\log n)$. 
By conclusion (iii) of Theorem \ref{thm:reduction-to-unweighted-instances-with-small-opt}, we can then recover and output a $k$-partition $\mcp$ of $V$ such that $\cost_{G_w}(\mcp)=(1+O(\epsilon))\opt(G_w, k)$ in time $2^{O(k)}\poly(n,1/\epsilon)\log(\OPT(G_w,k))$. 

By Theorem \ref{thm:reduction-to-unweighted-instances-with-small-opt}, the time required to compute $\mathcal{C}$ is
\[2^{O(k)}\poly(n,1/\epsilon)\log(\OPT(G_w,k)).\]
For each $H\in\mathcal{C}$, the total time to run the algorithm described in Theorem \ref{thm dp} with $\lambda = O((k/\epsilon^3)\log n)$ is 
\begin{align*}
    \left(\left(\frac{k^2}{\epsilon^3}\right)\log{n}\right)^{O(k^2)}n^{O(1)}
    = \left(\frac{k}{\epsilon^3}\right)^{O(k^2)}(\log{n})^{O(k^2)}n^{O(1)}
    &= \left(\frac{k}{\epsilon^3}\right)^{O(k^2)}(k^{O(k^2)}+n)n^{O(1)}\\
    &= \left(\frac{k}{\epsilon}\right)^{O(k^2)}n^{O(1)}.
\end{align*}
Finally, our algorithm requires $2^{O(k)}\poly(n,1/\epsilon)\log(\OPT(G_w,k))$ time to recover $\mcp$. Therefore, the total run time is
\begin{align*}
    |\mathcal{C}|\cdot(k/\epsilon)^{O(k^2)}n^{O(1)}&+2^{O(k)}\poly(n,1/\epsilon)\log(\OPT(G_w,k))
    \\&=(k/\epsilon)^{O(k^2)}n^{O(1)}\cdot 2^{O(k)}\poly(n,1/\epsilon)\log(\OPT(G_w,k))
    \\&=(k/\epsilon)^{O(k^2)}n^{O(1)}\log(\OPT(G_w,k)).
\end{align*}
The run-time bound stated in the theorem follows since $\opt(G_w, k)\le \binom{n}{2}\max_{e\in E} w(e)$. 
\end{proof}


\section{NP-hardness}
\label{sec:hardness}
In this section, we restate and prove the hardness result. 

\thmHardness*
\begin{proof}
We will show a reduction from \textsc{$h$-clique} in the unweighted case. In  \textsc{$h$-clique}, the input consists of a simple graph $G=(V,E)$ and a positive integer $h$, and the goal is to decide whether there exists a subset $S$ of $V$ of size $|S|=h$ with $G[S]$ being a clique. The \textsc{$h$-clique} problem is NP-hard and W[1]-hard when parameterized by $h$.

Let $(G=(V,E),h)$ be an instance of \textsc{$h$-clique}, where $V=\{v_1,\ldots,v_n\}$. We may assume that $h\geq 2$. Let $M:=\max\{(n+1)^2,3m\}$ and $N:=Mn+2$, where $m$ is the number of edges in $G$. 
We construct a graph $G'=(V', E')$ as follows: for each vertex $v_i\in V$, we create a clique of size $N$ over the vertex set $C_i:=\{u_{i,j}:j\in[N]\}$. We also create a clique of size $N$ over the vertex set $W:=\{w_j:j\in[N]\}$. For each edge $e=v_i v_{i'}\in E$, we add the edge $u_{i,1} u_{i', 1}$ (between the first copy of vertex $v_i$ and the first copy of vertex $v_{i'}$). For each $v_i\in V$, we also add $M-\deg_G(v_i)$ edges between arbitrary pair of vertices in $C_i$ and $W$---for the sake of clarity, we fix these $M-\deg_G(v_i)$ edges to be $u_{i,j}w_j$ for every $j\in [M-\deg_G(v_i)]$. We set $V':=W\cup \cup_{i\in [n]}C_i$.
We note that the size of the graph $G'$ is polynomial in the size of the input graph $G$. See Figure \ref{fig:hardness} for an example. 
The next claim completes the reduction. 

\end{proof}

\begin{figure}[ht]
    \centering
    \includegraphics[width=0.5\textwidth]{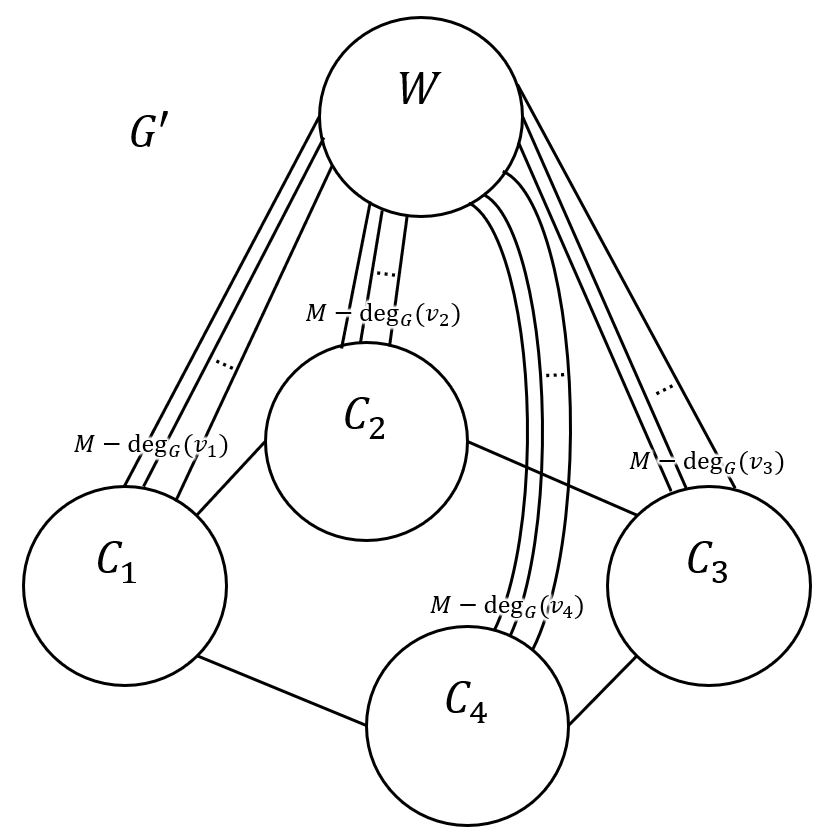}
    \caption{The graph $G'$ constructed in the proof of Theorem \ref{thm:hardness} when the input graph $G$ is a $4$-cycle.}
    \label{fig:hardness}
\end{figure}

\begin{claim}
The graph $G$ contains an $h$-clique if and only if there exists an $(h+1)$-partition $\mcp=(P_1,\ldots,P_{h+1})$ of $V'$ such that $\max_{i\in[h+1]} |\delta_{G'}(P_i)|\leq Mh-h(h-1)$.
\end{claim}

\begin{proof}
Suppose that $G$ contains an $h$-clique induced by a subset $V_0\subseteq V$. We may assume that $V_0=\{v_1,\ldots,v_h\}$ by relabelling the vertices of $V$. 
Consider the partition $\mcp=(P_1,\ldots,P_{h+1})$ given by
\begin{align*}
    P_i&:=C_i\ \forall i\in[h]\text{ and}
    \\P_{h+1}&:=V'\backslash(P_1\cup\ldots\cup P_h).
\end{align*}
We observe that 
\begin{align*}
    |\delta_{G'}(P_i)|&=(M-\deg_G(v_i))+\deg_G(v_i)=M\,\text{ for all } i\in[h]\text{ and}
    \\|\delta_{G'}(P_{h+1})|&=\sum_{v\in V_0}(M-\deg_G(v))+|E[V_0,V\backslash V_0]|
    \\&=\sum_{v\in V_0}(M-\deg_G(v))+\sum_{v\in V_0}\deg_G(v)-2|E[V_0]|
    \\&=Mh-2|E[V_0]|
    \\&=Mh-h(h-1).
\end{align*}
By choice of $M$, we know that $M\leq Mh-h(h-1)$, and hence $\max_{i\in[h+1]}|\delta_{G'}(P_i)|=Mh-h(h-1)$.

We now prove the converse. Suppose that we have an $(h+1)$-partition  $\mcp=(P_1,\ldots,P_{h+1})$ of $V'$ such that $\max_{i\in[h+1]}\{|\delta_{G'}(P_i)|\}\leq Mh-h(h-1)$.

We now show that $\mcp$ does not separate any two vertices in $C_i$ for all $i\in [n]$ or any two vertices in $W$. For the sake of contradiction, suppose that there exists vertices $u,v\in V'$ such that $u$ and $v$ are in the same set $A\in\{C_1,\ldots,C_n,W\}$ but they are in different parts of $\mcp$. Without loss of generality, let $u\in P_1$ and $v\in P_2$. Then $A\cap P_1$ and $A\backslash P_1$ together forms a non-trivial 2-cut of $G'[A]$, and hence $|\delta_{G'}(P_1)|\geq|\delta_{G'[A]}(A\cap P_1)|\geq N-1$. By our choice of $N$, we know that $N-1>Mh-h(h-1)$. This contradicts the fact that $|\delta_{G'}(P_i)|\le Mh-h(h-1)$ for all $i\in [h+1]$.

From now on, let us fix $P_{h+1}$ to be the part of $\mcp$ that contains $W$. We will show that $P_{h+1}$ contains exactly $n-h$ sets among $C_1,\ldots,C_n$. Since all parts of $\mcp$ are non-empty, it follows that $P_{h+1}$ cannot contain more than $n-h$ sets among $C_1,\ldots,C_n$. For the sake of contradiction, suppose that $P_{h+1}$ contains at most $n-h-1$ sets among $C_1,\ldots,C_n$. This implies that more than $h$ sets among $C_1,\ldots,C_n$ are not contained in $P_{h+1}$. Let $C_{i_1},\ldots,C_{i_{h'}}$ be the sets outside $P_{h+1}$, where $h'>h$. Then

\begin{align*}
    |\delta_{G'}(P_{h+1})|&\geq\sum_{\ell\in[h']}|E'[C_{i_\ell},W]|
    \\&=\sum_{\ell\in[h']}(M-\deg_G(v_{i_\ell}))
    \\&\geq\sum_{\ell\in[h']}(M-n)
    \\&\geq(h+1)(M-n)
    \\&>Mh-h(h-1).
\end{align*}
This contradicts the fact that $|\delta_{G'}(P_i)|\le Mh-h(h-1)$ for all $i\in [h+1]$. Hence, $P_{h+1}$ contains exactly $n-h$ sets among $C_1,\ldots,C_n$. 

Let $C_{i_1},\ldots,C_{i_h}$ be the sets that are not in $P_{h+1}$. We will now show that $S:=\{v_{i_1},\ldots,v_{i_h}\}$ induces a clique in $G$.
Since $\max_{i\in[h+1]}|\delta_{G'}(P_i)|\leq Mh-h(h-1)$, we know that
\begin{align*}
    Mh-h(h-1)&\geq |\delta_{G'}(P_{h+1})|
    \\&=\sum_{\ell\in[h]}(M-\deg_G(v_{i_\ell}))+|E'[\cup_{\ell\in[h]}C_{i_{\ell}},(\cup_{i\in[n]}C_i)\backslash(\cup_{\ell\in[h]}C_{i_{\ell}})]|
    \\&=\sum_{\ell\in[h]}(M-\deg_G(v_{i_\ell}))+|E[S,V\backslash S]|
    \\&=Mh-\sum_{\ell\in[h]}\deg_G(v_{i_\ell})+|E[S,V\backslash S]|
    \\&=Mh-2|E[S]|.
\end{align*}
Consequently, $|E[S]|\geq h(h-1)/2$, and thus $S$ induces an $h$-clique in $G$.
\end{proof}

We remark that our hardness reduction also implies that an exact algorithm for \mmkc in $n^{o(k)}$ time in simple graphs (i.e., unweighted graphs) would also imply an $n^{o(h)}$-time algorithm for \textsc{$h$-clique}, thus refuting the exponential time hypothesis (see Theorem 14.21 in \cite{FPT-book}).

\section{Conclusion}\label{sec:conclusion}

In this work, we addressed the graph $k$-partitioning problem under the minmax objective. We showed that it is NP-hard, W[1]-hard when parameterized by $k$, and admits a parameterized approximation scheme when parameterized by $k$. Our algorithmic ideas generalize in a natural manner to also lead to a parameterized approximation scheme for \mlpnormkc for every $p\ge 1$: in \mlpnormkc, the input is a graph $G=(V,E)$ with edge weights $w:E\rightarrow \R_+$, and the goal is to partition the vertices into $k$ non-empty parts so as to minimize $(\sum_{i=1}^k w(\delta(V_i))^p )^{1/p}$. We note that \mlpnormkc generalizes \mskc as well as \mmkc. 

Based on prior works in approximation literature for minmax and minsum objectives, it is a commonly held belief that the minmax objective is harder to approximate than the minsum objective. Our results suggest that for the graph $k$-partitioning problem, the complexity/approximability of the two objectives are perhaps the same. 
A relevant question towards understanding if the two objectives exhibit a complexity/approximability gap is the following: 
When $k$ is part of input, is \mmkc constant-approximable? 
We recall that when $k$ is part of input \mskc does not admit a $(2-\epsilon)$-approximation for any constant $\epsilon>0$ under the Small Set Expansion Hypothesis \cite{Ma18} and admits a $2$-approximation \cite{SV95}. The best approximation factor that we know currently for \mmkc is $2k$ (see Section \ref{sec:results}). 

The $2$-approximation for \mskc is based on solving the same problem in the \emph{Gomory-Hu} tree of the given graph. We note that solving \mmkc on the Gomory-Hu tree of the given graph could at best result in an $O(n)$-approximation: consider the complete graph $K_{2n+1}$ on $2n+1$ vertices with unit edge weights. The optimum partition for \mmkc  for $k=n$ is the partition that contains $n-1$ parts each containing $2$ vertices and one part containing $3$ vertices leading to an optimum value of $3(2n-2)=\Theta(n)$. The star graph on $2n+1$ vertices with all edge weights being $2n$ is a Gomory-Hu tree for $K_{2n+1}$. The optimum partition for \mmkc for $k=n$ on the  Gomory-Hu tree (i.e., the star graph with weighted edges) consists of $n-1$ parts corresponding to $n-1$ leaves of the star graph and one part containing the remaining leaves and the center, thus leading to an optimum value of $\Theta(n^2)$.

\bibliographystyle{amsplain}
\bibliography{main}

\providecommand{\bysame}{\leavevmode\hbox to3em{\hrulefill}\thinspace}
\providecommand{\MR}{\relax\ifhmode\unskip\space\fi MR }
\providecommand{\MRhref}[2]{%
  \href{http://www.ams.org/mathscinet-getitem?mr=#1}{#2}
}
\providecommand{\href}[2]{#2}
\begin{thebibliography}{10}

\bibitem{AKS19}
S.~Ahmadi, S.~Khuller, and B.~Saha, \emph{Min-max correlation clustering via
  multicut}, Integer Programming and Combinatorial Optimization, IPCO, 2019,
  pp.~13--26.

\bibitem{BFKMNNS14}
N.~Bansal, U.~Feige, R.~Krauthgamer, K.~Makarychev, V.~Nagarajan, J.~Naor, and
  R.~Schwartz, \emph{Min-max graph partitioning and small set expansion}, SIAM
  Journal on Computing \textbf{43} (2014), no.~2, 872--904.

\bibitem{benczur-karger}
A.~Bencz\'{u}r and D.~Karger, \emph{Randomized approximation schemes for cuts
  and flows in capacitated graphs}, SIAM Journal on Computing \textbf{44}
  (2015), no.~2, 290--319.

\bibitem{BCKM20}
K.~B{\'e}rczi, K.~Chandrasekaran, T.~Kir{\'a}ly, and V.~Madan, \emph{{Improving
  the Integrality Gap for Multiway Cut}}, Mathematical Programming (2020).

\bibitem{CC21}
K.~Chandrasekaran and C.~Chekuri, \emph{{Min-max partitioning of Hypergraphs
  and Symmetric Submodular Functions}}, Proceedings of the 32nd ACM-SIAM
  Symposium on Discrete Algorithms (to appear), SODA, 2021.

\bibitem{CGS17}
M.~Charikar, N.~Gupta, and R.~Schwartz, \emph{Local guarantees in graph cuts
  and clustering}, Integer Programming and Combinatorial Optimization, IPCO,
  2017, pp.~136--147.

\bibitem{CQX19}
C.~Chekuri, K.~Quanrud, and C.~Xu, \emph{{LP} relaxation and tree packing for
  minimum $k$-cuts}, 2nd Symposium on Simplicity in Algorithms, SOSA, 2019,
  pp.~7:1--7:18.

\bibitem{FPT-book}
M.~Cygan, F.~Fomin, L.~Kowalik, D.~Lokshtanov, D.~Marx, M.~Pilipczuk,
  M.~Pilipczuk, and S.~Saurabh, \emph{Parameterized algorithms}, Springer,
  2015.

\bibitem{CKLPPSW18}
M.~Cygan, P.~Komosa, D.~Lokshtanov, M.~Pilipczuk, M.~Pilipczuk, S.~Saurabh, and
  M.~Wahlstrom, \emph{{Randomized contractions meet lean decompositions}},
  arXiv: https://arxiv.org/abs/1810.06864, 2018.

\bibitem{DJPSY94}
E.~Dahlhaus, D.~Johnson, C.~Papadimitriou, P.~Seymour, and M.~Yannakakis,
  \emph{The complexity of multiterminal cuts}, SIAM Journal on Computing
  \textbf{23} (1994), no.~4, 864--894.

\bibitem{DEFPR03}
R.~Downey, V.~Estivill-Castro, M.~Fellows, E.~Prieto, and F.~Rosamund,
  \emph{Cutting up is hard to do: The parameterised complexity of k-cut and
  related problems}, Electronic Notes in Theoretical Computer Science
  \textbf{78} (2003), 209--222.

\bibitem{GH88}
O.~Goldschmidt and D.~Hochbaum, \emph{Polynomial algorithm for the k-cut
  problem}, Proceedings of the 29th Annual Symposium on Foundations of Computer
  Science, FOCS, 1988, pp.~444--451.

\bibitem{GH94}
\bysame, \emph{{A Polynomial Algorithm for the $k$-cut Problem for Fixed $k$}},
  Mathematics of Operations Research \textbf{19} (1994), no.~1, 24--37.

\bibitem{GHLL20}
A.~Gupta, D.~Harris, E.~Lee, and J.~Li, \emph{{Optimal Bounds for the $k$-cut
  Problem}}, arXiv: https://arxiv.org/abs/2005.08301, 2020.

\bibitem{GLL18-SODA}
A.~Gupta, E.~Lee, and J.~Li, \emph{{An FPT Algorithm Beating 2-Approximation
  for $k$-Cut}}, Proceedings of the 29th Annual ACM-SIAM Symposium on Discrete
  Algorithms, SODA, 2018, pp.~2821--2837.

\bibitem{GLL18-FOCS}
\bysame, \emph{Faster exact and approximate algorithms for k-cut}, Proceedings
  of the 59th IEEE annual Symposium on Foundations of Computer Science, FOCS,
  2018, pp.~113--123.

\bibitem{GLL20-STOC}
\bysame, \emph{{The Karger-Stein Algorithm is Optimal for $k$-Cut}},
  Proceedings of the 52nd Annual ACM SIGACT Symposium on Theory of Computing,
  STOC, 2020, p.~473–484.

\bibitem{KMZ19}
S.~Kalhan, K.~Makarychev, and T.~Zhou, \emph{Correlation clustering with local
  objectives}, Advances in Neural Information Processing Systems 32, 2019,
  pp.~9346--9355.

\bibitem{KS96}
D.~Karger and C.~Stein, \emph{A new approach to the minimum cut problem},
  Journal of the ACM \textbf{43} (1996), no.~4, 601--640.

\bibitem{KL20}
K-i Kawarabayashi and B.~Lin, \emph{{A nearly 5/3-approximation FPT Algorithm
  for Min-$k$-Cut}}, Proceedings of the 31st ACM-SIAM Symposium on Discrete
  Algorithms, SODA, 2020, pp.~990--999.

\bibitem{La73}
E.~Lawler, \emph{{Cutsets and Partitions of Hypergraphs}}, {Networks}
  \textbf{3} (1973), 275--285.

\bibitem{Li19}
J.~Li, \emph{Faster minimum k-cut of a simple graph}, Proceedings of the 60th
  Annual Symposium on Foundations of Computer Science, FOCS, 2019,
  pp.~1056--1077.

\bibitem{LSS}
D.~Lokshtanov, S.~Saurabh, and V.~Surianarayanan, \emph{{A Parameterized
  Approximation Scheme for \textsc{Min $k$-Cut}}}, Proceedings of the 61st IEEE
  annual Symposium on Foundations of Computer Science (to appear), FOCS, 2020.

\bibitem{Ma18}
P.~Manurangsi, \emph{{Inapproximability of Maximum Biclique Problems, Minimum
  $k$-Cut and Densest At-Least-$k$-Subgraph from the Small Set Expansion
  Hypothesis}}, Algorithms \textbf{11(1)} (2018), 10.

\bibitem{RS08}
R.~Ravi and A.~Sinha, \emph{Approximating k-cuts using network strength as a
  lagrangean relaxation}, European Journal of Operational Research \textbf{186}
  (2008), no.~1, 77 -- 90.

\bibitem{SV95}
H.~Saran and V.~Vazirani, \emph{{Finding k Cuts within Twice the Optimal}},
  SIAM Journal on Computing \textbf{24} (1995), no.~1, 101--108.

\bibitem{SV14}
A.~Sharma and J.~Vondr\'{a}k, \emph{Multiway cut, pairwise realizable
  distributions, and descending thresholds}, Proceedings of the forty-sixth
  annual ACM Symposium on Theory of Computing, STOC, 2014, pp.~724--733.

\bibitem{svitkina-tardos}
Z.~Svitkina and {\'E}.~Tardos, \emph{Min-max multiway cut}, Approximation,
  Randomization, and Combinatorial Optimization. Algorithms and Techniques,
  APPROX, 2004, pp.~207--218.

\bibitem{Th08}
M.~Thorup, \emph{{Minimum $k$-way Cuts via Deterministic Greedy Tree Packing}},
  Proceedings of the 40th Annual ACM Symposium on Theory of Computing, STOC,
  2008, pp.~159--166.

\bibitem{ZNI05}
L.~Zhao, H.~Nagamochi, and T.~Ibaraki, \emph{{Greedy splitting algorithms for
  approximating multiway partition problems}}, Mathematical Programming
  \textbf{102} (2005), no.~1, 167--183.

\end{thebibliography}

\end{document}